\documentclass{article}
\usepackage{authblk}
\usepackage[utf8]{inputenc}
\usepackage[T1]{fontenc}
\usepackage[a4paper]{geometry}
\usepackage{graphicx}
\usepackage[textwidth=8em,textsize=small]{todonotes}
\usepackage{amsmath,amsthm,amssymb}
\usepackage{natbib}
\usepackage{hyperref}
\usepackage[nameinlink]{cleveref}
\usepackage{subcaption}
\usepackage{bbm}
\usepackage{amsfonts}
\usepackage{mathtools}
\usepackage[titletoc,toc]{appendix}
\usepackage{url}
\usepackage{algpseudocode}
\usepackage{multirow}
\usepackage{booktabs}
\usepackage[normalem]{ulem}
\usepackage{enumitem}
\usepackage{array}
\usepackage{setspace}
\usepackage{orcidlink}
\graphicspath{{./figs/}}
\usepackage[linesnumbered,lined,commentsnumbered,ruled]{algorithm2e}
\SetKwInput{kwInit}{Init}
\usepackage{threeparttable}
\usepackage{pifont}

\definecolor{darkblue}{HTML}{0c7dbb}
\definecolor{darkgreen}{rgb}{0,0.48,0.65}
\hypersetup{
	colorlinks = true,
	citecolor=darkgreen,
	filecolor=black,
	linkcolor=darkblue,
	urlcolor=blue
}
\providecommand{\keywords}[1]
{
	\small
	\textbf{\textit{Keywords---}} #1
}

\DeclareMathOperator*{\argmin}{arg\,min}
\DeclareMathOperator*{\argmax}{arg\,max}
\newtheorem{theorem}{Theorem}[section]

\newtheorem{proposition}{Proposition}[section]

\newtheorem{lemma}[theorem]{Lemma}   
\newtheorem{assumption}{Assumption}[section]  
\Crefname{algocf}{Algorithm}{Algorithms}
\title{Modelling Loss of Complexity in Intermittent Time Series and its Application}
\author[1,3]{Jie Li\orcidlink{0000-0001-8353-1322}\footnote{Addresses for correspondence: Jie Li, School of Mathematics, Statistics and Actuarial Science, University of Kent, Canterbury CT2 7NF.\ \textbf{Email}:~\href{jl725@kent.ac.uk}{jl725@kent.ac.uk}}}
\author[1]{Jian Zhang}
\author[2]{Samantha L. Winter}
\author[2]{Mark Burnley}
\affil[1]{School of Mathematics, Statistics and Actuarial Science, University of Kent, Canterbury, CT2 7NF, United Kingdom}
\affil[2]{School of Sport, Exercise and Health Sciences, Loughborough University, Epinal Way, LE11 3TU, Leicestershire, United Kingdom}
\affil[3]{Inovision-IP Ltd, Culham Innovation Centre, D5 Culham Science Centre,Abingdon, Oxfordshire, OX14 3DB, UK}

\begin{document}
\maketitle
\begin{abstract}
    In this paper, we developed a novel method of nonparametric relative entropy (RlEn) for modelling loss of complexity in intermittent time series. The method consists of two steps. We first fit a nonlinear autoregressive model to each intermittent time series, where the corresponding lag order and the loss of complexity are determined by Bayesian Information Criterion (BIC) and relative entropy respectively. Then, change-points in the complexity are detected by a cumulative sum (CUSUM) based statistic. Compared to approximate entropy (ApEn), a popular method in literature, the performance of RlEn was assessed by simulations in terms of (1) ability to localize complexity change-points in intermittent time series; (2) ability to faithfully estimate underlying nonlinear models.  The performance of the proposal was then examined in a real analysis of fatigue-induced changes in the complexity of human motor outputs. The results showed that the proposed method outperformed the ApEn in accurately detecting changes of complexity in intermittent time series segments.

\end{abstract}

\section{Introduction}\label{sec:introduction}
Intermittent time series are frequently met in various fields, such as neurological examination~\citep{buriokaApproximateEntropyElectroencephalogram2005}, heart rate analysis~\citep{acharyauHeartRateAnalysis2004},  sports science~\citep{forrestEffectSignalAcquisition2014} and energy demand~\citep{LIHeatingdemandbehavior2024}.  The intermittent time series contains information on various patterns or models. For instance, in Magnetoencephalography (MEG) or Electroencephalogram (EEG) scan, neurologically healthy subjects respond differently in terms of EEG and MEG signals for different stimuli (faces v.s.\ scrambled faces, or familiar faces v.s.\ unfamiliar faces, see~\cite{wakemanMultisubjectMultimodalHuman2015}); In cardiology, the pattern of heart rate varies along the states of human: sleeping, sitting, walking, jogging or running, see~\citet{acharyau.NonlinearAnalysisEEG2005},~\citet{buriokaApproximateEntropyElectroencephalogram2005},~\citet{shiEntropyAnalysisShortTerm2017} and among others; In sports science, muscle force and power is maintained by the aerobic resynthesis of adenosine triphosphate (ATP).  When force or power exceeds a critical value, exercise is not sustainable and fatigue progressively develops.  Eventually, the task set (e.g., to produce 50\% of maximal force) failed, as the force decreases despite maximal effort~\citep{Burnley2012}.~\citet{pethickFatigueReducesComplexity2015} demonstrated that this fatigue was also associated with a loss of muscle force or torque complexity as the contractions progressed.  This observation implies that the fatigue process can be tracked during the contractions.  These changes in complexity can, in principle, be used to detect change-points in the patterns of force development during fatigue development. Scientists have a great interest in detecting the change-points among a consecutive intermittent time series to help scientists identify the patterns of brain activity, heart disease, or improve the performance of athletes.

Throughout this research, the terminology \textit{change-point} refers to the ``change-point'' among intermittent time series rather than the ``change-point'' within a specific time series. The usual change-points are the breakpoints within one piece of time series, see~\citet{killickOptimalDetectionChangepoints2012,piotrwbs2014,Jie2024autocpd}. However, the change-points in this article are the breakpoints among the intermittent time series. For example, suppose we have extracted 55 intermittent time series from~\Cref{fig:allquad.png}. Sports scientists want to know the time of muscle fatigue occurrence. In this case, the time of muscle fatigue occurrence is one type of change-point among intermittent time series.

To identify the change-points among a consecutive intermittent time series using cumulative summation (CUSUM) based method~\citep{pageContinuousInspectionSchemes1954,killickOptimalDetectionChangepoints2012,piotrwbs2014,wang2018inspect}, it is necessary to map each intermittent time series to a scalar. Subsequently, a CUSUM-based method can be applied to the resulting series of scalars.  Let \( \mathbbm{x}_{t}=(x_{t1},\ldots,x_{tN}) \) be a univariate intermittent time series with length \( N \) where \( t \) represents the time. \( T \) is the number of intermittent isometric time series. Then \( \mathbbm{x}_{1},\mathbbm{x}_{2},\ldots,\mathbbm{x}_{T} \) represent a consecutive intermittent isometric time series. Denote the map function as \( \mathcal{I}:\mathbb{R}^{N}\mapsto \mathbb{R} \). Mathematically, the detection of change-points among intermittent isometric time series is equivalent to the detection of change-points of \( \{\mathcal{I}(\mathbbm{x}_{1}),\mathcal{I}(\mathbbm{x}_{2}),\ldots,\mathcal{I}(\mathbbm{x}_{T})\} \). In particular, if \( N=1 \) and \(\mathcal{I}(\cdot)\) is an identity function, the problem of detecting change-point among intermittent time series degenerates to the classical change-point detection problem.

Intuitively, the map function \(\mathcal{I}(\cdot)\) should be \textit{transformation invariant} and \textit{background noise-free}. The transformation invariant property means that \(\mathcal{I}(\mathbbm{x}_{t})\) equals \(\mathcal{I}(h(\mathbbm{x}_{t}))\), where \(h(\mathbbm{x}_{t})\) is a sufficient statistic of \(\mathbbm{x}_{t}\). Let the noise of \( \mathbbm{x}_{t}, t=1,2,\ldots,T\) follow normal distribution with zero mean and variance \(\sigma^{2}\). The background noise-free property means that the result of map function \(\mathcal{I}(\cdot)\) should be free from the background noise level  \(\sigma^{2}\). The former property could eliminate the effects of the measurement unit, while the latter property guarantees that \( \mathcal{I}(\cdot) \) does not include the variance of background noise.
\Cref{tab:Potential Choice of I} summarizes five potential choices of \( \mathcal{I}(\cdot) \): mean, variance, entropy (En), conditional entropy (CoEn) and relative entropy (RlEn) and their properties.
\begin{table}[htbp]
    \centering
    \begin{threeparttable}
        \caption{Potential choices of \( \mathcal{I}(\cdot) \) and their properties. En, CoEn and RlEn represent entropy, conditional entropy and relative entropy, respectively.}\label{tab:Potential Choice of I}
        \begin{tabular}{cccccc}
            \hline
            & Mean & Variance & En & CoEn & RlEn \\ \hline
            Transformation invariant & \ding{56} & \ding{56} & \ding{56} & \ding{56} & \ding{52} \\
            Background noise-free & \ding{52} & \ding{56} & \ding{56} & \ding{56} & \ding{52} \\ \hline
        \end{tabular}
    \end{threeparttable}
\end{table}

\noindent
\textit{1. Mean and Variance}. Suppose \( \mathcal{I}(\cdot) \) represents mean function, \( \mathbbm{y}_{t}= h(\mathbbm{x}_{t}) = \alpha \mathbbm{x}_{t}\) is a linear transformation, \( \alpha\neq 0 \). Then \( \bar{\mathbbm{y}}_{t} \) = \( \alpha\bar{\mathbbm{x}}_{t}\neq \bar{\mathbbm{x}}_{t} \). If \( \mathbbm{x}_{t} \) is independent of \( \sigma^2 \), then for any transformation \( h(\cdot) \), \( h(\mathbbm{x}_{t}) \) is also independent of \( \sigma^2 \). Similarly, one can easily verify that variance does not have these two properties.

\noindent
\textit{2. Entropy (En) and Conditional Entropy (CoEn)}. Entropy and conditional entropy are inappropriate  choice of \( \mathcal{I}(\cdot) \) because:
\begin{itemize}
    \item Entropy is scale variant, for example, let \( \text{En}(x) \) represent the entropy of variable \( x \), for any scale transformation \( y=\alpha x \), \( \alpha\in \mathbb{R} \) and \( \alpha\neq 0 \) then the entropy of variable \( y \) is \( \text{En}(x)+\log \left\vert \alpha \right\vert \). More generally, entropy is not transformation invariant under change of variable as well, see\textcolor{darkblue}{~\citet[][p.18]{iharaInformationTheoryContinuous1993}}.
    \item Conditional entropy is neither transformation invariant. In nonparametric settings, the four entropies: \textit{Approximate Entropy} (ApEn)~\citep{pincusApproximateEntropyMeasure1991}, \textit{Sample Entropy} (SpEn)~\citep{richmanPhysiologicalTimeseriesAnalysis2000}, \textit{Multi-scale Entropy} (MsEn)~\citep{costaMultiscaleEntropyAnalysis2003} and \textit{Fuzzy Entropy} (FzEn)~\citep{chenMeasuringComplexityUsing2009} are the special cases of conditional entropy. They do not own the property of transformation invariant. For example, when using the multivariate uniform kernel to estimate the nonparametric CoEn, the difference between ApEn and CoEn is \( \log(2h)\) which comes from the scale transformation in the kernel function with bandwidth \(h\).
    \item In addition, entropy and conditional entropy are not background noise-free. A direct example can be found in~\Cref{sec:relative_entropy_for_stationary_ar2_model}. Equations~\eqref{eq:entropy_joint} and~\eqref{eq:conditional_entropy_final} are entropy and conditional entropy respectively, however, both are related to the \( \sigma^2 \).
\end{itemize}

\noindent
\textit{3. Relative Entropy (RlEn)}.~\citet{kullbackInformationSufficiency1951} have proved that RlEn has the property of \textit{transformation invariant}. The discussion of \textit{background noise-free} property is put off in~\textcolor{darkblue}{Proposition}~\ref{thm:relative_entropy} in Appendix.

In this article, we will use the relative entropy (RlEn) as the map function for \( \mathbbm{x}_{t} \). Relative entropy is also called Kullback-Leibler divergence~\citep{kullbackInformationSufficiency1951}. It is a measure to describe the divergence between two probability distributions.~\citet{robinsonConsistentNonparametricEntropyBased1991} proposed a consistent nonparametric entropy-based test for independence in time series which employed the sample splitting device. To avoid the hyperparameter for the sample splitting device,~\citet{hongAsymptoticDistributionTheory2005} developed an \textit{asymptotic distribution theory for nonparametric entropy of serial dependence}. Unlike~\citet{hongAsymptoticDistributionTheory2005} had obtained an asymptotic distribution of relative entropy for pairwise variables, we proposed the relative entropy for \(m\)-consecutive variables in high-dimensional context where \(m\geq 2\). Furthermore, we recommend using the BIC criterion to select the pre-determined parameter \( m \). The consistency theory of BIC is developed to ensure that the estimator of lag order converges to the true order with probability 1. Under certain assumptions, the limiting distribution of nonparametric RlEn is Gaussian with convergence rate \( \sqrt{n}h^{(m+1)/2} \).

This article is organized as follows. In~\Cref{sec:methodology}, we introduce the methodology of relative entropy for \(m\)-consecutive variables, the BIC criterion to select the optimal lag order as well as the algorithms. In~\Cref{sec:theory}, we derive the asymptotic distribution of relative entropy. In~\Cref{sec:numerical_study}, we provide simulation studies to illustrate the performance of the proposed method, and we apply the proposed method to the real data set in~\Cref{sec:real_data_analysis}. Finally, we conclude the article in~\Cref{sec:discussion}.


\section{Methodology}\label{sec:methodology}
Let \( X_{1},\ldots,X_{N} \) represent the time-varying scalar measurements, which form a strictly stationary process. Denote \( \mathbf{X}_{i;m} ={(X_{i},\ldots,X_{i+m-1} )}^{\top}\) as the \( m \) consecutive variables vector where \( m \) could be sufficiently large but be bounded by \( M \). The density function of  \( \mathbf{X}_{i;m} \) is defined as \( g(\mathbf{X}_{i;m}) \). Furthermore, let \( \mathbf{X}_{i;m+1} ={(X_{i},\ldots,X_{i+m} )}^{\top}\) and \( f(\mathbf{X}_{i;m+1}) \) be the \( m+1 \) consecutive variables vector and its probability density function.
Note that given the first vector \( \mathbf{X}_{i;m} \),  the conditional probability density function of \(X_{i+m}\) is \( f(X_{i+m}\mid\mathbf{X}_{i;m})=f(\mathbf{X}_{i;m+1}) / g(\mathbf{X}_{i;m})\). Furthermore, let \( g_{1}(X_{i+m}) \) be the density function of \( X_{i+m} \), then the relative entropy is defined as
\begin{equation*}
    \mathfrak{E}  \coloneqq \mathbb{E}_{f} \left[ \log\left( \frac{f\left(\mathbf{X}_{i;m+1}\right)}{g\left(\mathbf{X}_{i;m}\right)g_{1}(X_{i+m})} \right) \right],
\end{equation*}
where \(\mathbb{E}_{f}\) represent the expectation with respect to the density function \(f\). In mathematical statistics, \(\mathfrak{E}\) is also called Kullback–Leibler (KL) divergence~\citep{kullbackInformationSufficiency1951}.
Estimation of \(\mathfrak{E}\) can be divided into two parts: density estimation and expectation estimation.

For density estimation, we use a nonparametric kernel method to estimate \( f(\mathbf{X}_{i;m+1}) \), \( g(\mathbf{X}_{i;m})\) and \(g_{1}(X_{i+m}) \). Jackknife kernel has been proposed to mitigate the boundaries effect for the kernel with bounded support~\citep{johnBoundaryModificationKernel1984,hardleAppliedNonparametricRegression1990,jonesSimpleBoundaryCorrection1993}. Additionally, other popular methods in the literature address boundary effects in kernel density estimation, such as the reflection method~\citep{Schechner1985}, the transformation method~\citep{marron1994transformations}, and local polynomial density estimation~\citep{chen1999beta}. Researchers can choose to use these alternative methods but must adapt the theory presented in~\Cref{sec:theory}. For consistency, we employ the Jackknife kernel to estimate the densities in this section.

Following~\citet{johnBoundaryModificationKernel1984}, we suppose~\textcolor{darkblue}{Assumption}~\ref{assump:kernel_assumption} holds throughout this article.
\begin{assumption}\label{assump:kernel_assumption}
    The domain of kernel function \( K(\cdot) \) is \( [-1, 1] \) and \(
    K(\cdot) \) satisfies \(
    K(-x)=K(x) \) for any \( x\in [-1, 1] \), \( \int_{-1}^{+1}
    K(x)\,\mathrm{d}x =1 \) and \(
    \int_{-1}^{+1} x^{2} K(x)\,\mathrm{d}x < + \infty \).
\end{assumption}
In fact, the Jackknife kernel is a linear combination of two different self-normalized kernels, namely,
\begin{equation*}
    k_{\rho}(u)=(1+\beta(\rho))\frac{K(u)}{\omega_{0}(\rho)}-\frac{\beta(\rho)}{\alpha}\frac{K(u/\alpha)}{\omega_{0}(\rho/\alpha)},
\end{equation*}
where \( \omega_{l}(\rho)=\int_{-1}^{\rho} u^{l}K(u)\,\mathrm{d}u \), \( l=0,1,2 \), \( 0\leq \rho\leq 1 \) and \( \beta(\rho)=\frac{R_{1}(\rho)}{\alpha R_{1}(\rho/\alpha)-R_{1}(\rho)} \),
where \( R_{l}(\rho)= \omega_{l}(\rho)/\omega_{0}(\rho)\), \( l=1,2 \). In this article, we follow the choice of  \( \alpha \) in~\citet{johnBoundaryModificationKernel1984} and let \( \alpha=2-\rho \). Finally, for univariate \( x,y \in[0,1]\), the Jackknife kernel is defined as
\begin{equation*}
    K_{h}^{J}(x-y)\coloneqq\begin{cases}
        h^{-1}k_{(x/h)}\left( \frac{x-y}{h} \right),     & \text{if}\, x\in[0,h).   \\
        h^{-1}K\left( \frac{x-y}{h} \right),             & \text{if}\, x\in[h,1-h]. \\
        h^{-1}k_{[(1-x)/h]}\left( \frac{x-y}{h} \right), & \text{if}\, x\in(1-h,1].
    \end{cases}
\end{equation*}
For more details of Jackknife kernel, see~\citet{johnBoundaryModificationKernel1984} and~\citet{hongAsymptoticDistributionTheory2005}. For vectors \( \mathbf{x} = (x_{1},\ldots,x_{m}) \) and \( \mathbf{y} = (y_{1},\ldots,y_{m}) \), we define the scaled multivariate Jackknife kernel as
\begin{equation}\label{eq:multivariate_kernel}
    \mathcal{K}^{(m)}_{h}(\textbf{x}-\textbf{y})\coloneqq K_{h}^{J}(x_{1}-y_{1})\times K_{h}^{J}(x_{2}-y_{2})\times\cdots\times K_{h}^{J}(x_{m}-y_{m}).
\end{equation}
Following the assumption in~\citet{hongAsymptoticDistributionTheory2005}, we assume the bandwidths for each component of \( x_{1},\ldots,x_{m} \) in~\eqref{eq:multivariate_kernel} are same, i.e., \(h\).

Let \( x_{1},\ldots,x_{N} \) be the observations of \( X_{1},\ldots,X_{N} \), \( \mathbf{x}_{i;m}\) and \( \mathbf{x}_{i;m+1}\), \( i=1,\ldots,N-m\) are the corresponding observations of \(\mathbf{X}_{i;m}\) and \(\mathbf{X}_{i;m+1}\), respectively. From now on, we define \(n\coloneqq N-m\) to simplify the notations in this paper. The ``leave-one-out'' kernel density estimators of \(f(\cdot),g(\cdot)\) and \(g_{1}(\cdot)\) are:
\begin{align*}
    \hat{f}\left( \mathbf{x}_{i;m+1} \right) & = \frac{1}{n-1}\sum\nolimits_{j=1}^{n}\mathcal{K}^{(m+1)}_{h}\left( \mathbf{x}_{i;m+1}-\mathbf{x}_{j;m+1} \right)\mathbbm{1}(j\neq i), \\
    \hat{g}\left( \mathbf{x}_{i;m} \right) & = \frac{1}{n-1}\sum\nolimits_{j=1}^{n}\mathcal{K}^{(m)}_{h}\left( \mathbf{x}_{i;m}-\mathbf{x}_{j;m} \right)\mathbbm{1}(j\neq i), \\
    \hat{g}_{1}(x_{i+m}) & = \frac{1}{n-1}\sum\nolimits_{j=1}^{n}K_{h}^{J}(x_{i+m}-x_{j+m})\mathbbm{1}(j\neq i),
\end{align*}
where \(i,j\in\{1,2,\ldots,n\}\) and \(\mathbbm{1}(\cdot)\) is the indicator function. Then, the nonparametric estimator of \(\mathfrak{E}\) can be expressed as
\begin{equation}\label{eq:estimator_rlen}
    \hat{\mathfrak{E}}(m,h|\mathfrak{X})=\frac{1}{n}\sum_{i\in S(m)}\log\frac{\hat{f}\left( \mathbf{x}_{i;m+1} \right)}{\hat{g}\left( \mathbf{x}_{i;m} \right)\hat{g}_{1}(x_{i+m})},
\end{equation}
where \( S(m)=\{i\in \mathbb{N}: 1\leq i\leq n, \hat{f}( \mathbf{x}_{i;m+1} )>0,\hat{g}( \mathbf{x}_{i;m} )>0,\hat{g}_{1}( x_{i+m} )>0\} \) and \(\mathfrak{X} = \{x_{1},x_{2},\ldots,x_{N}\}\). We observe that the estimator \(\hat{\mathfrak{E}}(m,h|\mathfrak{X})\) depends on the unknown lag order \(m\) and bandwidth \(h\). In practice, one can either (a) estimate \(m\) first and then estimate the bandwidth \(h\); or (b) estimate the bandwidth \(h\) first and then estimate \(m\). However, given \(h\), maximizing the estimator~\eqref{eq:estimator_rlen} with respect to \(m\) is not an appropriate criterion for selecting \(m\), as the curve of \(\hat{\mathfrak{E}}(m,h|\mathfrak{X})\) against \(m\) varies significantly for different bandwidths (see~\Cref{fig:Relative_entropy_against_lag_order_for_different_bandwidths}). In practice, our suggestion is to determine the lag order \(m\) prior to computing the relative entropy \(\hat{\mathfrak{E}}(m,h|\mathfrak{X})\). Then, for the specified \(\hat{m}\), we select the bandwidth by maximizing the estimator~\eqref{eq:estimator_rlen}, i.e,
\begin{equation*}
    \hat{h}=\argmax_{h}\hat{\mathfrak{E}}(\hat{m},h|\mathfrak{X}).
\end{equation*}
In the next section, we use BIC criterion to select the optimal lag order \(\hat{m}\) based on the general nonlinear autoregression model.

\subsection{Lag Order Selection}\label{subsec:lag_order_selection}
There are various criteria proposed to address the lag order selection problem~\citep{shibataOptimalSelectionRegression1981,vieuOrderChoiceNonlinear1995,shaoAsymptoticTheoryLinear1997}. Here, we adopt an autoregression approach. We fit the following
general nonlinear autoregression model to each intermittent time series
\begin{equation}\label{eq:general_nonlinear_time_series}
    x_{i+m}=\mathfrak{F}\left( \mathbf{x}_{i;m} \right) + \varepsilon_{i},
\end{equation}
where \( 1\leq m \leq N-1 \), \( \varepsilon_{i} \) is Gaussian white noise with zero mean and variance \(\sigma^{2}\).~\( \mathfrak{F}(\cdot) \) is an unknown function. The Nadaraya-Watson estimator of \( \mathfrak{F}( \mathbf{x}_{i;m} ) \) can be expressed as:
\begin{equation}\label{eq:f_frak_estimator_nadaraya}
    \hat{\mathfrak{F}}\left( \mathbf{x}_{i;m}, h_{*} \right) = \sum\nolimits_{j=1}^{n}l_{j}\left(\mathbf{x}_{i;m}, h_{*} \right)x_{j+m},\quad i=1,2,\ldots,n
\end{equation}
where
\begin{equation*}
    l_{j}\left(\mathbf{x}_{i;m}, h_{*} \right)=\frac{\mathcal{K}_{h_{*}}^{(m)}\left( \mathbf{x}_{j;m}-\mathbf{x}_{i;m} \right)}{\sum_{s=1}^{n}\mathcal{K}_{h_{*}}^{(m)}\left( \mathbf{x}_{s;m}-\mathbf{x}_{i;m} \right)},\quad j=1,\ldots,n.
\end{equation*}
\(h_{*}\) is the bandwidth for nonlinear autoregression model.
Denote \( L(h_{*}) \) as
\begin{equation}\label{eq:Lmatrix}
    L(h_{*})\coloneqq\begin{bmatrix}
        l_{1}\left(\mathbf{x}_{1;m} , h_{*}\right) & l_{2}\left(\mathbf{x}_{1;m}, h_{*} \right) & \cdots & l_{n}\left(\mathbf{x}_{1;m}, h_{*} \right) \\
        l_{1}\left(\mathbf{x}_{2;m}, h_{*} \right) & l_{2}\left(\mathbf{x}_{2;m}, h_{*} \right) & \cdots & l_{n}\left(\mathbf{x}_{2;m}, h_{*} \right) \\
        \vdots & \vdots & \cdots & \vdots \\
        l_{1}\left(\mathbf{x}_{n;m}, h_{*} \right) & l_{2}\left(\mathbf{x}_{n;m}, h_{*} \right) & \cdots & l_{n}\left(\mathbf{x}_{n;m}, h_{*} \right) \\
    \end{bmatrix}.
\end{equation}
Then we have the following Lemma.
\begin{lemma}\label{lemma:degree_freedom}
    For the multivariate kernel \( \mathcal{K}_{h_{*}}^{(m)}(\cdot)  \)  and Nadaraya-Watson estimator~\eqref{eq:f_frak_estimator_nadaraya}, \( L(h_{*}) \) is defined in~\eqref{eq:Lmatrix}, the effective degrees of freedom \( v \) has explicit expression as
    \begin{equation*}
        v(m,h_{*})=\text{\textbf{tr}}(L(h_{*}))=\mathcal{K}_{h_{*}}^{(m)}\left( \mathbf{0} \right)\sum\nolimits_{i=1}^{n}{\left( \sum\nolimits_{s=1}^{n}\mathcal{K}_{h_{*}}^{(m)}\left( \mathbf{x}_{s;m}-\mathbf{x}_{i;m} \right) \right)}^{-1}=O(h_{*}^{-m}).
    \end{equation*}
\end{lemma}
The proof is straightforward, which will be omitted here. The bandwidth is selected by the so-called leave-one-out cross-validation method, i.e., minimizing
\begin{equation*}
    \hat{h}_{*}(m)=\argmin_{h_{*}}\sum\nolimits_{i=1}^{n}{\left(x_{i+m}-\hat{\mathfrak{F}}_{-i}\left( \mathbf{x}_{i;m} , h_{*}\right)\right)}^{2},
\end{equation*}
where
\begin{equation*}
    \hat{\mathfrak{F}}_{-i}\left( \mathbf{x}_{i;m}, h_{*} \right) = \sum_{j=1}^{n}\frac{\mathcal{K}_{h_{*}}^{(m)}\left( \mathbf{x}_{j;m}-\mathbf{x}_{i;m} \right)}{\sum_{s=1,s\neq i}^{n}\mathcal{K}_{h_{*}}^{(m)}\left( \mathbf{x}_{s;m}-\mathbf{x}_{i;m} \right)}x_{j+m}.
\end{equation*}
Define the average square predict error as
\begin{equation*}
    \hat{\sigma}^{2}(m) =\frac{1}{n}
    \sum\nolimits_{i=1}^{n}{\left(x_{i+m}-\hat{\mathfrak{F}}_{-i}\left( \mathbf{x}_{i;m}, \hat{h}_{*}(m)
    \right)\right)}^{2},
\end{equation*}
then, we have the following BIC criterion
\begin{equation}\label{eq:BIC_m}
    \mathfrak{B}(m) = n\log(\hat{\sigma}^{2}(m)) + v(m,\hat{h}_{*}(m))\log(n),\quad m=1,2,\ldots,N-1.
\end{equation}
We have

\begin{theorem}\label{thm:BIC_consistent}
    Supposing \( m_{0} \in \{ 1,2, \ldots,N-1\}\) be the underlying lag order. Let \( \hat{m}=\argmin_{m}\mathfrak{B}(m) \). Under conditions~\ref{cond:c11}--\ref{cond:c17} in~\Cref{sec:lag_order_selection_and_proof} of Appendix,  \( \hat{m} \)  converges to \( m_{0} \) in probability, i.e.,
    \begin{equation*}
        P(\hat{m}=m_{0})\rightarrow 1.
    \end{equation*}
\end{theorem}
The detailed proof of~\Cref{thm:BIC_consistent} can be found in~\Cref{sec:lag_order_selection_and_proof} of Appendix. The proof follows the framework of~\citet{vieuOrderChoiceNonlinear1995} combining the discussion in~\citet{shaoAsymptoticTheoryLinear1997}. We can use criterion~\eqref{eq:BIC_m} to choose \( m \) in advance, then implement the computation of relative entropy.
\subsection{Change-points Detection}\label{subsec:Change-points_Detection}%
Once we obtain the relative entropies of time series segments, denoted as \(\mathfrak{E}_{1},\mathfrak{E}_{2},\ldots,\mathfrak{E}_{J}\) where \( J \) represents the number of time series segments. Then, we can apply the existing detection methods such as CUSUM~\citep{pageContinuousInspectionSchemes1954} and its variants~\citep{inclanUseCumulativeSums1994,picardJointSegmentationCalling2011}, quasi-likelihood~\citep{braunMultipleChangepointFitting2000}. In this paper, we propose to use the optimal method~\citep{killickOptimalDetectionChangepoints2012} to search change-points because change-points in mean can be detected with a linear computational cost. Furthermore, their method has been officially adopted in the function~\textit{findchangepts} by MATLAB since 2018, which is convenient to be implemented in our algorithms below.

\subsection{Algorithms}\label{subsec:algorithms}
In practice, let \( \mathcal{X}={(x_{ij})}_{N\times J} \). Each column of \( \mathcal{X} \) represents a time series with length \( N \). Suppose \( \mathcal{X} \) has been transformed by the following logistic function,
\begin{equation}\label{eq:logistic_transform}
    \mathcal{X} = \frac{1}{1+\exp\left( -\mathcal{Y} \right)},
\end{equation}
where \( \mathcal{Y} \) represents the original time series observations without bounded support.
Similar to~\citet{hongAsymptoticDistributionTheory2005}, we use the logistic function~\eqref{eq:logistic_transform} to ensure the compact support in~\Cref{assump:density_assumption} throughout this chapter.

Finally, we summarize our approach using the following two algorithms:
\begin{algorithm}
    \DontPrintSemicolon{}
    \KwData{Matrix observation \(\mathcal{X}\) with size \( N \times J \).}
    \KwResult{\( \hat{m} \), the optimal lag order of \(\mathcal{X}\).}
    \kwInit{\(M \gets 10\)}
    \For{\( j\gets 1 \) \KwTo{} \(J \)}{
    \( m_{j} \gets 0 \) \;
    \( \mathcal{X}_{j} \gets \mathcal{X}(:, j)\) \;
    \For{\( m \gets 1 \) \KwTo{} \( M \)}{
    \( n\gets N-m \) \;
    \For{\( i\gets 1 \) \KwTo{} \( n \)}{
    \(\mathbf{x}_{i;m}\gets \mathcal{X}_{j}(i:(i+m-1))\)
    \;
    \( x_{i+m}\gets \mathcal{X}_{j}(i+m) \) \;
    }
    \( h_{*} \gets \argmin_{h_{*}}1/n\sum_{i=1}^{n}{\left(x_{i+m}-
    \hat{g}_{-i}\left(\mathbf{x}_{i;m}, h_{*} \right) \right)}^{2}\) \;
    \( v(m,h_{*})\gets \text{\textbf{tr}}(L(h_{*})) \) \;
    \( \hat{\sigma}^{2}
    \gets 1/n\sum_{i=1}^{n}{\left(x_{i+m}-\hat{g}\left( \mathbf{x}_{i;m}, h_{*}\right)\right)}^{2} \) \;
    \( \mathfrak{B}_{j}(m) \gets  n\log(\hat{\sigma}^{2}) + v(m,h_{*})\log(n) \) \;
    }
    }
    \(
    \overline{\mathfrak{B}}(m)=J^{-1}\sum_{j=1}^{J}\mathfrak{B}_{j}(m)
    \) \;
    \( \hat{m}\gets \argmin_{m}\{\overline{\mathfrak{B}}(m), m=1,\ldots,M	\}
    \) \;
    \caption{M step\label{al:m_selection_step}}
\end{algorithm}
\begin{algorithm}
    \DontPrintSemicolon{}
    \KwData{Matrix observation \(\mathcal{X}\) with size \( N \times J \).}
    \KwResult{\( \hat{j} \), the change-points.}
    \kwInit{\(\hat{m}\) from \textbf{Algorithm 1}}
    \For{\( j\gets 1 \) \KwTo{} \(J \)}{
        \( \mathfrak{E}_{j} \gets 0 \) \;
        \( \mathcal{X}_{j }\gets \mathcal{X}(:, j)\) \;
        \( \hat{h} \gets \argmin_{h} \left \{ \hat{\mathfrak{E}}(\hat{m}, h|\mathcal{X}_{j }) \right \} \)
        (Estimator~\eqref{eq:estimator_rlen}) \;
        \(\mathfrak{E}_{j}\gets \hat{\mathfrak{E}}(\hat{m}, \hat{h}|\mathcal{X}_{j })  \) \;
    }
    \( \hat{j}\gets \) change-point detected from \( \mathfrak{E}_{j},
    j=1,\ldots,J\) using the proposed  method~\citep{killickOptimalDetectionChangepoints2012} \;
    \caption{E Step\label{al:rlen_step}}
\end{algorithm}

In~\Cref{al:m_selection_step}, one can specify the initial value of \( M \). We take 10 as the default value following the suggestion from Section 4.5 in~\citet{wassermanAllNonparametricStatistics2006}. The bandwidth selection consumes the most computation time.  To reduce the time of selection \( h_{*} \), one can choose \( h_{*} \) moderately at the order \( O(n^{-1/(4+m)}) \) as an initial value, see Section 8.2 in~\citet{fanNonlinearTimeSeries2003} for more details.

In~\Cref{al:rlen_step}, we implement the bandwidth selection as well. The main difference is that the~\Cref{al:m_selection_step} includes the multivariate nonparametric regression, but in~\Cref{al:rlen_step}, \(\mathfrak{E}\) includes the multivariate nonparametric kernel density estimation.

\section{Theory}\label{sec:theory}
\citet{hongAsymptoticDistributionTheory2005} have proved that the relative entropy of pairwise variable \( (X_{t}, X_{t-j}) \) has a normal limiting distribution. The basic idea of~\citet{hongAsymptoticDistributionTheory2005}'s proof is to decompose the relative entropy into a few terms, followed by splitting each term into different parts and by neglecting the parts of smaller orders. Heuristically, the main parts can be expressed by the \textit{U}-statistics. By discussing the limiting distribution of these \textit{U}-statistics, they finally established the consistency theory of relative entropy for pairwise variables.

In this section, we develop a consistency theory of the relative entropy for \( m \) consecutive variables. By using their proving skills and ideas, we show that the limiting distribution of consecutive variables is Gaussian as well if \( m \) has an upper bound, say \( M \). The framework of our proof is very similar to that of~\citet{hongAsymptoticDistributionTheory2005}'s proof. Hence, the notations and most Lemmas and Theorems below originate from the theory and Appendix in~\citet{hongAsymptoticDistributionTheory2005}. However, our theory is not a straightforward extension from pairwise variables to \( m \) consecutive variables. There are some key points that need to be emphasized in our theory because they are different from those in~\citet{hongAsymptoticDistributionTheory2005}'s proof. In the following Theorems, Lemmas as well as their proofs in \textcolor{darkblue}{Appendix}~\ref{sec:consistency_of_rlen}, we will highlight these key points where they need to be emphasized.

From now on, we abbreviate \( \hat{\mathfrak{E}}(m, h|\mathfrak{X}) \) as \( \hat{\mathfrak{E}}(m) \).  Next we rewrite the estimator~\eqref{eq:estimator_rlen}, namely
\begin{equation}\label{eq:rewrite_estimator_rlen}\\
    \begin{aligned}
        \hat{\mathfrak{E}}(m) & = \frac{1}{n}\sum_{i\in S(m)}\left \{ \log\left[ \frac{f\left( \mathbf{x}_{i;m+1} \right)}{g\left( \mathbf{x}_{i;m} \right)g_{1}(x_{i+m})} \right]+\log\left[ \frac{\hat{f}\left( \mathbf{x}_{i;m+1} \right)}{f\left( \mathbf{x}_{i;m+1} \right)} \right] \right. \\
        & \qquad - \left. \log\left[ \frac{\hat{g}\left( \mathbf{x}_{i;m} \right)}{g\left( \mathbf{x}_{i;m} \right)} \right] -\log\left[ \frac{\hat{g}_{1}(x_{i+m})}{g_{1}(x_{i+m})} \right] \right \}, \\
        & = \hat{I}_{m}(f, g\cdot g_{1}) + \hat{I}_{m}(\hat{f}, f) - \hat{I}_{m}(\hat{g}, g) - \hat{I}_{1}(\hat{g}_{1}, g_{1}).
    \end{aligned}
\end{equation}

To obtain the consistency results of~\eqref{eq:rewrite_estimator_rlen}, we need the following assumption.
\begin{assumption}\label{assump:density_assumption}
    Suppose \( \{X_{t}\} \) is strictly stationary time series with the support
    \( \mathbb{I}=[0,1] \). Let \( \mathcal{G}:\mathbb{I}\rightarrow \mathbb{R}^{+}\)
    be the marginal density of \( X_{t} \). On support \( \mathbb{I} \), \( \mathcal{G}
    \) is away from 0 and has twice continuous differentiation
    \( \mathcal{G}^{(2)}(\cdot) \). Furthermore, \( \mathcal{G}^{(2)}(\cdot) \) satisfies the
    Lipschitz condition, i.e., for any \( x_1,x_2\in \mathbb{I}\),
    \( \left\vert \mathcal{G}^{(2)}(x_{1}) - \mathcal{G}^{(2)}(x_2) \right\vert\leq \mathcal{L} \left\vert x_{1}
    -x_{2} \right\vert \), where \( \mathcal{L} \) is the Lipschitz constant.
\end{assumption}

Moreover, at the bounds 0 and 1, the first and second derivatives of \( \mathcal{G}(\cdot) \) are defined by their right-hand derivative and left-hand derivative, respectively.~\textcolor{darkblue}{Assumption}~\ref{assump:density_assumption} is quite general and can avoid the slower convergence rate at the bounds of \( \mathbb{I} \)~\citep{hallEstimatingDirectionWhich1988,robinsonConsistentNonparametricEntropyBased1991,hongAsymptoticDistributionTheory2005}.

We prove the limiting distribution of \(\hat{\mathfrak{E}}(m)\) is Gaussian with the rate \( \sqrt{n}h^{(m+1)/2} \) in \textcolor{darkblue}{Theorem}~\ref{thm:I_m}.
\begin{theorem}\label{thm:I_m}
    Given \textcolor{darkblue}{Assumptions}~\ref{assump:kernel_assumption} and~\ref{assump:density_assumption}, if \(
    nh^{m+1}\rightarrow
    \infty \), \( {(\log
            n)}^{1/2}/(nh^{m+1})\rightarrow 0 \), \(nh^{m+13}\rightarrow 0  \)
    and \( 1\leq m<M \), under \(\mathbb{H}_{0}: f\left( \mathbf{x}_{i;m+1} \right)=g\left( \mathbf{x}_{i;m} \right)g_{1}(X_{i+m})\), we have \(  \)
    \begin{equation}\label{eq:I_m_clt}
        \sqrt{n}h^{(m+1)/2}\left[ 2\hat{\mathfrak{E}}(m)+d_{0}+d_{1} \right] \xrightarrow{d} N(0, \sigma_{*}^{2}),
    \end{equation}
    where \( \sigma_{*}^{2}= 2\beta\kappa^{m}+\beta_{1}\tau_{1}^{m}+2\beta_{2}\tau_{2}^{m}\), \( \tau_{1}= \int_{-1}^{1}{\left[ \int_{-1}^{1}
        K(u)K(u+v)\,\mathrm{d}u \right]}^{2} \,\mathrm{d}v \), \( \tau_{2}=
    \int_{-1}^{1} \int_{-1}^{1}K(u)K(v)K(u+v)\,\mathrm{d}u
    \,\mathrm{d}v \) and \(
    \beta=2n(n-m)(n-m+1)/[n^{2}{(n-1)}^{2}] \), \(
    \beta_{1}=\beta{{(n-2)}^{2}}/{(n-1)}^{2} \) and \( \beta_{2}=\beta(n-2)/(n-1) \).~\( d_{0} \) and \( d_{1} \) are the defined in Lemmas~\ref{lemma:A_2_nm} and~\ref{lemma:I_nm_divide} of the Appendix, respectively.
\end{theorem}
The proof of~\Cref{thm:I_m} can be found in \textcolor{darkblue}{Appendix}~\ref{sec:consistency_of_rlen}.~\textit{U}-statistics plays a significant role in the proof of relative entropy consistency, \textcolor{darkblue}{Appendix}~\ref{sec:lemmas_for_the_second_and_third_u_statistics} includes two lemmas about the second and third-order \textit{U}-statistics. In this section, we assume \( m \) can be arbitrarily large but be bounded by \( M \). It is desirable to relieve this limitation and let \( M \) go to infinity at a suitable rate, say \( M=O(\log(\log(n))) \). This type of extension of our theory is trivial, and the effect of \( M \rightarrow\infty \) needs to be carefully scrutinized, which will not be discussed here. Next, we carry out several numerical examples and real dataset analysis to evaluate our theory.

\section{Numerical Study}\label{sec:numerical_study}
\subsection{Case 1: Nonlinear intermittent time series.}\label{subsec:case_1_}
We consider the following two models, where Model 1 is adopted from Section 8.4 in~\citet[][]{fanNonlinearTimeSeries2003} and Model 2 is a variation of Model 1.
\begin{align*}
    \text{Model 1:} & \qquad x_{i} =-x_{i-2}\exp\left( -x_{i-2}^{2} /2 \right) + (1+x_{i-2}^{2})^{-1}\cos(\alpha x_{i-2})x_{i-1}+\varepsilon_{1i}, \\
    \text{Model 2:} & \qquad y_{i} = -y_{i-2}\exp\left( -y_{i-2}^{2} /2 \right) + (1+y_{i-2}^{2})^{-1}\sin(\alpha y_{i-2})y_{i-1}+\varepsilon_{2i},
\end{align*}
where \( \varepsilon_{1i} \) and \( \varepsilon_{2i} \) are Gaussian white noise with zero mean. The variances are \(\sigma_{1}^{2}=0.4^2\) and \( \sigma_{2}^{2}=0.5^2 \) respectively. Let \( N=400 \) be the length of time series, we generate \( P_{1}=30 \) time series from Model 1 and \( P_{2}=70 \) time series from Model 2. The total number of time series  \( P \) is \(P_{1}+P_{2} =100\). In this case, the change-point is at 31. The initial values of \( x_{1},x_{2},y_{1},y_{2} \) are all set to 1. In the first step, \textcolor{darkblue}{Algorithm}~\ref{al:m_selection_step} is implemented, and it correctly identifies the lag order, i.e., \( \hat{m}=2 \). Next, we apply \textcolor{darkblue}{Algorithm}~\ref{al:rlen_step} to the simulated dataset and compute the \(\mathfrak{E}\) values for each time series, namely, \(\hat{\mathfrak{E}}_{1},\ldots,\hat{\mathfrak{E}}_{P}\). Besides, we also use the approximated entropy (ApEn) method to calculate the conditional entropies of \(P\) time series, denoted as \(\hat{\mathfrak{A}}_{1},\ldots,\hat{\mathfrak{A}}_{P}\). Finally, we apply the change-point detection method~\citep{killickOptimalDetectionChangepoints2012} to sequences \(\hat{\mathfrak{E}}_{1},\ldots,\hat{\mathfrak{E}}_{P}\) and \(\hat{\mathfrak{A}}_{1},\ldots,\hat{\mathfrak{A}}_{P}\), respectively. Furthermore, we randomly draw \( \alpha \) from the interval \( [1,2] \) 150 times and repeat the previous procedures for each \( \alpha \) using ApEn and our algorithms. As described in \Cref{al:m_selection_step}, the logistic transformation is applied to ensure the data is confined within the compact support \([0,1]\). Following this transformation, the Augmented Dickey-Fuller test~\citep{adftest} is employed to assess the stationarity of the transformed time series, with an average stationarity rate of 100\% for all the time series under Case 1 settings.

\Cref{tab:The_Change-point_Detection_Based_on_ApEn_and_RlEn} presents the frequency of detected change-points based on the statistics \(\hat{\mathfrak{E}}_{1},\ldots,\hat{\mathfrak{E}}_{P}\) and \(\hat{\mathfrak{A}}_{1},\ldots,\hat{\mathfrak{A}}_{P}\), mean and variance respectively. For clarity, we only display the detected changes-points within the interval \([28,34]\). Our method identifies the underlying change-point at 31 with a percentage of 90.67\%. In contrast, ApEn method successfully detects the exact change-point in 85 out of 150 instances. Regarding change-point detection based on mean and variance, the percentages of correctly identifying the change-point at 31 are 0.007\% and 49.33\%, respectively.

\begin{table}[htbp]
    \centering
    \caption{The frequency of detected change-points based on ApEn, our method, mean and variance, respectively. The underlying change-point is at 31.~\(\mathfrak{A}\) and \(\mathfrak{E}\) represent the ApEn and our method. For simplicity, we only show the detected changes-points that are in the interval \([28,34]\).}\label{tab:The_Change-point_Detection_Based_on_ApEn_and_RlEn}
    \begin{tabular}{cccccccc}
        \toprule
        \multirow{2}{*}{Method}& \multicolumn{7}{c}{Detected Change-points} \\
        \cline{2-8}
        & 28 & 29 & 30 & 31 & 32 & 33 & 34 \\ \hline
        \(\mathfrak{A}\) & 6 & 8 & 21 & 85 & 15 & 7 & 3 \\
        \(\mathfrak{E}\) & 1 & 0 & 9 & 136 & 3 & 1 & 0 \\
        Mean & 0 & 0 & 4 & 1 & 0 & 0 & 1 \\
        Variance & 6 & 7 & 11 & 74 & 25 & 8 & 5 \\
        \bottomrule
    \end{tabular}

\end{table}
\subsection{Case 2: Linear intermittent time series.}\label{subsec:case_2}
Suppose there are two AR(3) processes:
\begin{align*} 
    \text{Process 1:} & \qquad x_{i} =\phi_{1}x_{i-1}+\phi_{2}x_{i-2} + \phi_{3}x_{i-3} +\varepsilon_{1i}, \\
    \text{Process 2:} & \qquad y_{i} =\phi_{1}^{*}y_{i-1}+\phi_{2}^{*}y_{i-2} + \phi_{3}^{*}y_{i-3} +\varepsilon_{2i},
\end{align*}
where \( \varepsilon_{1i} \), \( \varepsilon_{2i} \) are white noise with zero mean and variance \(\sigma^{2}_{1} \), \(\sigma^{2}_{2} \) respectively. It is easy to verify that the variance of \( x_{i} \) is \( \sigma_{1}^{2}/(1-\phi_{1}\rho_{1}-\phi_{2}\rho_{2} - \phi_{3}\rho_{3}) \) where

\begin{equation}\label{eq:rho_3}\\
    \begin{aligned}
        \rho_{1} & = -(\phi_{1}+\phi_{2}\,\phi_{3})/\phi_{d}, \\
        \rho_{2} & = -(\phi_{1}^2+\phi_{3}\,\phi_{1}-{\phi_{2}}^2+\phi_{2})/\phi_{d}, \\
        \rho_{3} & = -(\phi_{1}^3+{\phi_{1}}^2\,\phi_{3}+c_{1}\phi_{1}+c_{2})/\phi_{d},
    \end{aligned}
\end{equation}
and \( \phi_{d}= \phi_{3}^2+\phi_{1}\,\phi_{3}+\phi_{2}-1\), \( c_{1} = -{\phi_{2}}^2+2\,\phi_{2}-{\phi_{3}}^2 \), \( c_{2}={\phi_{2}}^2\,\phi_{3}-\phi_{2}\,\phi_{3}-{\phi_{3}}^3+\phi_{3}\). Suppose \( {x_{i}} \) and \( y_{i} \) have the same variance, then
\begin{equation}\label{eq:sigma1_sigma2_3}
    \sigma_{2}^{2}=\sigma_{1}^{2}\frac{1-\phi_{1}^{*}\rho_{1}^{*}-\phi_{2}^{*}\rho_{2}^{*} - \phi_{3}^{*}\rho_{3}^{*}}{1-\phi_{1}\rho_{1}-\phi_{2}\rho_{2} - \phi_{3}\rho_{3}},
\end{equation}
where \( \rho_{1}^{*},\rho_{2}^{*},\rho_{3}^{*} \) are the expressions of equation~\eqref{eq:rho_3} with \( \phi_{1},\phi_{2},\phi_{3} \) replaced by \( \phi_{1}^{*},\phi_{2}^{*},\phi_{3}^{*} \). We let \( \phi_{1}=0.8 \), \( \phi_{1}^{*}=0.7 \), \( \phi_{2}=\phi_{2}^{*}=-0.3 \), \( \phi_{3}=\phi_{3}^{*}=0.1 \) and \( \sigma^{2}_{1}=0.1 \),  \( \sigma_{2}^{2} \) is obtained according to equation~\eqref{eq:sigma1_sigma2_3}, i.e., 0.1168.  Let the length of time series be \( N=500 \), and randomly generate \( P_{1}=60 \) and \( P_{2}=40 \) time series from Process 1 and Process 2 respectively. Denote \( P = P_{1}+P_{2} \),  the change-point is at 61. To investigate the robustness of~\Cref{al:m_selection_step} with respect to the selection of \(m\), we appropriately allow \( m \) to change from 1 to 6. For each \( m \), both our method and ApEn are calculated using the same time series.
Last, repeat the above estimation procedure \( J=150 \) times. Let \( \tau  \) represent the change-point, we define the mean absolute distance (MAD) as \( \bar{\tau}={J}^{-1}\sum\nolimits_{j=1}^{J} \left\vert \tau_{j}-61 \right\vert \).

\textcolor{darkblue}{Table}~\ref{tab:Case_apen_rlen} shows the comparison results between our method and ApEn. The MAD based on our relative entropy is consistently smaller than that of ApEn for \( m=1,\ldots,6 \). The `Failure' columns in \textcolor{darkblue}{Table}~\ref{tab:Case_apen_rlen} represent the number of no change-point detected. The `Accuracy' column lists the percentage of exactly detecting the change-point at \( \tau=61 \).~\(\mathfrak{E}\) method can identify the change-point for the 150 repetitions,  out of which there are at least 105 exact detections. However, ApEn is not as robust as \(\mathfrak{E}\) when \( m \) is large, for instance, when \( m=6 \), the number of exactly detecting \( \tau=61 \) is 0 and the failure number of finding change-point is 116 for ApEn method. Especially, as~\citet{pincusApproximateEntropyMeasure1991}'s suggestion, \( m=2 \) is not a suitable choice in this simulation. Furthermore, this study also verifies that our \(\mathfrak{E}\) is robust with respect to the lag order. Even \( m \) is misspecified, the MAD is still less than 0.45. It needs to be pointed out that after the logistic transformation, 58.11\% of the transformed time series are stationary. ApEn results are totally based on the stationary time series. This means that our method is robust even for non-stationary time series as well.
\begin{table}[htbp]
    \centering
    \caption{The comparison between our method (\(\mathfrak{E}\)) and ApEn for different \(m\) in Case 2. The `Failure' column represents the number of instances where no change-point was detected. The `Accuracy' column represents the percentage of cases where the change-point was exactly detected at 61.}\label{tab:Case_apen_rlen}
    \begin{tabular}{ccccccc}
        \toprule
        \multirow{2}{*}{\( m \)} &\multicolumn{3}{c}{\(\mathfrak{E}\)}&\multicolumn{3}{c}{\(\mathfrak{A}\)} \\
        \cmidrule(lr){2-4}\cmidrule(lr){5-7}
        & MAD & Failure & \( Accuracy \) & MAD & Failure & \( Accuracy \) \\
        \hline
        1 & 0.3333 & 0 & 0.7800 & 0.6333 & 0 & 0.6467 \\
        2 & 0.3467 & 0 & 0.7333 & 33.9000 & 110 & 0.00 \\
        3 & 0.3533 & 0 & 0.7733 & 1.1200 & 0 & 0.5333 \\
        4 & 0.4200 & 0 & 0.7267 & 3.8200 & 0 & 0.2467 \\
        5 & 0.3667 & 0 & 0.7467 & 11.9145 & 33 & 0.1067 \\
        6 & 0.4600 & 0 & 0.7267 & 20.4762 & 108 & 0.0133 \\
        \bottomrule
    \end{tabular}

\end{table}


\section{Real Data Analysis}\label{sec:real_data_analysis}

\subsection{Multi-subjects Muscle Contraction Dataset}\label{subsec:Multi-subjects_Muscle_Contration_Dataset}
This dataset consists of 11 subjects' muscle contraction observations. Each subject needs to perform a series of intermittent isometric contractions (six seconds for contraction and four seconds for rest) until task failure~\citep[][]{pethickLossKneeExtensor2016}. Therefore, the number of each subject contractions is not consistent, see \textcolor{darkblue}{Table}~\ref{tab:Result_of_Change-point_Detection}. The sampling frequency is 1 kHz. We found that the \textcolor{darkblue}{Figures}~\ref{fig:rlen_three_groups1.png} and~\ref{fig:rlen_three_groups2.png} share the similar patterns, and both are significantly different to \textcolor{darkblue}{Figure}~\ref{fig:rlen_three_groups3.png}. Hence, in this study, we only find one change-point. Furthermore, based on the analysis of selection of \( m \) in Cases 2 and Case 3 in Appendix, the selection of \( m \)  is not sensitive to the change-point detection. In many research fields, ApEn is frequently employed to evaluate the complexity of signals~\citep{richmanPhysiologicalTimeseriesAnalysis2000,buriokaApproximateEntropyElectroencephalogram2005,pethickLossKneeExtensor2016}. Considering the computational complexity, we set \( m=2 \) to coordinate with ApEn. The change-point detections based on \(\mathfrak{E}\) for each subject are summarized in \textcolor{darkblue}{Table}~\ref{tab:Result_of_Change-point_Detection}. In contrast, we also obtain the change-point results based on ApEn, see \textcolor{darkblue}{Table}~\ref{tab:Result_of_Change-point_Detection_apen}. The parameter settings for ApEn follow the suggestions in~\citet{pincusApproximateEntropyMeasure1991}.
\begin{table}[htbp]
    \centering
    \caption{Result of Change-point Detection Based on
        \(\mathfrak{E}\)}\label{tab:Result_of_Change-point_Detection}
    \resizebox{0.85\linewidth}{!}{%
        \begin{tabular}{ccccccc}
            \toprule
            Subject & \( N_{c} \) & \( CP \) & \(CP/ N_{c} \) & \( \overline{\mathfrak{E}}_{1}(std.)  \) & \( \overline{\mathfrak{E}}_{2}(std.)  \) & \( p \)-value \\\hline
            1 & 70 & 22 & 31.43\% & 4.0315(0.2065) & 4.2406(0.2129) & 4.30e-04 \\
            2 & 38 & 11 & 28.95\% & 3.6582(0.1889) & 4.3606(0.2218) & 1.22e-08 \\
            3 & 54 & 23 & 42.59\% & 3.8257(0.2571) & 4.5578(0.2871) & 4.52e-13 \\
            4 & 79 & 54 & 68.35\% & 4.1137(0.2065) & 4.3692(0.2502) & 5.12e-05 \\
            5 & 289 & 236 & \textbf{81.66}\% & 3.4292(0.2286) & 3.8779(0.2563) & 1.07e-18 \\
            6 & 54 & 40 & 74.07\% & 3.8749(0.2395) & 4.6409(0.1624) & 5.78e-16 \\
            7 & 80 & 49 & 61.25\% & 3.6241(0.2112) & 4.1233(0.2889) & 2.81e-11 \\
            8 & 177 & 78 & 44.07\% & 4.1200(0.2251) & 4.4312(0.1759) & 4.03e-18 \\
            9 & 52 & 23 & 44.23\% & 3.7092(0.1310) & 4.3786(0.2141) & 1.31e-18 \\
            10 & 87 & 19 & \textbf{21.84}\% & 3.9454(0.1881) & 4.3561(0.1665) & 1.05e-08 \\
            11 & 89 & 38 & 42.70\% & 3.9879(0.1999) & 4.4219(0.2508) & 3.62e-14 \\
            \bottomrule
        \end{tabular}}
\end{table}
\begin{table}[htbp]
    \centering
    \caption{Result of Change-point Detection Based on
        ApEn}\label{tab:Result_of_Change-point_Detection_apen}
    \resizebox{0.85\linewidth}{!}{%
        \begin{tabular}{ccccccc}
            \toprule
            Subject & \( N_{c} \) & \( CP \) & \(CP/ N_{c} \) & \( \overline{ApEn}_{1}(std.)  \) & \( \overline{ApEn}_{2}(std.)  \) & \( p \)-value \\\hline
            1 & 70 & 68 & 97.14\% & 0.0062(0.0023) & 0.0114(0.0050) & 0.215 \\
            2 & 38 & 10 & 26.32\% & 0.0134(0.0043) & 0.0037(0.0018) & 1.15e-04 \\
            3 & 54 & 18 & 33.33\% & 0.0181(0.0067) & 0.0040(0.0025) & 1.33e-07 \\
            4 & 79 & -- & -- & -- & -- & -- \\
            5 & 289 & 239 & 82.70\% & 0.0139(0.0053) & 0.0073(0.0031) & 1.64e-22 \\
            6 & 54 & 26 & 48.15\% & 0.0144(0.0046) & 0.0040(0.0031) & 5.22e-12 \\
            7 & 80 & 48 & 60.00\% & 0.0129(0.0042) & 0.0059(0.0030) & 4.54e-13 \\
            8 & 177 & 78 & 44.07\% & 0.0068(0.0019) & 0.0047(0.0013) & 8.36e-13 \\
            9 & 52 & 19 & 36.54\% & 0.0231(0.0059) & 0.0046(0.0031) & 1.94e-11 \\
            10 & 87 & 33 & 37.93\% & 0.0098(0.0023) & 0.0063(0.0026) & 5.80e-09 \\
            11 & 89 & 39 & 43.82\% & 0.0123(0.0031) & 0.0047(0.0021) & 4.41e-19 \\
            \bottomrule
        \end{tabular}}
\end{table}

In \textcolor{darkblue}{Tables}~\ref{tab:Result_of_Change-point_Detection} and~\ref{tab:Result_of_Change-point_Detection_apen}, \( N_{c} \) represents the number of contractions in the series of experiments. \( CP \) stands for the change-point detected by ApEn or \(\mathfrak{E}\). \(CP/ N_{c} \) is the percentage of the way through the trial.~\( \overline{\mathfrak{E}}_{1}(std.)  \), \( \overline{\mathfrak{E}}_{2}(std.)  \), \( \overline{ApEn}_{1}(std.)  \) and \( \overline{ApEn}_{2}(std.)  \) stand for the two groups entropy averages (standard deviation) for \(\mathfrak{E}\) and ApEn respectively.
The last column shows the \( p \)-values of \( t \)-test for the mean comparison of the two groups.

In \textcolor{darkblue}{Table}~\ref{tab:Result_of_Change-point_Detection}, we can conclude that the intermittent isometric contractions of each subject can be divided into two groups which are supported by the \( p \)-values in the last column. The averages of \(\mathfrak{E}\) in the first group \( \overline{\mathfrak{E}}_{1}  \) are consistently smaller than \( \overline{\mathfrak{E}}_{2}  \). It is not surprising because muscle fatigue will increase the entropy of contraction signals~\citep[][]{pethickLossKneeExtensor2016}. According to \(CP/ N_{c} \), Subject 5 has the largest relative location of change-point while Subject 2's is just 21.84\%. Compared to other subjects, it means that the contraction torques are stable, and Subject 5 can maintain steady contractions for a long time.

In \textcolor{darkblue}{Table}~\ref{tab:Result_of_Change-point_Detection_apen}, the ``--'' represents the failure of change-point detection based on ApEn. Besides, the \( p \)-value of \( t \)-test for Subject 1 is even larger than 0.1, which means the change-point, 68, is not statistically reliable. It is also worth pointing out that for Subject 10, the change-point based on ApEn is 33 but is 19 based on \(\mathfrak{E}\).~\textcolor{darkblue}{Figure}~\ref{fig:Group_Division_for_Subject_10} shows the two divided groups using ApEn and \(\mathfrak{E}\) respectively. It is clear that the group in \textcolor{darkblue}{Figure}~\ref{fig:rlen_1.png} is more stable than the group illustrated by \textcolor{darkblue}{Figure}~\ref{fig:apen_1.png}. Moreover, there is no need to compare the averages of ApEn and \(\mathfrak{E}\) because ApEn has two free parameters and is not transformation invariant. The change-points of other subjects are almost the same for both ApEn and \(\mathfrak{E}\).
\begin{figure}[ht]
    \centering
    \begin{subfigure}[bt]{0.48\textwidth}
        \centering
        \includegraphics[width=\textwidth]{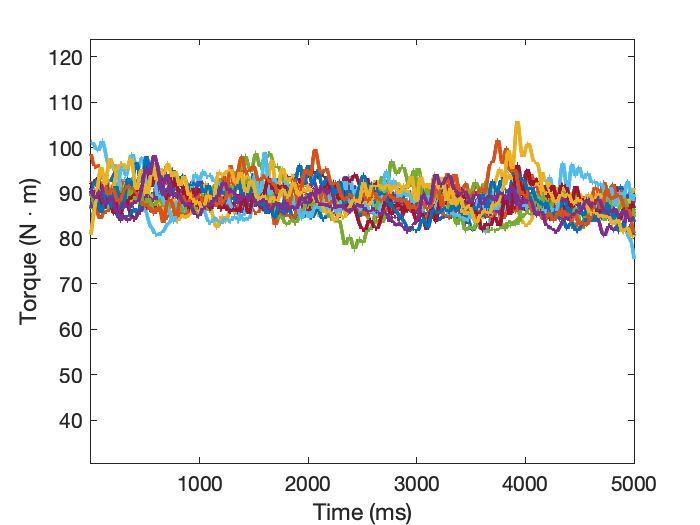}
        \caption{\(\mathfrak{E}\), First 18 Contractions}\label{fig:rlen_1.png}
    \end{subfigure}
    \begin{subfigure}[bt]{0.48\textwidth}
        \centering
        \includegraphics[width=\textwidth]{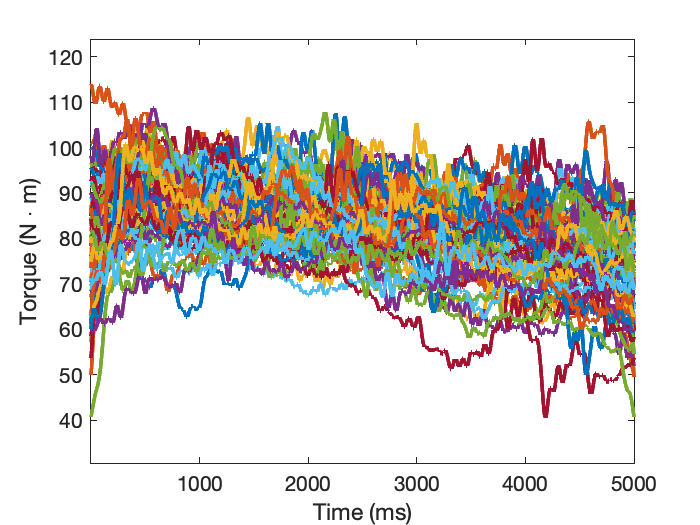}
        \caption{\(\mathfrak{E}\), Last 69 Contractions}\label{fig:rlen_2.png}
    \end{subfigure}
    \begin{subfigure}[bt]{0.48\textwidth}
        \centering
        \includegraphics[width=\textwidth]{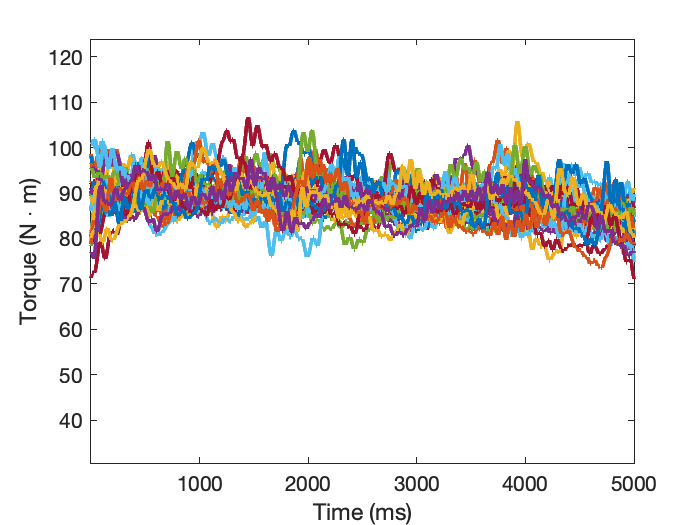}
        \caption{ApEn, First 32 Contractions}\label{fig:apen_1.png}
    \end{subfigure}
    \begin{subfigure}[bt]{0.48\textwidth}
        \centering
        \includegraphics[width=\textwidth]{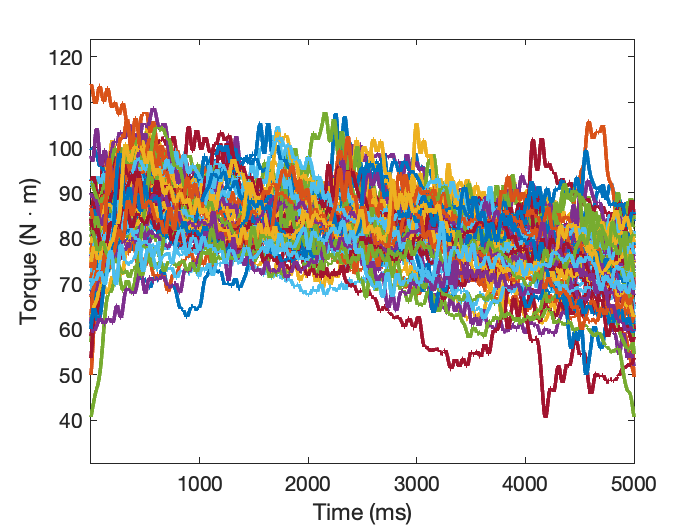}
        \caption{ApEn, Last 55 Contractions}\label{fig:apen_2.png}
    \end{subfigure}
    \caption{Divided Groups for Subject 10}\label{fig:Group_Division_for_Subject_10}
\end{figure}

Cases 1-3 show that the \(\mathfrak{E}\) is less sensitive to the lag order \( m \) and better than ApEn. Combining the results of muscle contraction data, i.e., \textcolor{darkblue}{Figure}~\ref{fig:Group_Division_for_Subject_10}, \textcolor{darkblue}{Tables}~\ref{tab:Result_of_Change-point_Detection} and~\ref{tab:Result_of_Change-point_Detection_apen}, our \(\mathfrak{E}\) performs better than the ApEn.


\section{Discussion}\label{sec:discussion}
In this article, we have proposed a nonparametric relative entropy as a testing statistic to detect the change-points among intermittent time series. In the nonparametric setting with weak dependence, we have developed a class of leave-one-out relative entropy. Given \textcolor{darkblue}{Assumptions}~\ref{assump:kernel_assumption} and~\ref{assump:density_assumption}, we have proved that the relative entropy has a limiting normal distribution with of order \( \sqrt{n}h^{(m+1)/2} \). Similarly, we have also discussed the selection of lag order \( m \). We suggest using BIC to select \( m \). If \( m \) has an upper bound, a theory of the selection of \(m\) can ensure consistency based on BIC from the point of view of nonparametric regression. Three simulations have shown that the relative entropy is appropriate to summaries the information of a time series. Based on RlEn, one can find the change-points among intermittent time series with high accurateness. One real data analysis has shown that our approach is effective in terms of change-point detection in practice as well.
\section*{Data and Code Availability}
The code is available at GitHub repository:~\href{https://github.com/Jieli12/RlEn}{https://github.com/Jieli12/RlEn}. Please contact Mark Burnley (m.burnley@lboro.ac.uk) and Samantha L. Winter (s.l.winter@lboro.ac.uk) to access the real sports data.

\section*{Competing interests}
We declare that we have no competing interests.

\section*{Acknowledgements}

Jie Li's work is supported by GTA PhD Scholarship and the Knowledge Transfer Partnership (KTP, No.10053865) project at the University of Kent.  The authors thank the editor, co-editor and the anonymous referees for very useful comments that improved the presentation of the paper.

\small\bibliographystyle{chicago}
\bibliography{reference.bib}
\clearpage
\newpage
\appendix
\section{Relative Entropy for Stationary AR(2) Process}\label{sec:relative_entropy_for_stationary_ar2_model}

Without loss of generality, let
\begin{equation}\label{eq:ar2_model}
    x_{i}=\phi_{1}x_{i-1}+\phi_{2}x_{i-2} + \varepsilon_{i},
\end{equation}
represent the AR\( (2) \) process without intercept, where \( \varepsilon_{i} \) is Gaussian white noise with zero mean and variance \( \sigma^{2} \). Suppose \( -1 < \phi_{2} < 1- \left\vert \phi_{1} \right\vert \), then process~\eqref{eq:ar2_model} is stationary. Let \( \gamma_{0}=E(X_{i}^{2}) \), \( \gamma_{1}=E(X_{i}X_{i-1})\) and \( \gamma_{2}=E(X_{i}X_{i-2})\). By~\eqref{eq:ar2_model}, we have the following Yule-Walker equations: \( \gamma_{0}  = \gamma_{0}\,{\phi_{1}}^2+2\,\gamma_{1}\,\phi_{1}\,\phi_{2}+\gamma_{0}\,{\phi_{2}}^2+\sigma^{2} \), \( \gamma_{1}  = \gamma_{0}\,\phi_{1}+\gamma_{1}\,\phi_{2} \) and \( \gamma_{2}  =\gamma_{1}\,\phi_{1}+\gamma_{0}\,\phi_{2} \).
Solving the above linear equations, we can get \( \gamma_{0}  = \sigma^{2}\,(\phi_{2}-1)/\phi_{c} \), \( \gamma_{1}  = -\phi_{1}\,\sigma^{2}/\phi_{c} \) and \( \gamma_{2}  = -\sigma^{2}\,({\phi_{1}}^2-{\phi_{2}}^2+\phi_{2})/\phi_{c} \)
where \( \phi_{c}=\left(\phi_{2}+1\right)\,\left({\phi_{1}}^2-{\phi_{2}}^2+2\,\phi_{2}-1\right) \). As \( \varepsilon_{i} \) is the Gaussian white noise, \( X_{i}, X_{i-1}, X_{i-2} \) have the following density function, namely,
\begin{equation*}
    f(x_{i},x_{i-1},x_{i-2}) = {(2\pi)}^{-\frac{3}{2}}{\left\vert \Sigma \right\vert}^{-\frac{1}{2}}\exp\left( - (x_{i},x_{i-1},x_{i-2})\Sigma^{-1}{(x_{i},x_{i-1},x_{i-2})}^{T}/2\right),
\end{equation*}
where
\begin{equation*}
    \Sigma=\begin{bmatrix}
        \gamma_{0} & \gamma_{1} & \gamma_{2} \\
        \gamma_{1} & \gamma_{0} & \gamma_{1} \\
        \gamma_{2} & \gamma_{1} & \gamma_{0}
    \end{bmatrix},\qquad
    \Sigma^{-1}=\frac{1}{\sigma^{2}}\left(\begin{array}{ccc} 1 & -\phi_{1} & -\phi_{2}\\ -\phi_{1} & {\phi_{1}}^2-{\phi_{2}}^2+1 & -\phi_{1}\\ -\phi_{2} & -\phi_{1} & 1 \end{array}\right),
\end{equation*}
and \( \left\vert\Sigma\right\vert=-\sigma^{6}/(\phi_{c}\left(\phi_{2}+1\right)) \).
Therefore, the entropy of \( f(x_{i},x_{i-1},x_{i-2}) \) is
\begin{equation}\label{eq:entropy_joint}\\
    \begin{aligned}
        \text{En}(f) & = -\iiint f(x_{i},x_{i-1},x_{i-2})\log\left( f(x_{i},x_{i-1},x_{i-2}) \right) \,\mathrm{d}x_{i}\,\mathrm{d}x_{i-1}\,\mathrm{d}x_{i-2}, \\
        & = 2^{-1}\log\left( {(2\pi e)}^{3} \left\vert\Sigma\right\vert\right).
    \end{aligned}
\end{equation}
By~\eqref{eq:ar2_model} and given \( x_{i-1},x_{i-2} \),  we can obtain the conditional density:
\begin{equation*}
    g(x_{i}|x_{i-1},x_{i-2})={(2\pi)}^{-1/2}{\left\vert \sigma \right\vert}^{-1}\exp\left( - {(x_{i}-\phi_{1}x_{i-1}-\phi_{2}x_{i-2})}^{2}/(2\sigma^{2})\right).
\end{equation*}
The conditional entropy is
\begin{align*}
    \text{CoEn}(f,g) & = -\iiint f(x_{i},x_{i-1},x_{i-2})\log\left( g(x_{i}|x_{i-1},x_{i-2}) \right) \,\mathrm{d}x_{i}\,\mathrm{d}x_{i-1}\,\mathrm{d}x_{i-2}, \\
    & = \frac{1}{2}\log(2\pi)+\log(\sigma)+\frac{1}{2\sigma^{2}}E\left[ {(x_{i}-\phi_{1}x_{i-1}-\phi_{2}x_{i-2})}^{2} \right], \\
    & =\frac{1}{2}\log(2\pi)+\log(\sigma)+\frac{1}{2\sigma^{2}}E\left[ y^{2} \right],
\end{align*}
where \( y=c^{T}\mathbf{x} \), \( c^{T}=(1, -\phi_{1}, -\phi_{2}) \), \( \mathbf{x}={(x_{i},x_{i-1},x_{i-2})}^{T} \). Apparently, \( y\sim N(0,c^{T}\Sigma c) \). It is easy to verify that \( c^{T}\Sigma c=\sigma^{2} \), so
\begin{equation}\label{eq:conditional_entropy_final}
    \text{CoEn}(f, g) = 2^{-1}\log(2\pi e)+\log(\sigma).
\end{equation}
We also notice that the density of \( x_{i} \) is \( g(x_{i})={(2\pi)}^{-1/2}{\left\vert \gamma_{0} \right\vert}^{-1/2}\exp\left( - x_{i}^{2}/2\gamma_{0}\right) \).
Finally, one can obtain the relative entropy
\begin{align}
    \text{RlEn}(f, g)
    & =2^{-1}\log\left( (\phi_{2}-1)/\phi_{c} \right)\label{eq:relative_entropy}.
\end{align}
Comparing Entropy~\eqref{eq:entropy_joint}, Conditional Entropy~\eqref{eq:conditional_entropy_final} with Relative Entropy~\eqref{eq:relative_entropy}, we conclude that RlEn is determined by the coefficients of autoregression coefficients and does not include the variance of noise in the AR(\textit{2}) process.
\section{Background noise-free}
For simplicity, we only show the property of background noise-free for AR\( (p) \) process. It is easy to extend to the general case such as MA\( (q) \), ARMA\( (p,q) \) process.
Without loss of generality, let the AR\( (p) \) process with zero mean be
\begin{equation}\label{eq:ar_p}
    x_{i}=\phi_{1}x_{i-1}+\phi_{2}x_{i-2} +\cdots + \phi_{p}x_{i-p} + \varepsilon_{i},
\end{equation}
where \( \phi_{1},\ldots, \phi_{p}  \) are autoregression coefficients, \( \varepsilon_{i} \) is the Gaussian white noise with zero mean and variance \( \sigma^{2} \), \( \{x_{i}\}_{1\leq i\leq n} \) and \( \{\varepsilon_{i}\}_{1\leq i\leq n} \) are dependent. Let \( \mathbf{x}_{m+1} ={(x_{i},\ldots,x_{i+m} )}^{\top}\), \( \mathbf{x}_{m+1-s} ={(x_{i},\ldots,x_{i+m-s} )}^{\top}\) and \( \mathbf{x}_{s} ={(x_{i+m-s+1},\ldots,x_{i+m} )}^{\top}\). \( \gamma_{k}=E(x_{i}x_{i-k}), k = 0, \pm 1, \pm 2, \ldots \) represent the auto-covariance functions, apparently \( \gamma_{k}=\gamma_{-k} \). Next, define \( \rho_{k}=\gamma_{k}/\gamma_{0}, k = 0, \pm 1, \pm 2, \ldots \) as the auto-correlation functions. Hence, we have the following proposition.
\begin{proposition}\label{thm:relative_entropy}
    Supposed \( \{x_{i}\} \) is a time series from the stationary AR\( (p) \) process
    defined in~\eqref{eq:ar_p}, \( \varepsilon_{i} \) is the Gaussian white noise
    with zero mean and variance \(\sigma^{2} \),
    then we have
    \begin{equation}\label{eq:mapping_function_y_to_x}
        \mathcal{I}_{s}\left(\mathbf{x}_{m+1}\right)=\frac{1}{2}\log\left( \frac{\left\vert R_{11;ms} \right\vert \left\vert R_{22;ms} \right\vert}{\left\vert R_{m} \right\vert} \right),\quad 1\leq s\leq m,
    \end{equation}
    which is independent of \( \sigma^{2} \), where \( \left\vert \cdot \right\vert \) is a matrix determinant operator.
\end{proposition}
To prove this proposition, we need the following lemma
\begin{lemma}\label{lemma:entropy}
    If multivariate variable \( \mathbf{x}\sim N_{p}(\mu,\Sigma) \), then the entropy of \( \mathbf{x} \), denoted as \( h(\mathbf{x}) \), is
    \begin{equation*}
        h(\mathbf{x}) = \frac{1}{2}\log\left( {(2\pi e)}^{p} \left\vert \Sigma \right\vert  \right).
    \end{equation*}
\end{lemma}
\begin{proof}
    By the definition of continuous entropy, we have
    \begin{align*}
        h(\mathbf{x}) & = -\int_{\mathbb{R}^{p}}f(\mathbf{x})\log\left( f(\mathbf{x}) \right)d\mathbf{x}, \\
        & = \frac{p}{2}\log(2\pi)+\frac{1}{2}\log(\left\vert \Sigma \right\vert) + \frac{1}{2}\mathbf{E}\left[ {\left( \mathbf{x}-\mu \right)}^{T}\Sigma^{-1}  \left( \mathbf{x}-\mu \right)\right], \\
        & =\frac{1}{2}\log\left( {(2\pi e)}^{p}\left\vert \Sigma \right\vert  \right),
    \end{align*}
    this completes the proof.
\end{proof}
Let \( \Sigma_{m}, \Sigma_{11;ms}, \Sigma_{22;ms}\) represent the auto-covariance matrices of vectors \( \mathbf{x}_{m+1} \), \( \mathbf{x}_{m+1-s} \) and \(\mathbf{x}_{s}  \) respectively, \( \Sigma_{12;ms} \) is the covariance matrix between \( \mathbf{x}_{m+1-s} \) and \(\mathbf{x}_{s}  \), i.e.,
\begin{equation*}
    \Sigma_{m}=\begin{bmatrix}
        \gamma_{0} & \gamma_{1} & \gamma_{2} & \cdots & \gamma_{m} \\
        \gamma_{1} & \gamma_{0} & \gamma_{1} & \cdots & \gamma_{m-1} \\
        \gamma_{2} & \gamma_{1} & \gamma_{0} & \cdots & \gamma_{m-2} \\
        \vdots & \vdots & \vdots & \vdots & \vdots \\
        \gamma_{m} & \gamma_{m-1} & \gamma_{m-2} & \cdots & \gamma_{0}
    \end{bmatrix},
    \Sigma_{11;ms}=\begin{bmatrix}
        \gamma_{0} & \gamma_{1} & \cdots & \gamma_{m-s} \\
        \gamma_{1} & \gamma_{0} & \cdots & \gamma_{m-s-1} \\
        \gamma_{2} & \gamma_{1} & \cdots & \gamma_{m-s-2} \\
        \vdots & \vdots & \vdots & \vdots \\
        \gamma_{m-s} & \gamma_{m-s-1} & \cdots & \gamma_{0}
    \end{bmatrix},
\end{equation*}
and
\begin{equation*}
    \Sigma_{22;ms}=\begin{bmatrix}
        \gamma_{0} & \gamma_{1} & \cdots & \gamma_{s-1} \\
        \gamma_{1} & \gamma_{0} & \cdots & \gamma_{s-2} \\
        \gamma_{2} & \gamma_{1} & \cdots & \gamma_{s-3} \\
        \vdots & \vdots & \vdots & \vdots \\
        \gamma_{s-1} & \gamma_{s-2} & \cdots & \gamma_{0}
    \end{bmatrix},
    \Sigma_{12;ms}=\begin{bmatrix}
        \gamma_{m-s+1} & \gamma_{m-s+2} & \cdots & \gamma_{m} \\
        \gamma_{m-s} & \gamma_{m-s+1} & \cdots & \gamma_{m-1} \\
        \gamma_{m-s-1} & \gamma_{m-s} & \cdots & \gamma_{m-2} \\
        \vdots & \vdots & \vdots & \vdots \\
        \gamma_{1} & \gamma_{2} & \cdots & \gamma_{s}
    \end{bmatrix}.
\end{equation*}
For simplicity, we have
\begin{equation*}
    \Sigma_{m}=\begin{bmatrix}
        \Sigma_{11;ms} & \Sigma_{12;ms} \\
        \Sigma_{21;ms} & \Sigma_{22;ms}
    \end{bmatrix},
\end{equation*}
we also notice that \( R_{m}=\gamma^{-1}_{0}\Sigma_{m} \), \( R_{11;ms}=\gamma^{-1}_{0}\Sigma_{11;ms} \), \( R_{12;ms}=\gamma^{-1}_{0}\Sigma_{12;ms} \), \( R_{22;s}=\gamma^{-1}_{0}\Sigma_{22;ms} \).
Based on these facts, we now prove \textcolor{darkblue}{Proposition}~\ref{thm:relative_entropy}.
\begin{proof}[Proof of \textcolor{darkblue}{Proposition}~\ref{thm:relative_entropy}]
    Since \( \varepsilon_{i} \) is the Gaussian white noise, the distribution of \( \mathbf{x}_{m+1} \), \( \mathbf{x}_{m+1-s} \) and \(\mathbf{x}_{s}  \) are multivariate Gaussian with covariance matrices \( \Sigma_{m}, \Sigma_{11;ms}, \Sigma_{22;ms}\) respectively. By \textcolor{darkblue}{Lemma}~\ref{lemma:entropy}, we can obtain the following results: \( h\left(\mathbf{x}_{m+1}\right)  = 2^{-1}\log\left( {(2\pi e)}^{m+1}\left\vert \Sigma_{m} \right\vert \right) \), \( h\left(\mathbf{x}_{s}\right) =2^{-1}\log\left( {(2\pi e)}^{s} \left\vert \Sigma_{22;ms} \right\vert  \right) \) and \( h\left(\mathbf{x}_{m+1-s}\right) = 2^{-1}\log\left( {(2\pi e)}^{m+1-s} \left\vert \Sigma_{11;ms} \right\vert  \right) \).
    The relative entropy can be expressed as
    \begin{align*}
        \mathcal{I}_{s}\left(\mathbf{x}_{m+1}\right) & =\text{RlEn}_{s}=\int_{\mathbb{R}^{m+1}}f\left(\mathbf{x}_{m+1}\right)\log\left( \frac{f\left(\mathbf{x}_{m+1}\right)}{g\left(\mathbf{x}_{m+1-s}\right)g\left(\mathbf{x}_{s}\right)} \right)d\mathbf{x}_{m+1}, \\
        & = \frac{1}{2}\log\left( \frac{\left\vert R_{11;ms}  \right\vert \left\vert R_{22;ms} \right\vert}{\left\vert R_{m} \right\vert} \right),\quad 1\leq s\leq m,\quad m\geq 1.
    \end{align*}
    According to the Yule-Walker equations,
    \begin{equation*}
        \begin{bmatrix}
            \rho_{1} \\
            \rho_{2} \\
            \rho_{3} \\
            \vdots \\
            \rho_{p}
        \end{bmatrix}=\begin{bmatrix}
            1 & \rho_{1} & \rho_{2} & \cdots & \rho_{p-1} \\
            \rho_1 & {1} & \rho_{2} & \cdots & \rho_{p-2} \\
            \rho_{2} & \rho_{1} & 1 & \cdots & \rho_{p-3} \\
            \vdots & \vdots & \vdots & \vdots & \vdots \\
            \rho_{p-1} & \rho_{p-2} & \rho_{p-3} & \cdots & 1 \\
        \end{bmatrix}
        \begin{bmatrix}
            \phi_{1} \\
            \phi_{2} \\
            \phi_{3} \\
            \vdots \\
            \phi_{p}
        \end{bmatrix},
    \end{equation*}
    the autocorrelation \( \rho_{i}, i=1,2,\ldots,\infty \) are totally determined by coefficient \( \phi_{k}, k=1,\ldots, p \). Hence, \( \mathcal{I}_{s}(\mathbf{x}_{m+1}) \) is a function of \( \phi_{k}, k=1,\ldots, p \). As \( \{x_{i}\} \) and \( \{\varepsilon_{i}\} \) are independent, the autocorrelation function of \( \{x_{i}\} \)  is independent of \( \sigma \) which implies \( \mathcal{I}_{s}(\mathbf{x}_{m+1}) \) is independent of \( \sigma \) as well which completes the proof.
\end{proof}

\section{Lag order selection and proof}\label{sec:lag_order_selection_and_proof}
\begin{proof}[Proof of \textcolor{darkblue}{Theorem}~\ref{thm:BIC_consistent}] The proof is  based on the proofs of~\citet{vieuOrderChoiceNonlinear1995} and~\citet{shaoAsymptoticTheoryLinear1997}. First, we construct a new equation \(\sigma^{2}_{\lambda}(m)=\hat{\sigma}^{2}(m)[1+\lambda(m,\hat{h}_{m})] \), where \( \lambda(m,\hat{h}_{m})= v(m,\hat{h}_{m})\log(n)/n\). We can regard \( \lambda(m,\hat{h}_{m}) \) as a penalty part in \( \sigma^{2}_{\lambda}(m) \), when \( \lambda(m,\hat{h}_{m})\rightarrow 0 \), minimizing \( BIC(m) \) is equivalent to minimizing \( \sigma^{2}_{\lambda}(m) \) based on the fact that \( \log(1+x)\approx x \) as \( x\rightarrow 0 \). This proof skill is frequently adopted in discussion of BIC or AIC consistency, see, for example,~\citet[][p.46]{shibataOptimalSelectionRegression1981},~\citet[][314]{vieuOrderChoiceNonlinear1995} and~\citet[][232]{shaoAsymptoticTheoryLinear1997}. Therefore, the sketch of our proof can be summarized into the following two steps: (1) We need to discuss the consistency of lag order selected via \( \hat{\sigma}^{2}(m)\); (2) We extend the result to the penalty version \( \sigma^{2}_{\lambda}(m) \) with controlling \( \lambda(m,\hat{h}_{m})\rightarrow 0 \) in an  appropriate rate. This proof is very similar to that in lag order selection~\citep{vieuOrderChoiceNonlinear1995} except the definitions of \( \hat{\sigma}^{2}(m)\) and \( \lambda(m,\hat{h}_{m}) \). Next, we introduce the conditions used in our proof.
    \begin{enumerate}[label= (C\arabic*),start=1]
        \item\label{cond:c11} The time series \( {\{X_{i}\}}_{i\in \mathbb{N}} \) is \( \alpha \)-mixing, the mixing coefficient \( \alpha(n) \) satisfies: \( \exists\ s>0, \exists\ 0<t<1, \forall\ n\geq 1, \alpha(n)\leq st^{n} \).
        \item\label{cond:c12} For each \( 1\leq m<M \), there exists the nonlinear autoregression functions such that
        \begin{equation*}
            x_{i+m}=\mathfrak{F}\left( \mathbf{X}_{i;m} \right)+\varepsilon_{i,m},
        \end{equation*}
        where \( \mathbf{X}_{i;m} \) is independent of \( \varepsilon_{i,m} \) and \( \varepsilon_{i,m}, i=1,\ldots,n \) are noise with mean zero.

        \item\label{cond:c13} The unknown function \( \mathfrak{F} \) has second-order continuous derivation.
        \item\label{cond:c14} \( \forall\ q\geq 1, \exists\ \mathbb{M}_{q}  \) such that for any \( i, \mathbb{E}{\left\vert X_{i} \right\vert}^{q}\leq \mathbb{M}_{q} <\infty \).
        \item\label{cond:c15} Given \( m \), for some \( 0< \gamma_{m} < \infty \) and some \( 0<\eta_{m}<1/(4+m)  \), the  bandwidth satisfies:
        \begin{equation*}
            h_{*}\in \mathcal{H}_{n,m}=\left[ \gamma^{-1}_{m}n^{-\eta_{m}-1/(4+m)},\gamma_{m} n^{\eta_{m}-1/(4+m)} \right].
        \end{equation*}
        \item\label{cond:c16} \( \lambda(m,\hat{h}_{m}) \) is of order \( \log(n)/(n{(h_{*})}^{m}) \).
        \item\label{cond:c17} \( m \) could be arbitrary large but has an upper bound \(M\).
    \end{enumerate}
    Conditions~\ref{cond:c11}--\ref{cond:c16} are quoted from~\citet[][310--311]{vieuOrderChoiceNonlinear1995} with appropriate adjustments for our circumstance. Condition~\ref{cond:c17} controls the upper bound \( M \) to coordinate the proof of \(\mathfrak{E}\) in \textcolor{darkblue}{Appendix}~\ref{sec:consistency_of_rlen}. Given lag order \( m, 1\leq m\leq M \) and the underlying lag order \( m_{0} \), we define the distance between \( \mathfrak{F} ( \mathbf{X}_{i;m_{0}} ) \) and \( \hat{\mathfrak{F}} ( \mathbf{X}_{i;m}, h_{*} ) \) by
    \begin{equation*}
        \sigma_{0}^{2}(m, h_{*})=\frac{1}{N-\max(m,m_{0})}\sum_{i=1}^{\max(m,m_{0})}{\left[ \mathfrak{F} \left( \mathbf{X}_{i;m_{0}} \right)- \hat{\mathfrak{F}} \left( \mathbf{X}_{i;m}, h_{*}\right)\right]}^{2}.
    \end{equation*}
    Furthermore, let \( \sigma_{0}^{2}(m)=\inf_{h_{*}\in \mathcal{H}}\sigma_{0}^{2}(m, h_{*}) \), we have
    \begin{lemma}\label{lemma:bic_consistency}
        Given Conditions~\ref{cond:c11}--\ref{cond:c15}, \textcolor{darkblue}{Assumption}~\ref{assump:kernel_assumption} and \textcolor{darkblue}{Assumption}~\ref{assump:density_assumption}, the nonlinear autoregression process~\eqref{eq:general_nonlinear_time_series} has underlying lag order \( m_{0} \) and \( m_{0}\in \left \{ 1,\ldots,M  \right \} \), then we have
        \begin{enumerate}[label={(\alph*)}]
            \item \( \sigma_{0}^{2}(m_{0})\rightarrow 0 \), a.e.,
            \item  For \( 1\leq m< m_{0} \), there exists real positive constant \( c_{m}>0 \) such that
            \begin{equation*}
                \hat{\sigma}^{2}(m)-\hat{\sigma}^{2}(m_{0})\geq c_{m},\quad a.s.,
            \end{equation*}
            \item For \( m_{0}< m\leq M \), \( \exists\ c_{0}>0\ s.t. \) \( \hat{\sigma}^{2}(m_{0})-c_{0}\geq 0  \), a.s.\ and
            \begin{equation*}
                \frac{\hat{\sigma}^{2}(m)-c_{0}}{\hat{\sigma}^{2}(m_{0})-c_{0}}\rightarrow +\infty,\quad a.s.
            \end{equation*}
        \end{enumerate}
    \end{lemma}
    \begin{proof}[Proof of \textcolor{darkblue}{Lemma}~\ref{lemma:bic_consistency}]
        This proof employs the same techniques used in Lemma 1, Lemma 2 and Theorem 3 in~\citet{vieuOrderChoiceNonlinear1995}. It can be regarded as a special case of~\citet{vieuOrderChoiceNonlinear1995} except the upper bound \( M \). Given \( m \), \( n=N-m \), the average square predicted error of nonlinear autoregression is defined as
        \begin{equation*}
            \sigma^{2}(m, h_{*})=\frac{1}{n}\sum_{i=1}^{n}{\left[ \mathfrak{F} \left( \mathbf{X}_{i;m} \right)- \hat{\mathfrak{F}} \left( \mathbf{X}_{i;m} , h_{*}\right)\right]}^{2}.
        \end{equation*}
        Under \textcolor{darkblue}{Assumptions}~\ref{assump:kernel_assumption} and~\ref{assump:density_assumption}, we can rewrite the average square predict error~\citep[e.g.,][83--85]{liNonparametricEconometricsTheory2007} as
        \begin{equation}\label{eq:sigma_m_h_order}
            \sigma^{2}(m, h_{\ast})=\alpha_{1m}\frac{1}{nh_{*}^{m}}+\alpha_{2m}mh_{*}^{4}+o\left( \frac{1}{nh_{*}^{m}}+h_{*}^{4} \right),\quad a.s.,
        \end{equation}
        where \( \alpha_{1m} \) and \( \alpha_{2m} \) are constant. We can easily obtain the optimal bandwidth if we minimize the first two leading terms in equation~\eqref{eq:sigma_m_h_order}, denoted as
        \begin{equation*}
            \hat{h}_{m}={\left( \frac{4\alpha_{2m}}{\alpha_{1m}}n \right)}^{-\frac{1}{4+m}}.
        \end{equation*}
        Finally, we get
        \begin{equation}\label{eq:inf_order}
            \sigma^{2}(m)=\inf_{h_{*}\in \mathcal{H}_{n,m}}\sigma^{2}(m, h_{*})=\alpha_{3m}n^{-\frac{4}{4+m}}+o\left(n^{-\frac{4}{4+m}}\right),\quad a.s.,
        \end{equation}
        where \( \alpha_{3m} \) is a constant. Especially, when \( m=m_{0} \), then \( \sigma^{2}(m)=\sigma^{2}_{0}(m_{0}) \), immediately, \( \sigma_{0}^{2}(m_{0})\rightarrow 0 \) holds almost surely. This completes the proof of (a).

        \textit{Proof of (b).} For given \( m \), \( 1\leq m<m_{0} \), let \( \hat{h}_{m,cv}=\argmin_{h_{*}\in\mathcal{H}_{n,m}}\hat{\sigma}^{2}(m, h_{*})	 \), the bandwidth selected by least square cross-validation is still of order \( n^{-1/(4+m)} \) as \( \hat{h}_{m} \)~\citep{vieuSmoothingTechniquesTime1991,hallCrossValidationEstimationConditional2004}, so \( \hat{h}_{m,cv}\in \mathcal{H}_{n,m} \), we have
        \begin{equation}\label{eq:ratio_cv_if}
            \frac{\sigma^{2}(m,\hat{h}_{m,cv})}{\inf\limits_{ h_{*}\in\mathcal{H}_{n,m}}\sigma^{2}(m,h_{*})}\rightarrow 1,\quad a.s.
        \end{equation}
        In the proof of this property,~\citet{vieuSmoothingTechniquesTime1991} and~\citet{vieuOrderChoiceNonlinear1995} employed the following statement for nonlinear autoregression:
        \begin{equation}\label{eq:cv_difference}
            \hat{\sigma}^{2}(m,h_{*})-\sigma^{2}(m,h_{*})=\frac{1}{n}\sum_{i=1}^{n}\varepsilon^{2}_{i,m}+o\left( \sigma^{2}(m,h_{*}) \right), \quad a.s.,
        \end{equation}
        where the operation \( o(\cdot) \) is uniform over \( h_{*}\in\mathcal{H}_{n,m} \). Therefore, by equation~\eqref{eq:inf_order} and equation~\eqref{eq:ratio_cv_if}, \( \forall \) \( m \), \( 1\leq m<m_{0} \), we have  \( \sigma^{2}(m,\hat{h}_{m,cv})=o(n^{-4/(4+m)} ) \) almost surely. Then, minimizing equation~\eqref{eq:cv_difference}, we obtain for any \( 1\leq m<m_{0} \)
        \begin{equation}\label{eq:sigma_e_m}
            \hat{\sigma}^{2}(m)=\frac{1}{n}\sum_{i=1}^{n}\varepsilon^{2}_{i,m}+o(1),\quad a.s.
        \end{equation}
        Let \( \delta_{m}=\text{Var}(\varepsilon_{i,m}) \) and \( c_{m}= \delta_{m}-\delta_{m_{0}}\), by equation~\eqref{eq:sigma_e_m}, we have \( \hat{\sigma}^{2}(m)-\hat{\sigma}^{2}(m_{0})\rightarrow c_{m}\), a.s., and \( c_{m}>0 \) because from~\ref{cond:c15} and \textcolor{darkblue}{Assumption}~\ref{assump:density_assumption}, we know \( c_{m}\geq \text{ Var }[ \mathbb{E}(x_{i+m}|\mathbf{X}_{i;m_{0}})- \mathbb{E}(x_{i+m}|\mathbf{X}_{i;m})] \), because \( 1\leq m < m_{0} \), so \( \text{ Var }[ \mathbb{E}(x_{i+m}|\mathbf{X}_{i;m_{0}})- \mathbb{E}(x_{i+m}|\mathbf{X}_{i;m})]>0 \) which completes the proof (b).

        \textit{Proof of (c)}. For \( m_{0}<m\leq M \), replacing \( h_{*} \) in equation~\eqref{eq:cv_difference} with \( \hat{h}_{m,cv} \), we obtain
        \begin{equation*}
            \hat{\sigma}^{2}(m)-\sigma^{2}(m,\hat{h}_{m,cv})=\frac{1}{n}\sum_{i=1}^{n}\varepsilon^{2}_{i,m}+o\left(n^{-4/(4+m)} \right), \quad a.s.
        \end{equation*}
        We also note that~\eqref{eq:ratio_cv_if} implies
        \begin{equation*}
            \sigma^{2}(m,\hat{h}_{m,cv})=\inf\limits_{ h_{*}\in\mathcal{H}_{n,m}}\sigma^{2}(m,h_{*})+o\left(n^{-4/(4+m)} \right), \quad a.s.,
        \end{equation*}
        therefore
        \begin{equation}\label{eq:a_12}
            \hat{\sigma}^{2}(m)=\sigma^{2}_{0}(m)+\frac{1}{n}\sum_{i=1}^{n}\varepsilon^{2}_{i,m}+o\left(n^{-4/(4+m)} \right), \quad a.s.
        \end{equation}
        Like the discussion in~\citet[]{vieuOrderChoiceNonlinear1995}, by Bernstein's inequality for \( \alpha \)-mixing, for example, see the Theorem 3.1 in~\citet[][]{roussas1988probability}, we have
        \begin{equation}\label{eq:a_13}
            \frac{1}{N-m_{0}}\sum_{i=1}^{N-m_{0}}\varepsilon^{2}_{i,m_{0}}-\text{Var}(\varepsilon_{i,m_{0}})=o\left(n^{-4/(4+m_{0})}\right),\quad a.s.
        \end{equation}
        Combining~\eqref{eq:a_12} and~\eqref{eq:a_13}, we get
        \begin{equation*}
            \frac{\left\vert \hat{\sigma}^{2}(m)-c_{0} \right\vert}{\left\vert \hat{\sigma}^{2}(m_{0})-c_{0} \right\vert}=\frac{\sigma^{2}_{0}(m)}{\sigma^{2}_{0}(m_{0})}+o\left( \frac{\sigma^{2}_{0}(m)}{\sigma^{2}_{0}(m_{0})} \right),\quad a.s.,
        \end{equation*}
        where \( c_{0}=\text{Var}(\varepsilon_{i,m_{0}}) \). By the fact~\eqref{eq:inf_order}, we have \( \sigma^{2}_{0}(m)/\sigma^{2}_{0}(m_{0}) \rightarrow\infty \) almost surely. Immediately, by~\eqref{eq:a_12}, we can conclude \( \hat{\sigma}^{2}(m_{0})-c_{0}\geq 0  \), a.s., because \( \sigma^{2}_{0}(m) \) is positive. This completes the proof.
    \end{proof}

    Let \( \bar{m}=\argmin \hat{\sigma}^{2}(m) \), based on \textcolor{darkblue}{Lemma}~\ref{lemma:bic_consistency}, we immediately have
    \begin{equation*}
        \sigma^{2}_{0}(\bar{m})/
        \sigma^{2}_{0}(m_{0}) \rightarrow 1,\quad a.s.
    \end{equation*}
    Moreover, we add the penalty part to \( \hat{\sigma}^{2}(m) \), i.e., \(\sigma^{2}_{\lambda}(m)=\hat{\sigma}^{2}(m)[1+\lambda(m,\hat{h}_{m})] \), where \( \lambda(m,\hat{h}_{m})= v(m,\hat{h}_{m})\log(n)/n\). Based on \textcolor{darkblue}{Lemma}~\ref{lemma:degree_freedom}, Condition~\ref{cond:c16} makes sure that the penalty part is not arbitrarily large compared with 0. Based on \textcolor{darkblue}{Lemma}~\ref{lemma:bic_consistency}, previous discussion and Condition~\ref{cond:c17}, we have
    \begin{equation}\label{eq:converge_rate}
        \begin{aligned}
            \max_{m=1,\ldots,M}\frac{\left\vert \sigma^{2}_{\lambda}(m)-\hat{\sigma}^{2}(m) \right\vert}{\sigma^{2}_{0}(m)} & = \frac{\log(n)}{n^{\frac{4}{4+m}}}
            \leq \frac{\log(n)}{n^{\frac{4}{4+O \left( \sqrt{\log(n)} \right)}}}
            = o(1),\quad a.s.
        \end{aligned}
    \end{equation}
    Let \( \hat{m}=\argmin \sigma_{\lambda}^{2}(m) \), we have \(  \sigma^{2}_{0}(\hat{m}_{\lambda})/ \sigma^{2}_{0}(m_{0})\rightarrow 1\) almost surely. Because the previous results are almost surely convergence, so
    \begin{equation*}
        P(\hat{m}=m_{0})\rightarrow 1,
    \end{equation*}
    holds as well, which completes the whole proof.
\end{proof}
Note: \citet{vieuOrderChoiceNonlinear1995} claimed \( M=O(\log(n)) \), however, our result shows that \( M \) is at least of order \( O(\sqrt{\log(n)}) \), see equation~\eqref{eq:converge_rate}. Furthermore, if one want to control \( M \)  to tend to infinity not as fast as \( O(\sqrt{\log(n)}) \), the order of \( M \) could be \( O(\log(\log(n))) \). However, in order to keep \( M \) consistent in the proof of relative entropy theory, we sacrifice the relaxation of \( m \)  to infinity discussed above.

\begin{figure}[ht]
    \centering
    \begin{subfigure}[bt]{0.4\textwidth}
        \centering
        \includegraphics[width=\textwidth]{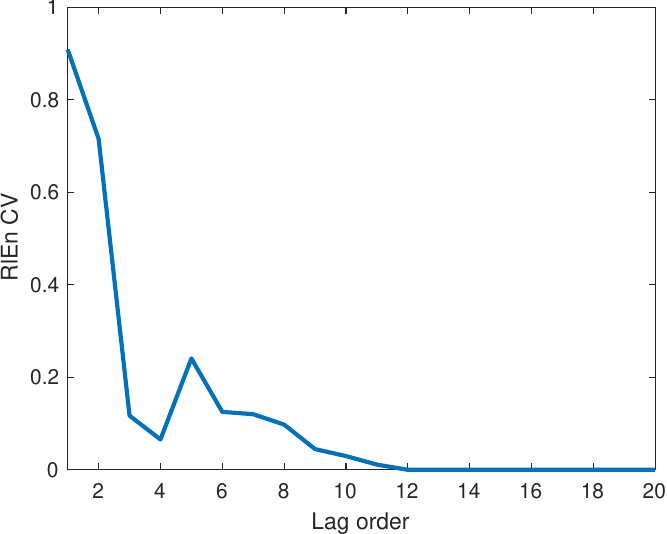}
        \caption{\( h=0.05 \)}\label{fig:RlEnh005.pdf}
    \end{subfigure}
    \begin{subfigure}[bt]{0.4\textwidth}
        \centering
        \includegraphics[width=\textwidth]{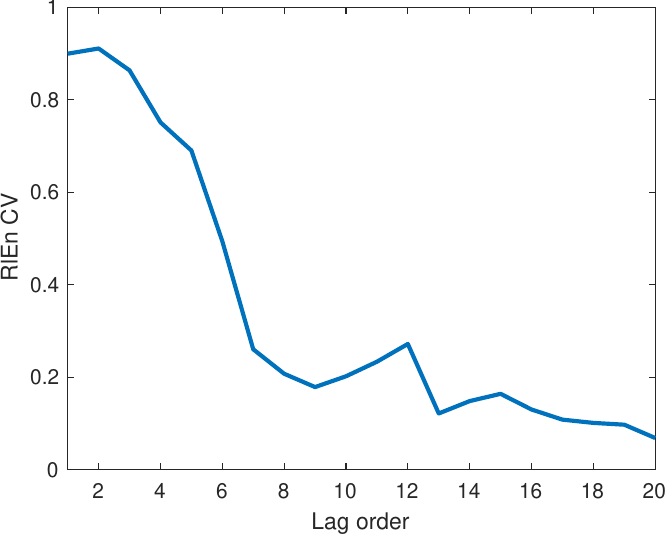}
        \caption{\( h=0.1 \)}\label{fig:RlEnh01.pdf}
    \end{subfigure}
    \begin{subfigure}[bt]{0.4\textwidth}
        \centering
        \includegraphics[width=\textwidth]{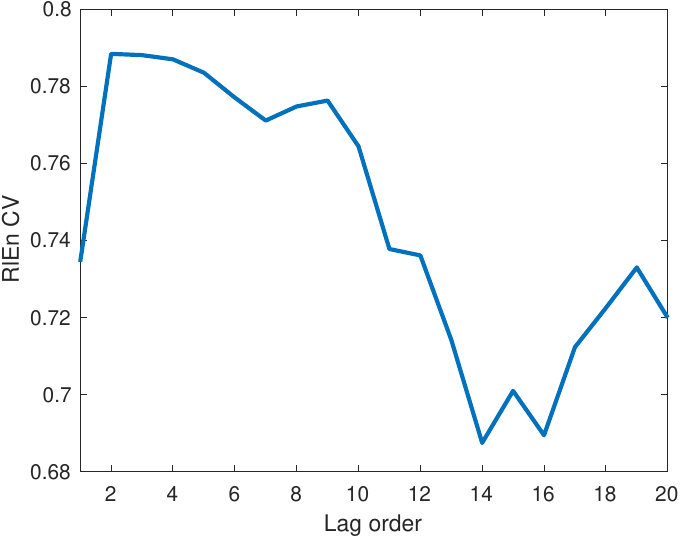}
        \caption{\( h=0.2 \)}\label{fig:RlEnh02.pdf}
    \end{subfigure}
    \begin{subfigure}[bt]{0.4\textwidth}
        \centering
        \includegraphics[width=\textwidth]{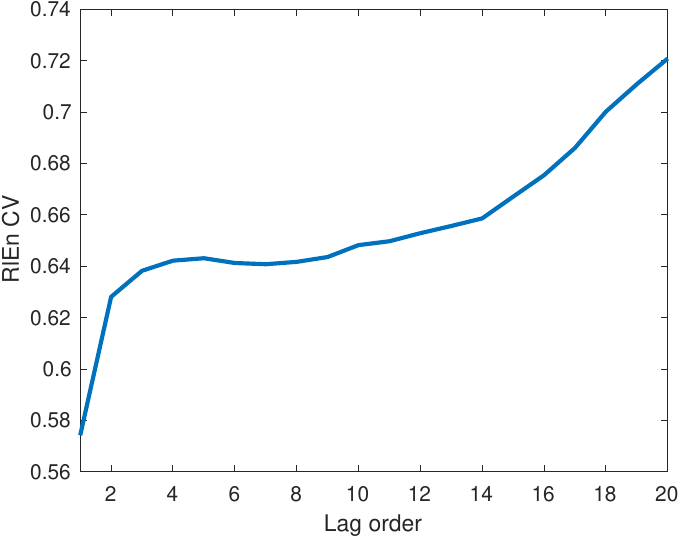}
        \caption{\( h=0.3 \)}\label{fig:RlEnh03.pdf}
    \end{subfigure}
    \caption{Relative Entropy against Lag Order for Different Bandwidths}\label{fig:Relative_entropy_against_lag_order_for_different_bandwidths}
\end{figure}

\section{Consistency of nonparametric complexity measure}\label{sec:consistency_of_rlen}
Under the following null hypothesis: \(\mathbb{H}_{0}: f\left( \mathbf{x}_{i;m+1} \right)=g\left( \mathbf{x}_{i;m} \right)g_{1}(X_{i+m})\),
the first term in equation~\eqref{eq:rewrite_estimator_rlen}, \( \hat{I}_{m}(f, g\cdot g_{1})=0 \) almost surely.\@ Note that for \( \left\vert x \right\vert<1 \), we have the inequality \( \left\vert \log(1+x)-x+\frac{1}{2}x^2 \right\vert \leq {\left\vert x \right\vert}^3\), so the third term in equation~\eqref{eq:rewrite_estimator_rlen} can be expressed as
\begin{equation}\label{eq:Inm_f_hat_f}\\
    \begin{aligned}
        \hat{I}_{m}(\hat{g}, g) & = \frac{1}{n}\sum_{i\in S(m)}\log\left[ 1+\frac{\hat{g}\left( \mathbf{x}_{i;m} \right)-g\left( \mathbf{x}_{i;m} \right)}{g\left( \mathbf{x}_{i;m} \right)} \right], \\
        & = \frac{1}{n}\sum_{i=1}^{n}\left[\frac{\hat{g}\left( \mathbf{x}_{i;m} \right)-g\left( \mathbf{x}_{i;m} \right)}{g\left( \mathbf{x}_{i;m} \right)} \right] \\
        & \qquad
        -\frac{1}{2}\sum_{i=1}^{n} {\left[\frac{\hat{g}\left( \mathbf{x}_{i;m} \right)-g\left( \mathbf{x}_{i;m} \right)}{g\left( \mathbf{x}_{i;m} \right)} \right]}^{2} + \text{remainder}, \\
        & =\hat{W}_{1}(m)-\frac{1}{2}\hat{W}_{2}(m)+\text{remainder}.
    \end{aligned}
\end{equation}
To expand the term \( \hat{W}_{1}(m) \), we need to introduce some notations: for any vector \( \mathbf{z}_{1}, \mathbf{z}_{2} \in \mathbb{I}^{m}\), define \( \bar{g}(\mathbf{z}_{1})=\int_{\mathbb{I}^{m}}\mathcal{K}^{(m)}_{h}(\mathbf{z}_{1},\mathbf{z}_{2})g(\mathbf{z}_{2})d\mathbf{z}_{2} \),
where \( \mathcal{K}^{(m)}_{h}(\mathbf{z}_{1},\mathbf{z}_{2}) =\mathcal{K}^{(m)}_{h}(\mathbf{z}_{1}-\mathbf{z}_{2}) \). Let
\begin{equation*}
    \tilde{\mathcal{K}}_{h}^{(m)}(\mathbf{z}_{1},\mathbf{z}_{2})=\mathcal{K}^{(m)}_{h}(\mathbf{z}_{1},\mathbf{z}_{2})-\int_{\mathbb{I}^{m}}\mathcal{K}^{(m)}_{h}(\mathbf{z},\mathbf{z}_{2})d\mathbf{z},
\end{equation*}
\begin{equation*}
    \tilde{A}_{m}(\mathbf{z}_{1},\mathbf{z}_{2})=\left[ \tilde{\mathcal{K}}_{h}^{(m)}(\mathbf{z}_{1},\mathbf{z}_{2})- \int_{\mathbb{I}^{m}}\tilde{\mathcal{K}}_{h}^{(m)}(\mathbf{z}_{1},\mathbf{z})g(\mathbf{z})d\mathbf{z}\right]/g(\mathbf{z}_{1}),
\end{equation*}
\begin{equation}\label{eq:A_nm}
    A_{m}(\mathbf{z}_{1},\mathbf{z}_{2})=\left[ \mathcal{K}_{h}^{(m)}(\mathbf{z}_{1},\mathbf{z}_{2})- \int_{\mathbb{I}^{m}}\mathcal{K}_{h}^{(m)}(\mathbf{z}_{1},\mathbf{z})g(\mathbf{z})d\mathbf{z}\right]/g(\mathbf{z}_{1}),
\end{equation}
\begin{equation*}
    \gamma_{m}(\mathbf{z}_{1},\mathbf{z}_{2})=\int_{\mathbb{I}^{m}}\left[ \mathcal{K}_{h}^{(m)}(\mathbf{z},\mathbf{z}_{2})- \int_{\mathbb{I}^{m}}\mathcal{K}_{h}^{(m)}(\mathbf{z},\mathbf{z}^{*})g(\mathbf{z}^{*})d\mathbf{z}^{*}\right]d\mathbf{z}/g(\mathbf{z}_{1}),
\end{equation*}
\begin{equation*}
    B_{m}(\mathbf{z}_{1})=\left[ \int_{\mathbb{I}^{m}}\mathcal{K}_{h}^{(m)}(\mathbf{z}_{1},\mathbf{z})g(\mathbf{z})d\mathbf{z}-g(\mathbf{z}_{1})\right]/g(\mathbf{z}_{1}),
\end{equation*}
\begin{equation}\label{eq:H_1nm_z1_z2}
    H_{1m}(\mathbf{z}_{1}, \mathbf{z}_{2})=\tilde{A}_{m}(\mathbf{z}_{1},\mathbf{z}_{2})+\tilde{A}_{m}(\mathbf{z}_{2},\mathbf{z}_{1}),
\end{equation}
\begin{equation*}
    H_{2m}(\mathbf{z}_{1}, \mathbf{z}_{2})=\int_{\mathbb{I}^{m}}A_{m}(\mathbf{z},\mathbf{z}_{1})A_{m}(\mathbf{z},\mathbf{z}_{2})g(\mathbf{z})\,\mathrm{d}\mathbf{z},
\end{equation*}
\begin{equation}\label{eq:c_hat_nm}
    \hat{C}(m)=\frac{1}{n}\sum\nolimits_{i=1}^{n}\int_{\mathbf{z}\in \mathbb{I}^{m}}A_{m}\left( \mathbf{z} ,\mathbf{x}_{i;m}\right) B_{m}(\mathbf{z})g(\mathbf{z})\,\mathrm{d}\mathbf{z},
\end{equation}
then, we have
\begin{equation}\label{eq:w_hat_1_expand}\\
    \begin{aligned}
        \hat{W}_{1}(m)
        & = \frac{1}{2}\binom{n}{2}^{-1} \sum_{j=2}^{n}\sum_{i=1}^{j-1} \left[ \tilde{A}_{m}\left( \mathbf{x}_{i;m},\mathbf{x}_{j;m} \right) + \tilde{A}_{m}\left( \mathbf{x}_{j;m},\mathbf{x}_{i;m} \right)\right] \\
        & \qquad+\frac{1}{2}\binom{n}{2}^{-1} \sum_{j=2}^{n}\sum_{i=1}^{j-1} \left[ \gamma_{m}\left( \mathbf{x}_{i;m},\mathbf{x}_{j;m} \right) + \gamma_{m}\left( \mathbf{x}_{j;m},\mathbf{x}_{i;m} \right)\right] \\
        & \qquad+\frac{1}{n}\sum_{i=1}^{n}B_{m}\left( \mathbf{x}_{i;m} \right), \\
        & = \frac{1}{2}\hat{H}_{1}(m)+\frac{1}{2}\hat{\Gamma}(m)+\hat{B}(m).
    \end{aligned}
\end{equation}

Next, we discuss the expansion of the second term in equation~\eqref{eq:Inm_f_hat_f}. We write
\begin{equation}\label{eq:w_hat_2_m}\\
    \begin{aligned}
        \hat{W}_{2}(m) & = \hat{W}_{21}(m)+ \hat{W}_{22}(m)+\hat{W}_{23}(m),
    \end{aligned}
\end{equation}
where
\begin{equation*}
    \hat{W}_{21}(m)=n^{-1}\sum\nolimits_{i=1}^{n}{\left[   \left( \hat{g}\left( \mathbf{x}_{i;m} \right)-\bar{g}\left( \mathbf{x}_{i;m} \right) \right) \big/g\left( \mathbf{x}_{i;m} \right) \right]}^{2},
\end{equation*}
\begin{equation*}
    \hat{W}_{22}(m)=n^{-1}\sum\nolimits_{i=1}^{n}{\left[   \left( \bar{g}\left( \mathbf{x}_{i;m} \right)-g\left( \mathbf{x}_{i;m} \right) \right) \big/g\left( \mathbf{x}_{i;m} \right) \right]}^{2},
\end{equation*}
and
\begin{equation*}
    \hat{W}_{23}(m)=\frac{2}{n}\sum_{i=1}^{n}\left[ \frac{\hat{g}\left( \mathbf{x}_{i;m} \right)-\bar{g}\left( \mathbf{x}_{i;m} \right)}{g\left( \mathbf{x}_{i;m} \right)} \right]\left[ \frac{\bar{g}\left( \mathbf{x}_{i;m} \right)-g\left( \mathbf{x}_{i;m} \right)}{g\left( \mathbf{x}_{i;m} \right)} \right].
\end{equation*}
Let \( D_{m}(\mathbf{z}_{1},\mathbf{z}_{2})=A_{m}^{2}(\mathbf{z}_{1},\mathbf{z}_{2})+A_{m}^{2}(\mathbf{z}_{2},\mathbf{z}_{1}) \) and
\begin{equation*}
    \begin{aligned}
        \tilde{H}_{2m}(\mathbf{z}_{1},\mathbf{z}_{2},\mathbf{z}_{3}) & = A_{m}(\mathbf{z}_{1},\mathbf{z}_{2})A_{m}(\mathbf{z}_{1},\mathbf{z}_{3}) + A_{m}(\mathbf{z}_{2},\mathbf{z}_{3})A_{m}(\mathbf{z}_{2},\mathbf{z}_{1}) \\
        & \qquad+A_{m}(\mathbf{z}_{3},\mathbf{z}_{1})A_{m}(\mathbf{z}_{3},\mathbf{z}_{2}),
    \end{aligned}
\end{equation*}
then after some simple calculations, we can obtain
\begin{equation}\label{eq:W_hat_21_m}
    \hat{W}_{21}(m) = \frac{1}{2(n-1)}\hat{D}(m)+\frac{1}{3}\frac{n-2}{n-1}\tilde{H}_{2}(m),
\end{equation}
where \( \hat{D}(m)=\binom{n}{2}^{-1}\sum_{j=2}^{n}\sum_{i=1}^{j-1}D_{m}( \mathbf{x}_{i;m},\mathbf{x}_{j;m} ) \)
and
\begin{equation*}
    \tilde{H}_{2}(m)=\binom{n}{3}^{-1}\sum_{k=3}^{n}\sum_{j=2}^{k-1}\sum_{i=1}^{j-1}\tilde{H}_{2m}( \mathbf{x}_{k;m}, \mathbf{x}_{i;m},\mathbf{x}_{j;m} ).
\end{equation*}
The remainder term in~\Cref{eq:Inm_f_hat_f} has the order of \(O_{p}\left(
n^{-3/2} h^{-3m/2}{(\log{n})}^{1/2}
+m^{2}h^{6} \right)\) (see~\Cref{lemma:remainder_term_in_I_ff})
Combining Equations~\eqref{eq:Inm_f_hat_f},~\eqref{eq:w_hat_1_expand},~\eqref{eq:w_hat_2_m} and~\eqref{eq:W_hat_21_m}, we finally have
\begin{equation}\label{eq:I_hat_nm_f_f_hat_remainder}\\
    \begin{aligned}
        \hat{I}_{m}(\hat{g}, g) & = \frac{1}{2}\hat{H}_{1}(m)+\frac{1}{2}\hat{\Gamma}(m)+\hat{B}(m)-\frac{1}{4(n-1)}\hat{D}(m) \\
        & \qquad-\frac{n-2}{6(n-1)}\tilde{H}_{2}(m)-\frac{1}{2}\hat{W}_{22}(m)-\frac{1}{2}\hat{W}_{23}(m) \\
        & \qquad+O_{p}\left( \frac{{(\log{n})}^{1/2}}{n^{3/2}h^{3m/2}} + m^{2}h^{6} \right).
    \end{aligned}
\end{equation}

In the following,~\textcolor{darkblue}{Appendix}~\ref{sec:consistency_of_rlen},~\Cref{lemma:Gamma_m,lemma:D_n_m,lemma:H_2n_m,lemma:w_22_hat,lemma:w_23_hat} give the orders of terms \(\hat{\Gamma}(m)\), \(\hat{D}(m)\), \(\tilde{H}_{2}(m)\), \(\hat{W}_{22}(m)\) and \(\hat{W}_{23}(m)\) in~\Cref{eq:I_hat_nm_f_f_hat_remainder} respectively. Based on these theory results, we immediately have~\Cref{thm:I_g_g}.
\begin{theorem}\label{thm:I_g_g}
    Given \textcolor{darkblue}{Assumptions}~\ref{assump:kernel_assumption} and~\ref{assump:density_assumption}, if \( 2\leq m<M \), \( nh^{m}\rightarrow
    \infty \), \( nh^{m+12}\rightarrow 0  \) and
    \( {(\log n)}^{1/2}/(nh^{m})\rightarrow 0 \), then under \( \mathbb{H}_{0} \),
    \begin{equation*}
        \hat{I}_{m}(\hat{g}, g)=\frac{1}{2}\hat{H}(m)-\frac{1}{2}L(m)+\left[\hat{B}(m)- \hat{C}(m)\right] +o_{p}(n^{-1/2}h^{-m/2}),
    \end{equation*}
    where \( \hat{H}(m)=\hat{H}_{1}(m)-(n-2)/(n-1)\hat{H}_{2}(m) \), \(
    L(m)={(n-1)}^{-1}\mathbb{E}A^{2}_{m}(\mathbf{z}_{1},\mathbf{z}_{2})
    + \mathbb{E}B_{m}^{2}(\mathbf{z}_{1})\), \(
    \mathbf{z}_{1},\mathbf{z}_{2} \in \mathbb{I}^{m}\) and \(
    \mathbf{z}_{1},\mathbf{z}_{2} \) have no overlap variable,
    \begin{align*}
        \hat{H}_{1}(m) & = \binom{n}{2}^{-1} \sum\nolimits_{j=2}^{n}\sum\nolimits_{i=1}^{j-1}H_{1m}\left( \mathbf{x}_{i;m}, \mathbf{x}_{j;m} \right), \\
        \hat{H}_{2}(m) & = \binom{n}{2}^{-1} \sum\nolimits_{j=2}^{n}\sum\nolimits_{i=1}^{j-1}H_{2m}\left( \mathbf{x}_{i;m}, \mathbf{x}_{j;m} \right).
    \end{align*}
\end{theorem}
The proof of \textcolor{darkblue}{Theorem}~\ref{thm:I_g_g} is straightforward which we omit here.
Similarly,  we can obtain the corresponding results with \( g(\cdot) \) being replaced by \( g_{1}(\cdot) \) and \( f(\cdot) \) respectively. Specifically, we have the following two theorems.
\begin{theorem}\label{thm:I_f_f}
    Given \textcolor{darkblue}{Assumptions}~\ref{assump:kernel_assumption} and~\ref{assump:density_assumption}, if \( nh^{m+1}\rightarrow
    \infty \), \( {(\log
            n)}^{1/2}/(nh^{m+1})\rightarrow 0 \), \(nh^{m+13}\rightarrow 0  \)
    and \( 1\leq m<M \),
    for any vector \( \mathbf{z}_{1}, \mathbf{z}_{2} \in
    \mathbb{I}^{m+1}\), define
    \begin{equation*}
        \bar{f}(\mathbf{z}_{1})=\int_{\mathbb{I}^{m+1}}\mathcal{K}^{(m+1)}_{h}(\mathbf{z}_{1},\mathbf{z}_{2})f(\mathbf{z}_{2})d\mathbf{z}_{2},
    \end{equation*}
    where \( \mathcal{K}^{(m+1)}_{h}(\mathbf{z}_{1},\mathbf{z}_{2})
    =\mathcal{K}^{(m+1)}_{h}(\mathbf{z}_{1}-\mathbf{z}_{2}) \). Let
    \begin{equation*}
        \tilde{K}_{h}^{J}(\mathbf{z}_{1},\mathbf{z}_{2})=\mathcal{K}^{(m+1)}_{h}(\mathbf{z}_{1},\mathbf{z}_{2})-\int_{\mathbb{I}^{m+1}}\mathcal{K}^{(m+1)}_{h}(\mathbf{z},\mathbf{z}_{2})d\mathbf{z},
    \end{equation*}
    \begin{equation*}
        \tilde{A}_{(m+1)}(\mathbf{z}_{1},\mathbf{z}_{2})=\left[ \tilde{K}_{h}^{J}(\mathbf{z}_{1},\mathbf{z}_{2})- \int_{\mathbb{I}^{m+1}}\tilde{K}_{h}^{J}(\mathbf{z}_{1},\mathbf{z})f(\mathbf{z})d\mathbf{z}\right]/f(\mathbf{z}_{1}),
    \end{equation*}
    \begin{equation*}
        A_{(m+1)}(\mathbf{z}_{1},\mathbf{z}_{2})=\left[ \mathcal{K}_{h}^{(m+1)}(\mathbf{z}_{1},\mathbf{z}_{2})- \int_{\mathbb{I}^{m+1}}\mathcal{K}_{h}^{(m+1)}(\mathbf{z}_{1},\mathbf{z})f(\mathbf{z})d\mathbf{z}\right]/f(\mathbf{z}_{1}),
    \end{equation*}
    \begin{equation*}
        \gamma_{(m+1)}(\mathbf{z}_{1},\mathbf{z}_{2})=\int_{\mathbb{I}^{m+1}}\left[ \mathcal{K}_{h}^{(m+1)}(\mathbf{z},\mathbf{z}_{2})- \int_{\mathbb{I}^{m+1}}\mathcal{K}_{h}^{(m+1)}(\mathbf{z},\mathbf{z}^{*})f(\mathbf{z}^{*})d\mathbf{z}^{*}\right]d\mathbf{z}/f(\mathbf{z}_{1}),
    \end{equation*}
    \begin{equation*}
        B_{(m+1)}(\mathbf{z}_{1})=\left[ \int_{\mathbb{I}^{m+1}}\mathcal{K}_{h}^{(m+1)}(\mathbf{z}_{1},\mathbf{z})f(\mathbf{z})d\mathbf{z}-f(\mathbf{z}_{1})\right]/f(\mathbf{z}_{1}),
    \end{equation*}
    \begin{equation*}
        H_{1(m+1)}(\mathbf{z}_{1}, \mathbf{z}_{2})=\tilde{A}_{(m+1)}(\mathbf{z}_{1},\mathbf{z}_{2})+\tilde{A}_{(m+1)}(\mathbf{z}_{2},\mathbf{z}_{1}),
    \end{equation*}
    \begin{equation*}
        H_{2(m+1)}(\mathbf{z}_{1}, \mathbf{z}_{2})=\int_{\mathbb{I}^{m+1}}A_{(m+1)}(\mathbf{z},\mathbf{z}_{1})A_{(m+1)}(\mathbf{z},\mathbf{z}_{2})f(\mathbf{z})\,\mathrm{d}\mathbf{z},
    \end{equation*}
    then under \( \mathbb{H}_{0} \), we have
    \begin{equation*}
        \begin{aligned}
            \hat{I}_{m}(\hat{f}, f) & = \frac{1}{2}\hat{H}(m+1)-\frac{1}{2}L(m+1) \\
            & \qquad +\left[\hat{B}(m+1)- \hat{C}(m+1)\right] +o_{p}(n^{-1/2}h^{-(m+1)/2}),
        \end{aligned}
    \end{equation*}
    where \( \hat{H}(m+1)=\hat{H}_{1}(m+1)-(n-2)/(n-1)\hat{H}_{2}(m+1) \), \(
    L(m+1)={(n-1)}^{-1}\mathbb{E}A^{2}_{(m+1)}(\mathbf{z}_{1},\mathbf{z}_{2})+ \mathbb{E}B_{(m+1)}^{2}(\mathbf{z}_{1})\), \(
    \mathbf{z}_{1},\mathbf{z}_{2} \in \mathbb{I}^{m+1}\) and \(
    \mathbf{z}_{1},\mathbf{z}_{2} \) have no overlap variable,
    \begin{align*}
        \hat{H}_{1}(m+1) & = \binom{n}{2}^{-1} \sum\nolimits_{j=2}^{n}\sum\nolimits_{i=1}^{j-1}H_{1(m+1)}\left( \mathbf{x}_{i;m+1}, \mathbf{x}_{j;m+1} \right), \\
        \hat{H}_{2}(m+1) & = \binom{n}{2}^{-1} \sum\nolimits_{j=2}^{n}\sum\nolimits_{i=1}^{j-1}H_{2(m+1)}\left( \mathbf{x}_{i;m+1}, \mathbf{x}_{j;m+1} \right), \\
        \hat{B}(m+1) & =n^{-1}\sum\nolimits_{i=1}^{n}B_{(m+1)}\left( \mathbf{x}_{i;m+1} \right),
    \end{align*}
    and \( \hat{C}(m+1)=n^{-1}\sum\nolimits_{i=1}^{n}\int_{\mathbf{z}\in \mathbb{I}^{m+1}}A_{(m+1)}\left( \mathbf{z} ,\mathbf{x}_{i;m+1}\right) B_{(m+1)}(\mathbf{z})f(\mathbf{z})\,\mathrm{d}\mathbf{z} \).
\end{theorem}
\begin{theorem}\label{thm:I_g2_g2}
    Given \textcolor{darkblue}{Assumptions}~\ref{assump:kernel_assumption} and~\ref{assump:density_assumption}, if \( nh\rightarrow
    \infty \), \( {(\log n)}^{1/2}/(nh)\rightarrow 0 \), \(nh^{13}\rightarrow 0  \),
    for any vector \( z_{1}, z_{2} \in
    \mathbb{I}\), define \( \bar{g}_{1}(z_{1})=\int_{0}^{1}K^{J}_{h}(z_{1},z_{2})g_{1}(z_{2})dz_{2} \),
    where \( K^{J}_{h}(z_{1},z_{2})
    =K^{J}_{h}(z_{1}-z_{2}) \). Let
    \begin{equation*}
        \tilde{K}_{h}^{J}(z_{1},z_{2})=K^{J}_{h}(z_{1},z_{2})-\int_{0}^{1}K^{J}_{h}(z,z_{2})dz,
    \end{equation*}
    \begin{equation*}
        \tilde{a}(z_{1},z_{2})=\left[ \tilde{K}_{h}^{J}(z_{1},z_{2})- \int_{0}^{1}\tilde{K}_{h}^{J}(z_{1},z)g_{1}(z)dz\right]/g_{1}(z_{1}),
    \end{equation*}
    \begin{equation*}
        a(z_{1},z_{2})=\left[ K_{h}^{J}(z_{1},z_{2})- \int_{0}^{1}K_{h}^{J}(z_{1},z)g_{1}(z)dz\right]/g_{1}(z_{1}),
    \end{equation*}
    \begin{equation*}
        \gamma(z_{1},z_{2})=\int_{0}^{1}\left[ K_{h}^{J}(z,z_{2})- \int_{0}^{1}K_{h}^{J}(z,z^{*})g_{1}(z^{*})dz^{*}\right]dz/g_{1}(z_{1}),
    \end{equation*}
    \begin{equation*}
        b(z_{1})=\left[ \int_{0}^{1}K_{h}^{J}(z_{1},z)g_{1}(z)dz-g_{1}(z_{1})\right]/g_{1}(z_{1}),
    \end{equation*}
    \begin{equation*}
        v_{1}(z_{1}, z_{2})=\tilde{a}(z_{1},z_{2})+\tilde{a}(z_{2},z_{1}),
    \end{equation*}
    \begin{equation*}
        v_{2}(z_{1}, z_{2})=\int_{0}^{1}a(z,z_{1})a(z,z_{2})g_{1}(z)\,\mathrm{d}z,
    \end{equation*}
    then under \( \mathbb{H}_{0} \), we have
    \begin{equation*}
        \hat{I}_{1}(\hat{g}_{1}, g_{1})=2^{-1}(\hat{V}_{m}-l_{m})+[\hat{b}_{m}- \hat{c}_{m}] +o_{p}(n^{-1/2}h^{-1/2}),
    \end{equation*}
    where \( \hat{V}_{m}=\hat{V}_{1}-(n-2)/(n-1)\hat{V}_{2} \), \(
    l_{m}={(n-1)}^{-1}\mathbb{E}a^{2}(z_{1},z_{2})+ \mathbb{E}b^{2}(z_{1})\), \(z_{1},z_{2} \in \mathbb{I}\),
    \begin{align*}
        \hat{V}_{1} & = \binom{n}{2}^{-1} \sum\nolimits_{j=2}^{n}\sum\nolimits_{i=1}^{j-1}v_{1}\left( x_{i+m}, x_{j+m} \right), \\
        \hat{V}_{2} & = \binom{n}{2}^{-1} \sum\nolimits_{j=2}^{n}\sum\nolimits_{i=1}^{j-1}v_{2}\left( x_{i+m}, x_{j+m} \right), \\
        \hat{b}_{m} & =n^{-1}\sum\nolimits_{i=1}^{n}b\left( x_{i+m} \right),
    \end{align*}
    and \( \hat{c}_{m}=n^{-1}\sum_{i=1}^{n}\int_{z\in \mathbb{I}}a( z ,x_{i+m}) b(z)g_{1}(z)\,\mathrm{d}z \).
\end{theorem}
By~\Cref{thm:I_g_g,thm:I_f_f,thm:I_g2_g2}, the estimator~\eqref{eq:rewrite_estimator_rlen} can be written as:
\begin{equation}\label{eq:I_m}\\
    \begin{aligned}
        2 \hat{\mathfrak{E}}(m) & =\hat{I}_{m}(f, g\cdot g_{1}) + \hat{I}_{m}(\hat{f}, f) - \hat{I}_{m}(\hat{g}, g) - \hat{I}_{1}(\hat{g}_{1}, g_{1}), \\
        & = \left[ \hat{H}(m+1) -\hat{H}(m)-\hat{V}_{m}\right] \\
        & \quad-  (n-1)^{-1}\left[ \mathbb{E}A^{2}_{(m+1)}(\mathbf{z}_{1},\mathbf{z}_{2})-\mathbb{E}A^{2}_{m}(\mathbf{y}_{1},\mathbf{y}_{2}) -\mathbb{E}a^{2}(z_{1m},z_{2m})\right] \\
        & \quad-  \left[ \mathbb{E}B^{2}_{(m+1)}(\mathbf{z}_{1})-\mathbb{E}B^{2}_{m}(\mathbf{y}_{1}) -\mathbb{E}b^{2}(z_{1m})\right] \\
        & \quad + 2\left[ \hat{B}(m+1)- \hat{B}(m)-\hat{b}_{m}\right] \\
        & \quad -2\left[ \hat{C}(m+1)- \hat{C}(m)-\hat{c}_{m}\right] \\
        & \quad + o_{p}(n^{-1/2}h^{-(m+1)/2}),
    \end{aligned}
\end{equation}
where \( \mathbf{z}_{1}={(z_{10},\ldots,z_{1(m-1)},z_{1m})}^{\top}={( \mathbf{y}_{1}^{\top},z_{1m} )}^{\top} \), \( \mathbf{z}_{2}={(z_{20},\ldots,z_{2(m-1)},z_{2m})}^{\top}={( \mathbf{y}_{2}^{\top},z_{2m} )}^{\top} \). Furthermore, the term \((n-1)^{-1}[ \mathbb{E}A^{2}_{(m+1)}(\mathbf{z}_{1},\mathbf{z}_{2})-\mathbb{E}A^{2}_{m}(\mathbf{y}_{1},\mathbf{y}_{2}) -\mathbb{E}a^{2}(z_{1m},z_{2m})]\) in~\eqref{eq:I_m} has order of \(d_{0}+O(n^{-1}h^{-m})\) where  \( d_{0}={(n-1)}^{-1}\kappa^{m+1}h^{-(m+1)} \) and \( \kappa = \int_{-1}^{1} K^{2}(u)\,\mathrm{d}u \), see~\Cref{lemma:A_2_nm}.
The term \(2[ \hat{B}(m+1)- \hat{B}(m)-\hat{b}_{m}]-[ \mathbb{E}B^{2}_{(m+1)}(\mathbf{z}_{1})-\mathbb{E}B^{2}_{m}(\mathbf{y}_{1}) -\mathbb{E}b^{2}(z_{1m})]\) in~\eqref{eq:I_m} has order of \(O(h^{6} )+O_{p}(n^{-1/2}h^{4})\), see~\Cref{lemma:b_part}, and the term \(\hat{C}(m+1)- \hat{C}(m)-\hat{c}_{m}\) in~\eqref{eq:I_m} has order of \(O_{p}(n^{-1/2}h^{4})\), see~\Cref{lemma:c_part}.

By~\Cref{lemma:A_2_nm,lemma:b_part,lemma:c_part}, we have
\begin{equation*}
    2\hat{\mathfrak{E}}(m)+d_{0}=\hat{H}(m+1) -\hat{H}(m)-\hat{V}_{m}+o_{p}(n^{-1/2}h^{-(m+1)/2}).
\end{equation*}
Recalling that \( \mathbf{z_{1}} \) and \( \mathbf{z}_{2} \) may have common variables in multivariate \textit{U}-statistics, it is impossible to apply the central limit theorem of \textit{U}-statistics to our case directly. So we need to divide \( \hat{H}(m+1) -\hat{H}(m)-\hat{V}_{m} \) into two parts: one part includes  independent components of \( \mathbf{z_{1}} \) and \( \mathbf{z}_{2} \), the other part includes dependent components of \( \mathbf{z_{1}} \) and \( \mathbf{z}_{2} \). We rewrite
\begin{align*}
    \hat{H}_{1}(m)
    & = \binom{n}{2}^{-1}\sum_{j=1+m}^{n}\sum_{i=1}^{j-m}H_{1m}\left( \mathbf{x}_{i;m}, \mathbf{x}_{j;m} \right) \\
    & \qquad +\binom{n}{2}^{-1}\sum_{j=2}^{n}\sum_{i=1\vee(j-m+1)}^{j-1}H_{1m}\left( \mathbf{x}_{i;m}, \mathbf{x}_{j;m} \right), \\
    & = T_{1}(m)+T_{10}(m),
\end{align*}
\begin{align*}
    \hat{H}_{2}(m)
    & = \binom{n}{2}^{-1}\sum_{j=1+m}^{n}\sum_{i=1}^{j-m}H_{2m}\left( \mathbf{x}_{i;m}, \mathbf{x}_{j;m} \right) \\
    & \qquad +\binom{n}{2}^{-1}\sum_{j=2}^{n}\sum_{i=1\vee(j-m+1)}^{j-1}H_{2m}\left( \mathbf{x}_{i;m}, \mathbf{x}_{j;m} \right), \\
    & = T_{2}(m)+T_{20}(m).
\end{align*}
Similarly, \( \hat{H}_{1}(m+1)=T_{1}(m+1)+T_{10}(m+1) \), \( \hat{H}_{2}(m+1)=T_{2}(m+1)+T_{20}(m+1) \). We have the following Lemma:
\begin{lemma}\label{lemma:I_nm_divide}
    Given \( \mathbb{H}_{0} \),
    \textcolor{darkblue}{Assumptions}~\ref{assump:kernel_assumption} and~\ref{assump:density_assumption}, let \( \breve{H}=\hat{H}(m+1) -\hat{H}(m)-\hat{V}_{m}
    \). If \(
    nh^{m+1}\rightarrow
    \infty \), \(nh^{m+13}\rightarrow 0  \), \( {(\log
            n)}^{1/2}/(nh^{m+1})\rightarrow 0 \)
    and \( 1\leq m<M \), we have \( \mathbb{E}\breve{H}=-d_{1}+o_{p}(n^{-1/2}h^{-(m+1)/2}) \)
    where \( \tau=\int_{-1}^{1} \int_{-1}^{1}
    K(u)K(u+v)\,\mathrm{d}u\,\mathrm{d}v \), \(
    d_{1}=(n-2)/(n-1)[ c_{1}(\tau^{m+1}-1)-c_{2}(\tau^{m}-1)
    ] \)

\end{lemma}
\begin{proof}[Proof of~\Cref{lemma:I_nm_divide}]
    By~\Cref{lemma:H1_expectation} and~\Cref{lemma:H2_expectation}, we have \( \mathbb{E}[T_{10}(m+1)]=\mathbb{E}[T_{10}(m)]=O(mn^{-1}h^{2}) \), \( \mathbb{E}[T_{20}(m)]=c_{2}(\tau^{m}-1) + O(mn^{-1}h) \), \( \mathbb{E}[T_{20}(m+1)]=c_{1}(\tau^{m+1}-1) + O(mn^{-1}h) \) where \( c_{2}=[(2n-m)(m-1)]/[(n-m)(n-m+1)] \) and \( c_{1}=[(2n-m-1)m]/[(n-m-1)(n-m)] \). Using the similar change of variable in~\Cref{lemma:H1_expectation} and~\Cref{lemma:H2_expectation}, one can verify that \( \mathbb{E}[T^{2}_{10}(m)]=O(m^2n^{-2}h^{-m}) \), \( \mathbb{E}[T^{2}_{10}(m+1)]=O(m^{2}n^{-2}h^{-m-1}) \), \( \mathbb{E}[T^{2}_{20}(m)]=O(m^2n^{-2}h^{-m}) \) and \( \mathbb{E}[T^{2}_{20}(m+1)]=O(m^{2}n^{-2}h^{-m-1}) \). Hence, by Chebyshev inequality, we have
    \begin{align*}
        T_{10}(m) & = O(mn^{-1}h^{2})+O_{p}(mn^{-1}h^{-\frac{m}{2}}), \\
        T_{20}(m) & = c_{2}(\tau^{m}-1)+ O(mn^{-1}h)+O_{p}(mn^{-1}h^{-\frac{m}{2}}), \\
        T_{10}(m+1) & = O(mn^{-1}h^{2})+O_{p}(mn^{-1}h^{-\frac{m+1}{2}}), \\
        T_{20}(m+1) & = c_{1}(\tau^{m+1}-1)+ O(mn^{-1}h)+O_{p}(mn^{-1}h^{-\frac{m+1}{2}}),
    \end{align*}
    which immediately completes the proof comparing with order \( n^{-1/2}h^{-(m+1)/2} \).
\end{proof}
\begin{lemma}\label{lemma:remainder_term_in_I_ff}
    Given \( \mathbb{H}_{0} \), under \textcolor{darkblue}{Assumptions}~\ref{assump:kernel_assumption} and~\ref{assump:density_assumption}, if \(
    nh^{m}/\log{n}\rightarrow\infty \), \( h\rightarrow 0 \)
    and \( m<M \). The order of the remainder term in
    equation~\eqref{eq:Inm_f_hat_f} is
    \begin{equation*}
        O_{p}\left(
        n^{-3/2} h^{-3m/2}{(\log{n})}^{1/2}
        +m^{2}h^{6} \right).
    \end{equation*}
\end{lemma}
\begin{proof}[Proof of~\Cref{lemma:remainder_term_in_I_ff}.]
    Under~\textcolor{darkblue}{Assumptions}~\ref{assump:kernel_assumption} and~\ref{assump:density_assumption}, we can obtain the uniform rates of convergence for multivariate kernel density estimator~\citep[e.g.,][p.30--32]{liNonparametricEconometricsTheory2007}. For \( \mathbf{z}\in \mathbbm{I}^{m} \), it follows that
    \begin{equation}\label{eq:bias_density_uniform}
        \sup_{\mathbf{z}\in \mathbbm{I}^{m}} \left\vert \hat{g}(\mathbf{z})-g(\mathbf{z}) \right\vert=O_{p}\left( \frac{{(\log{n})}^{1/2}}{n^{1/2}h^{m/2}} + mh^{2}\right),
    \end{equation}
    and
    \begin{equation}\label{eq:bias_square_expectation_density_uniform}
        \sup_{\mathbf{z}\in \mathbbm{I}^{m}} \mathbb{E} \left \{ {\left[ \hat{g}(\mathbf{z})-g(\mathbf{z}) \right]}^{2} \right \}   =O\left( n^{-1}h^{-m}+mh^{4} \right).
    \end{equation}
    Using~\eqref{eq:bias_density_uniform} and~\eqref{eq:bias_square_expectation_density_uniform}, we have
    \begin{equation}\label{eq:fhat_f_abs_cubic}
        \begin{aligned}
            & \frac{1}{n}\sum_{i=1}^{n}{\left\vert \hat{g}\left( \mathbf{x}_{i;m} \right)- g\left( \mathbf{x}_{i;m} \right) \right\vert}^{3} \\
            & \quad\leq \sup_{\mathbf{x}_{i;m}\in \mathbbm{I}^{m}} \left\vert \hat{g}\left( \mathbf{x}_{i;m} \right)-g\left( \mathbf{x}_{i;m} \right) \right\vert \left \{ \frac{1}{n}\sum_{i=1}^{n} {\left[ \hat{g}\left( \mathbf{x}_{i;m} \right)-g\left( \mathbf{x}_{i;m} \right) \right]}^{2}\right \}, \\
            & \quad=O_{p}\left( \frac{{(\log{n})}^{1/2}}{n^{3/2}h^{3m/2}} + m^{2}h^{6} \right).
        \end{aligned}
    \end{equation}
    Then by~\eqref{eq:fhat_f_abs_cubic} and the inequality \( \left\vert \log(1+x)-x+\frac{1}{2}x^2 \right\vert \leq {\left\vert x \right\vert}^3\), obviously we have
    \begin{equation}\label{eq:fhat_f_abs_cubic_S}
        \begin{aligned}
            & \left\vert \hat{I}_{m}( \hat{g}, g )- \hat{W}_{1S}(m)  +\frac{1}{2}\hat{W}_{2S}(m) \right\vert \\
            & \quad\leq \frac{1}{n}\sum_{i\in S(m)}{\left\vert \frac{\hat{g}\left( \mathbf{x}_{i;m} \right)-g\left( \mathbf{x}_{i;m} \right)}{g\left( \mathbf{x}_{i;m} \right)} \right\vert}^{3,} \\
            & \quad =O_{p}\left( \frac{{(\log{n})}^{1/2}}{n^{3/2}h^{3m/2}} + m^{2}h^{6} \right),
        \end{aligned}
    \end{equation}
    where
    \begin{equation*}
        \hat{W}_{1S}(m)=\frac{1}{n}\sum_{i\in S(m)}\left[\frac{\hat{g}\left( \mathbf{x}_{i;m} \right)-g\left( \mathbf{x}_{i;m} \right)}{g\left( \mathbf{x}_{i;m} \right)} \right],
    \end{equation*}
    and
    \begin{equation*}
        \hat{W}_{2S}(m)=\sum_{i\in S(m)} {\left[\frac{\hat{g}\left( \mathbf{x}_{i;m} \right)-g\left( \mathbf{x}_{i;m} \right)}{g\left( \mathbf{x}_{i;m} \right)} \right]}^{2}.
    \end{equation*}
    Note that the definitions of \( \hat{W}_{1}(m) \) and \( \hat{W}_{2}(m) \) include the summation of \( n \) observations. The difference is
    \begin{align}\nonumber
        & \left[ \hat{W}_{1S}(m)-\frac{1}{2}\hat{W}_{2S}(m) \right]-\left[ \hat{W}_{1}(m)-\frac{1}{2}\hat{W}_{2}(m) \right] \\ \label{eq:not_in_S}
        & \quad= O_{p}\left[ \frac{1}{n}\sum_{i=1}^{n}P\left( i\not\in S(m) \right) \right], \\ \label{eq:expectation_cubic}
        & \quad=O_{p}\left[ \frac{1}{n}\sum_{i=1}^{n}\mathbb{E}{\left\vert \frac{\hat{g}\left( \mathbf{x}_{i;m} \right)-g\left( \mathbf{x}_{i;m} \right)}{g\left( \mathbf{x}_{i;m} \right)} \right\vert}^{3} \right], \\ \label{eq:final_order_for_notin_S}
        & \quad=O_{p}\left( \frac{{(\log{n})}^{1/2}}{n^{3/2}h^{3m/2}} + m^{2}h^{6} \right).
    \end{align}
    Step~\eqref{eq:not_in_S} to step~\eqref{eq:expectation_cubic} is based on the fact that
    \begin{equation*}
        \begin{aligned}
            P\left( i\not\in S(m) \right) & = P\left( \hat{g}\left( \mathbf{x}_{i;m} \right)\leq 0 \right), \\
            & \leq P\left[ \left\vert \hat{g}\left( \mathbf{x}_{i;m} \right)-g\left( \mathbf{x}_{i;m}\right) \right\vert > g\left( \mathbf{x}_{i;m}\right) \right], \\
            & \leq  \mathbb{E}{\left\vert \frac{\hat{g}\left( \mathbf{x}_{i;m} \right)-g\left( \mathbf{x}_{i;m} \right)}{g\left( \mathbf{x}_{i;m} \right)} \right\vert}^{3}.
        \end{aligned}
    \end{equation*}
    Combining equations~\eqref{eq:fhat_f_abs_cubic_S} and~\eqref{eq:final_order_for_notin_S}, we complete the proof of this lemma.
\end{proof}

\begin{lemma}\label{lemma:Gamma_m}
    Given \textcolor{darkblue}{Assumptions}~\ref{assump:kernel_assumption} and~\ref{assump:density_assumption}, if \( h\rightarrow 0 \)
    and \( m<M \). Then under \( \mathbb{H}_{0} \), we have \( P(\lim_{n\rightarrow\infty}\hat{\Gamma}(m)=0)=1 \).
\end{lemma}

\begin{proof}[Proof of~\Cref{lemma:Gamma_m}]
    Firstly, we give a similar result for the univariate kernel; then we extend this result to the multivariate kernel. For any \( x,y\in [0, 1] \), denote \( \gamma_{1}(x,y)=\int_{0}^{1}[ K_{h}^{J}(x^{*},y)- \int_{0}^{1}K_{h}^{J}(x^{*},y^{*})g_{1}(y^{*})dy^{*}]dx^{*}/g_{1}(x) \). The numerator of \( \gamma_{1}(x,y) \) includes two terms. The first term can be expanded as
    \begin{equation}\label{eq:univariate_kernel_expanded}
        \begin{aligned}
            & \int_{0}^{1} K_{h}^{J}\left( x^{*}-y \right)\,\mathrm{d}x^{*} \\
            & \quad=  \int_{0}^{h} K_{h}^{J}\left( x^{*}-y \right)\,\mathrm{d}x^{*}+\int_{1-h}^{1} K_{h}^{J}\left( x^{*}-y \right)\,\mathrm{d}x^{*}+ \int_{h}^{1-h} K_{h}\left( x^{*}-y \right)\,\mathrm{d}x^{*},
        \end{aligned}
    \end{equation}
    when \( n\rightarrow \infty \), then \( h\rightarrow 0 \) such that \( y/h \rightarrow \infty \) and \( (1-y)/h \rightarrow \infty \) for any \( y\in(0,1) \), see Appendix A in~\citet{hongAsymptoticDistributionTheory2005}. Using \( K(\cdot) \) having bounded support \( [-1,1] \) and change of variable, when \( n \) is sufficient large, the first and second terms in equation~\eqref{eq:univariate_kernel_expanded} are zero, and the third term is 1. When \( n \) is sufficiently large, the term
    \begin{equation}\label{eq:second_term_gamma}
        \begin{aligned}
            & \int_0^{1}\int_{0}^{1}K_{h}^{J}(x^{*},y^{*})g_{1}(y^{*})dy^{*}dx^{*} \\
            & \quad= \int_0^{1}g_{1}(y^{*})\int_{0}^{1}K_{h}^{J}(x^{*},y^{*})dx^{*}dy^{*} \\
            & \quad = \int_0^{1}g_{1}(y^{*})\cdot 1dy^{*}=1,
        \end{aligned}
    \end{equation}
    as well. Therefore, when \( n \) is sufficiently large, \( \gamma_{1}(x,y) =0  \) with probability 1. Recalling that \( \mathcal{K}^{(m)}_{h}(\cdot) \) is a multiplicative kernel, we can easily extend this result to a multivariate case. It follows that for sufficiently large \( n \),
    \begin{equation*}
        P\left[ \gamma_{m}\left( \mathbf{x}_{i;m}, \mathbf{x}_{j;m} \right)=0 \right]=1,
    \end{equation*}
    this completes the proof.
\end{proof}
\begin{lemma}\label{lemma:D_n_m}
    Given \textcolor{darkblue}{Assumptions}~\ref{assump:kernel_assumption} and~\ref{assump:density_assumption},
    if \( nh^{m}\rightarrow \infty \), \(
    h\rightarrow 0  \)
    and
    \( m<M \), then under \( \mathbb{H}_{0} \),
    \begin{equation}\label{eq:D_n_m}
        \hat{D}(m)=2\mathbb{E}A^{2}_{m}(\mathbf{z}_{1},\mathbf{z}_{2})+O_{p}(n^{-1/2}m^{1/2}h^{-m} ),
    \end{equation}
    where \( \mathbf{z}_{1}, \mathbf{z}_{2} \) have no overlap variable.
\end{lemma}
\begin{proof}[Proof of~\Cref{lemma:D_n_m}]
    Let
    \begin{align*}
        \phi_{m}(\mathbf{z}_{1},\mathbf{z}_{2}) & = h^{m}D_{m}(\mathbf{z}_{1},\mathbf{z}_{2})=  h^{m}\left[ A^{2}_{m}(\mathbf{z}_{1},\mathbf{z}_{2})+A^{2}_{m}(\mathbf{z}_{2},\mathbf{z}_{1}) \right],
    \end{align*}
    then we have \( \hat{\phi}(m)=h^{m}\hat{D}(m) \) and
    \begin{equation*}
        \begin{aligned}
            \phi_{m0} & = \int_{\mathbbm{I}^{m}}\int_{\mathbbm{I}^{m}}\phi_{m}\left( \mathbf{z}_{1},
            \mathbf{z}_{2} \right)
            g(\mathbf{z}_{1})g(\mathbf{z}_{2})\,\mathrm{d}\mathbf{z}_{1}\,\mathrm{d}\mathbf{z}_{2} \\
            & = 2h^{m}\mathbb{E}A^{2}_{m}\left( \mathbf{z}_{1},
            \mathbf{z}_{2} \right).
        \end{aligned}
    \end{equation*}
    We note that \( g(\mathbf{z}) \) is bounded away from zero by~\ref{assump:density_assumption}, then it follows that
    \begin{equation*}
        \begin{aligned}
            \mathbb{E}\phi^{2}_{m} & \leq h^{2m}C \int_{\mathbbm{I}^{m}}\int_{\mathbbm{I}^{m}}A^{4}_{m}\left( \mathbf{z}_{1},
            \mathbf{z}_{2}\right)\,\mathrm{d}\mathbf{z}_{1}\,\mathrm{d}\mathbf{z}_{2} = O(1),
        \end{aligned}
    \end{equation*}
    because \( \mathbb{E}A^{2}_{m}\left( \mathbf{z}_{1}, \mathbf{z}_{2}\right) =O(h^{-m})\), Jensen's inequality and Cauchy-Schwarz inequality. Using the same way, one can also verify that \( \mathbb{E} \phi_{m1}^{2}( \mathbf{x}_{i;m})-\phi_{m0}^{2}=O(1)\), so \( h^{m}D_{m}(\mathbf{z}_{1},\mathbf{z}_{2}) \) satisfies the conditions in~\Cref{lemma:second_U_statistics}, immediately have the result~\eqref{eq:D_n_m}.
\end{proof}
\begin{lemma}\label{lemma:H_2n_m}
    Given \textcolor{darkblue}{Assumptions}~\ref{assump:kernel_assumption} and~\ref{assump:density_assumption}, if \( nh^{m}\rightarrow \infty \), \(
    h\rightarrow 0  \)
    and \( m<M \), then under \( \mathbb{H}_{0} \),
    \begin{equation}\label{eq:H_2n_m}
        \tilde{H}_{2}(m)=3\hat{H}_{2}(m) +O_{p}(n^{-3/2}m^{3/2}h^{-m} ),
    \end{equation}
    where \( \hat{H}_{2}(m) \) is defined in equation~\eqref{eq:H_hat_2nm}.
    \begin{equation}\label{eq:H_hat_2nm}
        \hat{H}_{2}(m)=\binom{n}{2}^{-1} \sum\nolimits_{j=2}^{n}\sum\nolimits_{i=1}^{j-1}H_{2m}\left( \mathbf{x}_{i;m},\mathbf{x}_{j;m} \right).
    \end{equation}
\end{lemma}
\begin{proof}[Proof of~\Cref{lemma:H_2n_m}]
    Let \begin{equation*}
        \begin{aligned}
            \phi_{m}(\mathbf{z}_{1},\mathbf{z}_{2},\mathbf{z}_{3}) & = h^{m}\tilde{H}_{2m}(\mathbf{z}_{1},\mathbf{z}_{2},\mathbf{z}_{3}), \\
            & = h^{m}\left[ A_{m}(\mathbf{z}_{1},\mathbf{z}_{2})A_{m}(\mathbf{z}_{1},\mathbf{z}_{3}) + A_{m}(\mathbf{z}_{2},\mathbf{z}_{3})A_{m}(\mathbf{z}_{2},\mathbf{z}_{1})\right. \\
            & \left.\qquad+A_{m}(\mathbf{z}_{3},\mathbf{z}_{1})A_{m}(\mathbf{z}_{3},\mathbf{z}_{2}) \right], \\
            & \qquad
        \end{aligned}
    \end{equation*}
    then we have
    \begin{equation*}
        \hat{\phi}(m)=h^{m}\binom{n}{3}^{-1}\sum_{k=3}^{n}\sum_{j=2}^{k-1}\sum_{i=1}^{j-1}\tilde{H}_{2m}\left( \mathbf{x}_{k;m}, \mathbf{x}_{i;m},\mathbf{x}_{j;m} \right) =h^{m}\tilde{H}_{2}(m),
    \end{equation*}
    and \( \phi_{m2}(\mathbf{z}_{1},\mathbf{z}_{2})=h^{m}H_{2m}(\mathbf{z}_{1},\mathbf{z}_{2}) \)
    based on the fact \( \int_{\mathbbm{I}^{m}}A_{m}(\mathbf{z}_{1},\mathbf{z}_{2})g(\mathbf{z}_{2}) \,\mathrm{d}\mathbf{z}_{2}=0 \). Furthermore, we can easily verify that \( \mathbb{E} \phi_{m}^{2}( \mathbf{x}_{i;m}, \mathbf{x}_{j;m},\mathbf{x}_{k;m} )=O(1)\), then  by~\Cref{lemma:third_U_statistics}, we immediately obtain equation~\eqref{eq:H_2n_m} which completes the proof.
\end{proof}
\begin{lemma}\label{lemma:w_22_hat}
    Given \textcolor{darkblue}{Assumptions}~\ref{assump:kernel_assumption} and~\ref{assump:density_assumption}, if \( nh^{m}\rightarrow \infty \), \(
    h\rightarrow 0  \)
    and \( m<M \), then under \( \mathbb{H}_{0} \),
    \begin{equation*}
        \hat{W}_{22}(m)=\mathbb{E}B_{m}^{2}\left( \mathbf{x}_{1;m} \right)+O_{p}(n^{-1/2}h^{4}).
    \end{equation*}
\end{lemma}
\begin{proof}[Proof of~\Cref{lemma:w_22_hat}]
    Firstly, \( \hat{W}_{22}(m) \) can be expressed as
    \begin{equation*}
        \begin{aligned}
            \hat{W}_{22}(m) & = \frac{1}{n}\sum_{i=1}^{n}B_{m}^{2}\left( \mathbf{x}_{i;m} \right), \\
            & = \mathbb{E}B_{m}^{2}\left( \mathbf{x}_{1;m} \right) + \frac{1}{n}\sum_{i=1}^{n}\left[ B_{m}^{2}\left( \mathbf{x}_{i;m} \right)- \mathbb{E}B_{m}^{2}\left( \mathbf{x}_{1;m} \right)\right].
        \end{aligned}
    \end{equation*}
    By~\Cref{lemma:order_an_bn}, we have \( \mathbb{E} B_{m}^{4}( \mathbf{x}_{i;m} )=O(h^{8})\). When \( i\geq m \), \( \mathbf{x}_{i;m} \) is independent of \( \mathbf{x}_{1;m} \) and \( m \) is bounded by \( M \), by Chebyshev's inequality, we immediately have \( \hat{W}_{22}(m)=\mathbb{E}B_{m}^{2}( \mathbf{x}_{1;m} )+O_{p}(n^{-1/2}h^{4}) \), this completes the proof.
\end{proof}
\begin{lemma}\label{lemma:w_23_hat}
    Given \textcolor{darkblue}{Assumptions}~\ref{assump:kernel_assumption} and~\ref{assump:density_assumption}, if \( nh^{m}\rightarrow \infty \), \(
    h\rightarrow 0  \)
    and \( m<M \), then under \( \mathbb{H}_{0} \),
    \begin{equation*}
        \hat{W}_{23}(m)=2\hat{C}(m)+O_{p}(n^{-1}mh^{2-m/2}),
    \end{equation*}
    where \( \hat{C}(m) \) is defined in equation~\eqref{eq:c_hat_nm}.
\end{lemma}

\begin{proof}[Proof of~\Cref{lemma:w_23_hat}]
    Define a new symmetric function
    \begin{equation*}
        \tilde{C}_{m}(\mathbf{z}_{1},\mathbf{z}_{2})=A_{m}(\mathbf{z}_{1},\mathbf{z}_{2})B_{m}(\mathbf{z}_{1})+A_{m}(\mathbf{z}_{2},\mathbf{z}_{1})B_{m}(\mathbf{z}_{2}),
    \end{equation*}
    then
    \begin{align*}
        \hat{W}_{23}(m) & = 2\frac{1}{n}\sum_{i=1}^{n}\left[ \frac{\hat{g}\left( \mathbf{x}_{i;m} \right)-\bar{g}\left( \mathbf{x}_{i;m} \right)}{g\left( \mathbf{x}_{i;m} \right)} \right]\left[ \frac{\bar{g}\left( \mathbf{x}_{i;m} \right)-g\left( \mathbf{x}_{i;m} \right)}{g\left( \mathbf{x}_{i;m} \right)} \right], \\
        & =  \binom{n}{2}^{-1} \sum_{j=2}^{n}\sum_{i=1}^{j-1}\tilde{C}_{m}\left( \mathbf{x}_{i;m}, \mathbf{x}_{j;m} \right).
    \end{align*}
    Let \( \phi_{m}(\mathbf{z}_{1},\mathbf{z}_{2})=\tilde{C}_{m}\left( \mathbf{z}_{1}, \mathbf{z}_{2} \right)\), then \( \phi_{m0}=0\) and
    \begin{align*}
        \phi_{m1}(\mathbf{z}) & = \int_{\mathbbm{I}^{m}} \phi_{m}(\mathbf{z},\mathbf{z}_{2})g(\mathbf{z}_{2})\,\mathrm{d}\mathbf{z}_{2}=\int_{\mathbbm{I}^{m}} \phi_{m}(\mathbf{z}_{2},\mathbf{z})g(\mathbf{z}_{2})\,\mathrm{d}\mathbf{z}_{2}, \\
        & = \int_{\mathbbm{I}^{m}} A_{m}(\mathbf{z}_{2},\mathbf{z})B_{m}(\mathbf{z}_{2})g(\mathbf{z}_{2})\,\mathrm{d}\mathbf{z}_{2}.
    \end{align*}
    As discussion in~\Cref{lemma:D_n_m}, one can verify that \( \mathbb{E}[ \phi^{2}_{m} ( \mathbf{x}_{i;m} ,\mathbf{x}_{j;m}) ]=O(c_{m}^{2})  \), where \( c_{m}=h^{2-m/2} \). By~\Cref{lemma:second_U_statistics}, we have
    \begin{align*}
        \hat{W}_{23}(m) & = \hat{\phi}(m)= \frac{2}{n}\sum_{i=1}^{n} \phi_{m1}\left( \mathbf{x}_{i;m} \right)  + O_{p}(n^{-1}mh^{2-m/2}), \\
        & =2\hat{C}(m)+O_{p}(n^{-1}mh^{2-m/2}),
    \end{align*}
    which completes the proof.
\end{proof}
\begin{lemma}\label{lemma:A_2_nm}
    Given \textcolor{darkblue}{Assumptions}~\ref{assump:kernel_assumption} and~\ref{assump:density_assumption}, under \( \mathbb{H}_{0} \), we have
    \begin{equation*}
        (n-1)^{-1}\left[ \mathbb{E}A^{2}_{(m+1)}(\mathbf{z}_{1},\mathbf{z}_{2})-\mathbb{E}A^{2}_{m}(\mathbf{y}_{1},\mathbf{y}_{2}) -\mathbb{E}a^{2}(z_{1m},z_{2m})\right]=d_{0}+O(n^{-1}h^{-m}),
    \end{equation*}
    where \( \kappa = \int_{-1}^{1} K^{2}(u)\,\mathrm{d}u \) and \( d_{0}={(n-1)}^{-1}\kappa^{m+1}h^{-(m+1)} \).
\end{lemma}
\begin{proof}[Proof of~\Cref{lemma:A_2_nm}]
    Let \( \kappa_{1}= 2\int_{0}^{1}\int_{-1}^{\rho}k_{\rho}^{2}(u)\,\mathrm{d}u\,\mathrm{d}\rho-2\kappa -1\), then \( \mathbb{E}a^{2}(z_{1m},z_{2m})=\kappa h^{-1}+\kappa_{1}+O(h) \), by equation~\eqref{eq:A_expand} and \( \mathbb{H}_{0} \), we immediately obtain the desired result.
\end{proof}
\begin{lemma}\label{lemma:b_part}
    Given \( \mathbb{H}_{0} \) and \( 1\leq m< M \),
    \begin{equation*}
        \begin{aligned}
            & 2\left[ \hat{B}(m+1)- \hat{B}(m)-\hat{b}_{m}\right] \\
            & \quad- \left[ \mathbb{E}B^{2}_{(m+1)}(\mathbf{z}_{1})-\mathbb{E}B^{2}_{m}(\mathbf{y}_{1}) -\mathbb{E}b^{2}(z_{1m})\right] \\
            & =O(h^{6} )+O_{p}(n^{-1/2}h^{4}).
        \end{aligned}
    \end{equation*}
\end{lemma}
\begin{proof}[Proof of~\Cref{lemma:b_part}]
    By equation~\eqref{eq:B_expand}, \( [ \hat{B}(m+1)- \hat{B}(m)-\hat{b}_{m}] \) can be expressed as
    \begin{align*}
        & \left[ \hat{B}(m+1)-
            \hat{B}(m)-\hat{b}_{m}\right] \\
        & \quad = \mathbb{E}B_{m}(\mathbf{y}_{1})b(z_{1m})+\frac{1}{n}\sum_{i=1}^{n}\left[ B_{m}\left(\mathbf{x}_{i;m}\right)b(x_{i+m})-\mathbb{E}B_{m}(\mathbf{y}_{1})b(z_{1m}) \right],
    \end{align*}
    According to Cauchy-Schwarz inequality, \( \mathbb{E}{ [ B_{m}\left( \mathbf{x}_{i;m}\right)b(x_{i+m})] }^{2}=O(h^{8}) \), so the second term
    \begin{equation*}
        n^{-1}\sum_{i=1}^{n} \left[ B_{m}\left(\mathbf{x}_{i;m}\right)b(x_{i+m})-\mathbb{E}B_{m}(\mathbf{y}_{1})b(z_{1m})\right]=O_{p}(n^{-1/2}h^{4} ),
    \end{equation*}
    by Markov inequality and Chebyshev inequality. Furthermore, by equation~\eqref{eq:B_expand} and~\Cref{lemma:order_an_bn}, after simple calculations, we have
    \begin{align*}
        & \mathbb{E}B^{2}_{(m+1)}(\mathbf{z}_{1})-\mathbb{E}B^{2}_{m}(\mathbf{y}_{1}) -\mathbb{E}b^{2}(z_{1m}) \\
        & \quad =2\mathbb{E}B_{m}(\mathbf{y}_{1})b(z_{1m}) +2\mathbb{E}B^{2}_{m}(\mathbf{y}_{1})b(z_{1m})+\mathbb{E}B_{m}(\mathbf{y}_{1})b^{2}(z_{1m})+O(h^{8}), \\
        & \quad=2\mathbb{E}B_{m}(\mathbf{y}_{1})b(z_{1m})+O(h^{6}),
    \end{align*}
    which immediately completes the proof.
\end{proof}
\begin{lemma}\label{lemma:c_part}
    Given \( \mathbb{H}_{0} \) and \( 1\leq m< M \),
    \begin{equation*}
        \hat{C}(m+1)- \hat{C}(m)-\hat{c}_{m}=O_{p}(n^{-1/2}h^{4}),
    \end{equation*}
    where \( \hat{C}(m+1)= n^{-1}\sum_{i=1}^{n}\breve{C}_{m+1}(
    \mathbf{x}_{i;m+1} )\), \( \hat{C}(m)=
    n^{-1}\sum_{i=1}^{n}\breve{C}_{m}( \mathbf{x}_{i;m}
    )\), \( \hat{c}_{m}= n^{-1}\sum_{i=1}^{n}\breve{c}(x_{i+m})\), \( \mathbf{z}_{1}={(z_{10},\ldots,z_{1(m-1)},z_{1m})}^{\top}={(
    \mathbf{y}_{1}^{\top},z_{1m} )}^{\top} \),
    \begin{align*}
        \breve{C}_{m+1}\left(
        \mathbf{x}_{i;m+1} \right) & = \int_{\mathbf{z}_{1}\in \mathbb{I}^{m+1}}A_{(m+1)}\left( \mathbf{z}_{1} ,\mathbf{x}_{i;m+1}\right) B_{(m+1)}(\mathbf{z}_{1})f(\mathbf{z}_{1})\,\mathrm{d}\mathbf{z}_{1}, \\
        \breve{C}_{m}\left( \mathbf{x}_{i;m}
        \right) & = \int_{\mathbf{y}_{1}\in \mathbb{I}^{m}}A_{m}\left( \mathbf{y}_{1} ,\mathbf{x}_{i;m}\right) B_{m}(\mathbf{y}_{1})g(\mathbf{y}_{1})\,\mathrm{d}\mathbf{y}_{1}, \\
        \breve{c}(x_{i+m}) & =\int_{0}^{1}a\left( z_{1m} ,x_{i+m}\right) b(z_{1m})g_{1}(z_{1m})\,\mathrm{d}z_{1m}.
    \end{align*}
\end{lemma}
\begin{proof}[Proof of~\Cref{lemma:c_part}]
    Let
    \begin{align*}
        \bar{K}^{J}_{h}(z_{1m}) & = \int_{0}^{1}
        K_{h}^{J}(z_{1m},z)g_{1}(z)\,\mathrm{d}z, \\
        \bar{\mathcal{K}}^{(m)}_{h}(\mathbf{y}_{1}) & =\int_{\mathbbm{I}^{m}}
        \mathcal{K}_{h}^{(m)}(\mathbf{y}_{1},\mathbf{y})g(\mathbf{y})\,\mathrm{d}\mathbf{y}, \\
        \bar{\mathcal{K}}^{(m+1)}_{h}(\mathbf{z}_{1}) & =\int_{\mathbbm{I}^{m+1}}
        \mathcal{K}_{h}^{(m+1)}(\mathbf{z}_{1},\mathbf{z})f(\mathbf{z})\,\mathrm{d}\mathbf{z},
    \end{align*}
    and \( \psi_{1}(z_{1m})= \bar{K}_{h}^{J}(z_{1m})/g_{1}(x_{i+m})\), \( \psi_{1}(z_{1m},x_{i+m})= K_{h}^{J}(z_{1m},x_{i+m})/g_{1}(x_{i+m})\),\\ \( \psi_{m}( \mathbf{y}_{1} ) =\bar{\mathcal{K}}^{(m)}_{h}(\mathbf{y}_{1})/g(\mathbf{y}_{1}) \),\( \psi_{m}( \mathbf{y}_{1}, \mathbf{x}_{i;m} ) =\mathcal{K}_{h}^{(m)}(\mathbf{y}_{1},\mathbf{x}_{i;m})/g(\mathbf{y}_{1}) \). Given the definition of multivariate kernel and \( \mathbb{H}_{0} \), we obtain
    \begin{align*}
        a\left( z_{1m} ,x_{i+m}\right) & = \psi_{1}(z_{1m},x_{i+m})-\psi_{1}(z_{1m}), \\
        A_{m}\left( \mathbf{y}_{1} ,\mathbf{x}_{i;m}\right) & = \psi_{m}\left( \mathbf{y}_{1}, \mathbf{x}_{i;m} \right) -\psi_{m}\left( \mathbf{y}_{1} \right), \\
        A_{(m+1)}\left( \mathbf{z}_{1} ,\mathbf{x}_{i;m+1}\right) & =\psi_{m}\left( \mathbf{y}_{1}, \mathbf{x}_{i;m} \right)\psi_{1}(z_{1m},x_{i+m})-\psi_{m}\left( \mathbf{y}_{1} \right)\psi_{1}(z_{1m}).
    \end{align*}
    Using equation~\eqref{eq:B_expand} and \( f(\mathbf{z}_{1})=g(\mathbf{y}_{1})g_{1}(z_{1m}) \), we can separately  write \( \breve{c}(x_{i+m}) \), \( \breve{C}_{m} ( \mathbf{x}_{i;m} )\) and \( \breve{C}_{m+1} ( \mathbf{x}_{i;m+1})\) as
    \begin{equation*}
        \begin{aligned}
            \breve{c}(x_{i+m}) & = \int_{0}^{1}\psi_{1}\left( z_{1m} ,x_{i+m}\right) b(z_{1m})g_{1}(z_{1m})\,\mathrm{d}z_{1m} \\
            & \qquad-\int_{0}^{1}\psi_{1}\left( z_{1m}\right) b(z_{1m})g_{1}(z_{1m})\,\mathrm{d}z_{1m}, \\
            & = \breve{c}_{1}(x_{i+m})-\breve{c}_{2},
        \end{aligned}
    \end{equation*}
    \begin{equation*}
        \begin{aligned}
            \breve{C}_{m} \left( \mathbf{x}_{i;m} \right) & = \int_{\mathbf{y}_{1}\in \mathbbm{I}^{m}}\psi_{m}\left( \mathbf{y}_{1} ,\mathbf{x}_{i;m}\right) B_{m}(\mathbf{y}_{1})g(\mathbf{y}_{1})\,\mathrm{d}\mathbf{y}_{1} \\
            & \qquad-\int_{\mathbf{y}_{1}\in \mathbbm{I}^{m}}\psi_{m}\left( \mathbf{y}_{1} \right) B_{m}(\mathbf{y}_{1})g(\mathbf{y}_{1})\,\mathrm{d}\mathbf{y}_{1}, \\
            & = \breve{C}_{1} \left( \mathbf{x}_{i;m}\right)-\breve{C}_{2},
        \end{aligned}
    \end{equation*}
    \begin{equation*}
        \begin{aligned}
            \breve{C}_{m+1} \left( \mathbf{x}_{i;m+1} \right) & =\breve{C}_{1} \left( \mathbf{x}_{i;m}\right)\breve{c}_{1}(x_{i+m})- \breve{C}_{2}\breve{c}_{2} \\
            & \qquad +\breve{C}_{1} \left( \mathbf{x}_{i;m}\right)\int_{0}^{1}K_{h}^{J}\left( z_{1m} ,x_{i+m}\right) \,\mathrm{d}z_{1m} \\
            & \qquad -\breve{C}_{2}\int_{0}^{1}\bar{K}_{h}^{J}\left( z_{1m} \right) \,\mathrm{d}z_{1m} \\
            & \qquad+\breve{c}_{1}(x_{i+m})\int_{\mathbf{y}_{1}\in \mathbbm{I}^{m}}\mathcal{K}_{h}^{(m)}\left( \mathbf{y}_{1} ,\mathbf{x}_{i;m}\right)\,\mathrm{d}\mathbf{y}_{1} \\
            & \qquad-\breve{c}_{2}\int_{\mathbf{y}_{1}\in \mathbbm{I}^{m}}\bar{\mathcal{K}}_{h}^{(m)}\left( \mathbf{y}_{1} \right)\,\mathrm{d}\mathbf{y}_{1},
        \end{aligned}
    \end{equation*}
    then we have
    \begin{equation*}
        \begin{aligned}
            & \breve{C}_{m+1} \left( \mathbf{x}_{i;m+1} \right)-\breve{C}_{m} \left( \mathbf{x}_{i;m} \right)-\breve{c}(x_{i+m}) \\ &\quad=\breve{C}_{1} \left( \mathbf{x}_{i;m}\right)\breve{c}_{1}(x_{i+m})- \breve{C}_{2}\breve{c}_{2}\\
            & \qquad +\breve{C}_{1} \left( \mathbf{x}_{i;m}\right)\left[ \int_{0}^{1}K_{h}^{J}\left( z_{1m} ,x_{i+m}\right) \,\mathrm{d}z_{1m}-1 \right] \\
            & \qquad -\breve{C}_{2}\left[ \int_{0}^{1}\bar{K}_{h}^{J}\left( z_{1m} \right) \,\mathrm{d}z_{1m}-1 \right] \\
            & \qquad+\breve{c}_{1}(x_{i+m})\left[ \int_{\mathbf{y}_{1}\in \mathbbm{I}^{m}}\mathcal{K}_{h}^{(m)}\left( \mathbf{y}_{1} ,\mathbf{x}_{i;m}\right)\,\mathrm{d}\mathbf{y}_{1}-1 \right] \\
            & \qquad-\breve{c}_{2}\left[ \int_{\mathbf{y}_{1}\in \mathbbm{I}^{m}}\bar{\mathcal{K}}_{h}^{(m)}\left( \mathbf{y}_{1} \right)\,\mathrm{d}\mathbf{y}_{1}-1 \right] \\
            & \quad =\breve{C}_{1} \left( \mathbf{x}_{i;m}\right)\breve{c}_{1}(x_{i+m})- \breve{C}_{2}\breve{c}_{2}+\sum_{s=1}^{4}\delta_{s}.
        \end{aligned}
    \end{equation*}
    Firstly, for terms \( \delta_{s},s=1,2,3,4 \), we prove \( \delta_{1}=0, \delta_{2}=0,  a.e., \) for univariate kernel when \( n \) is sufficiently large, then we extend this result to the multivariate kernel. For any \(y\in [0, 1] \), by equations~\eqref{eq:univariate_kernel_expanded} and~\eqref{eq:second_term_gamma} in~\Cref{lemma:Gamma_m}, we have \( \int_{0}^{1} K_{h}^{J}\left( x,y \right)\,\mathrm{d}x =1\) and \( \int_0^{1}\bar{K}_{h}^{J}(x)dx=1 \) almost surely.
    Therefore, when \( n \) is sufficiently large, \( \delta_{1}=0, \delta_{2}=0  \) almost everywhere. Recalling that \( \mathcal{K}^{(m)}_{h}(\cdot) \) is a multiplicative kernel, we can easily extend this result to a multivariate case. It follows that for sufficiently large \( n \), \( \sum_{s=1}^{4}\delta_{s}=0 \) almost surely. Therefore, \( \hat{C}(m+1)- \hat{C}(m)-\hat{c}_{m}=n^{-1}\sum_{i=1}^{n} \breve{C}_{1} ( \mathbf{x}_{i;m})\breve{c}_{1}(x_{i+m})- \breve{C}_{2}\breve{c}_{2}\). We also notice that \( \mathbb{E}[ \breve{C}_{1} ( \mathbf{x}_{i;m})\breve{c}_{1}(x_{i+m}) ]= \breve{C}_{2}\breve{c}_{2}\) and \( \mathbb{E}{[ \breve{C}_{1} ( \mathbf{x}_{i;m})\breve{c}_{1}(x_{i+m}) ]}^{2}=O(h^8) \) because of~\Cref{lemma:order_an_bn}. Hence, by Chebyshev inequality, \( \hat{C}(m+1)- \hat{C}(m)-\hat{c}_{m}=O_{p}(n^{-1/2}h^{4}) \), this completes the proof.
\end{proof}

\begin{proof}[Proof of~\Cref{thm:I_m}]
    Following the Theorems A.6 -- A.9 in~\citet[][]{hongAsymptoticDistributionTheory2005}, one can extend their theory to multivariate U-statistics with bounded \( m \).~\citet[]{hongAsymptoticDistributionTheory2005} discussed the pair variables \( Z_{jt}={(X_{t},X_{t-j})}^{\top} \) and \( j=o(n) \). Theorem A.6 constructs a new \( 2j \)-dependent process to show the dependent part of U-statistics is negligible. They employed a martingale difference sequence in Theorem A.7  so that the U-statistics can be projected on it. Theorem A.7 implies that one can apply the central limit theorem to the martingale difference sequence according to \citet[][]{brownMartingaleCentralLimit1971}'s theorem if two conditions in Theorem A.9 are satisfied. Finally, by Slusky's theorem, Brown's theorem and Theorems A.6 -- A.9, the central limit theorem of U-statistics is completed. Analogue to \citet[][]{hongAsymptoticDistributionTheory2005}'s idea but more tedious than that, for bounded \( m \),~\eqref{eq:I_m_clt} holds as well.
\end{proof}

\section{Lemmas for The Second and Third Order U-statistics}\label{sec:lemmas_for_the_second_and_third_u_statistics}
\begin{lemma}\label{lemma:second_U_statistics}
    Let \( \mathbf{Z}_{i;m}={(X_{i},\ldots,X_{i+m-1})}^{\top} \), \( m<M \) and \( \{X_{t}\} \) is \textit{i.i.d.} with CDF \( G_{1}(\cdot) \). Consider a second-order U-statistics
    \begin{equation*}
        \hat{\phi}(m)=\binom{n}{2}^{-1}\sum_{j=2}^{n}\sum_{i=1}^{j-1}\phi_{m}\left( \mathbf{Z}_{i;m}, \mathbf{Z}_{j;m} \right),
    \end{equation*}
    where \( \phi_{m}\left( \cdot, \cdot \right) \) is a kernel function such that \( \phi_{m}\left( \mathbf{z}_{1}, \mathbf{z}_{2} \right) =\phi_{m}\left( \mathbf{z}_{2}, \mathbf{z}_{1} \right) \).  Let
    \begin{equation*}
        \phi_{m0}=\int_{\mathbb{I}^{m}}\int_{\mathbb{I}^{m}}\phi_{m}\left( \mathbf{z}_{1},
        \mathbf{z}_{2} \right)
        \,\mathrm{d}G_{m}(\mathbf{z}_{1})\,\mathrm{d}G_{m}(\mathbf{z}_{2}),
    \end{equation*}
    and \( \phi_{m1}\left( \mathbf{z} \right)=\int_{\mathbb{I}^{m}}\phi_{m}\left( \mathbf{z}, \mathbf{z}_{1} \right) \,\mathrm{d}G_{m}(\mathbf{z}_{1}) \), where \( G_{m}(\mathbf{z})=G_{1}(x_{1})G_{1}(x_{2})\cdots G_{1}(x_{m}) \) and \( \mathbf{z}={(x_{1},\ldots,x_{m})}^{\top} \). Suppose \( \mathbb{E} \phi_{m}^{2}( \mathbf{Z}_{i;m}, \mathbf{Z}_{j;m} )-\phi_{m0}^{2}=O(c_{m}^{2})\) holds, then we have
    \begin{equation*}
        \hat{\phi}(m)=\phi_{m0}+\frac{2}{n}\sum_{i=1}^{n}\left[ \phi_{m1}\left( \mathbf{Z}_{i;m} \right) -\phi_{m0}\right] +O_{p}(mn^{-1}c_{m}).
    \end{equation*}
    If in addition \( \mathbb{E}\phi_{m1}^{2} ( \mathbf{Z}_{i;m} )-\phi_{m0}^{2}\leq C \) and \( c_{m}=O(n^{1/2}) \), then \( \hat{\phi}(m)=\phi_{m0} + O_{p}(m^{1/2}n^{-1/2}) \).
\end{lemma}

\begin{proof}[Proof of \textcolor{darkblue}{Lemma}~\ref{lemma:second_U_statistics}]
    \textcolor{darkblue}{Lemma}~\ref{lemma:second_U_statistics} is the version of Lemma B.1 in~\citet{hongAsymptoticDistributionTheory2005}. For \( m \)-consecutive variables, we still use the same notations as~\citet{hongAsymptoticDistributionTheory2005}'s. Denote \( \tilde{\phi}_{m}( \mathbf{z}_{1}, \mathbf{z}_{2} ) =\phi_{m}( \mathbf{z}_{1}, \mathbf{z}_{2} )-\phi_{m1}( \mathbf{z}_{1} )-\phi_{m1}( \mathbf{z}_{2} )+\phi_{m0}\), then \( \forall\ \mathbf{z}_{1}, \mathbf{z}_{2}\in \mathbb{I}^{m} \), we have
    \begin{equation}\label{eq:phi_tilde_integration}
        \int_{\mathbb{I}^{m}} \tilde{\phi}_{m}\left( \mathbf{z}_{1}, \mathbf{z}_{2} \right)\,\mathrm{d}G_{m}(\mathbf{z}_{2})=\int_{\mathbb{I}^{m}} \tilde{\phi}_{m}\left( \mathbf{z}_{1}, \mathbf{z}_{2} \right)\,\mathrm{d}G_{m}(\mathbf{z}_{1})=0.
    \end{equation}

    Using \( \tilde{\phi}_{m}( \mathbf{z}_{1}, \mathbf{z}_{2} ) \), we can reshape \( \hat{\phi}(m) \) after simple computations,
    \begin{equation}\label{eq:rewrite_phi_nm}
        \begin{aligned}
            \hat{\phi}(m) & =\phi_{m0}+ \frac{2}{n}\sum_{i=1}^{n}\left[ \phi_{m1}\left( \mathbf{Z}_{i;m} \right) -\phi_{m0}\right] \\
            & \qquad +\binom{n}{2}^{-1}\sum_{j=2}^{n}\sum_{i=1}^{j-1}\tilde{\phi}_{m}\left( \mathbf{Z}_{i;m}, \mathbf{Z}_{j;m} \right), \\
            & = \phi_{m0}+ \frac{2}{n}\sum_{i=1}^{n}\left[ \phi_{m1}\left( \mathbf{Z}_{i;m} \right) -\phi_{m0}\right] + \tilde{\phi}(m).
        \end{aligned}
    \end{equation}
    Note that, if \( j-i\geq m \), then \( \mathbf{Z}_{j;m} \) and \( \mathbf{Z}_{i;m} \) have no overlap variable. By this fact, we divide \( \tilde{\phi}(m) \) into two parts,
    \begin{equation*}
        \begin{aligned}
            \tilde{\phi}(m) & = \binom{n}{2}^{-1}\sum_{j=1+m}^{n}\sum_{i=1}^{j-m}\tilde{\phi}_{m}\left( \mathbf{Z}_{i;m}, \mathbf{Z}_{j;m} \right) \\
            & \qquad +\binom{n}{2}^{-1}\sum_{j=2}^{n}\sum_{i=1\vee(j-m+1)}^{j-1}\tilde{\phi}_{m}\left( \mathbf{Z}_{i;m}, \mathbf{Z}_{j;m} \right), \\
            & = \tilde{\phi}_{1}(m)+\tilde{\phi}_{2}(m).
        \end{aligned}
    \end{equation*}
    The double summation of \( \tilde{\phi}_{1}(m) \) includes \( (n-m)(n-m+1)/2 \) terms and there are \( (2n-m)(m-1)/2 \) summation terms in \( \tilde{\phi}_{2}(m) \). One can easily verify
    \begin{equation*}
        \mathbb{E}
        \tilde{\phi}_{m}^{2}\left( \mathbf{Z}_{i;m},
        \mathbf{Z}_{j;m} \right) = \mathbb{E}
        \phi_{m}^{2}\left( \mathbf{Z}_{i;m},
        \mathbf{Z}_{j;m} \right)-\phi_{m0}^{2} = O(c_{m}^{2}).
    \end{equation*}
    Hence, using Cauchy-Schwarz inequality, we have
    \begin{equation*}
        \begin{aligned}
            \mathbb{E}\left\vert \tilde{\phi}_{2}(m) \right\vert & = \mathbb{E}\left\vert \binom{n}{2}^{-1}\sum_{j=2}^{n}\sum_{i=1\vee(j-m+1)}^{j-1}\tilde{\phi}_{m}\left( \mathbf{Z}_{i;m}, \mathbf{Z}_{j;m} \right) \right\vert, \\
            & = O(mn^{-1}) \mathbb{E}\left\vert \tilde{\phi}_{m}\left( \mathbf{Z}_{i;m}, \mathbf{Z}_{j;m} \right) \right\vert, \\
            & \leq O(mn^{-1}) \sqrt{\mathbb{E} \tilde{\phi}_{m}^{2}\left( \mathbf{Z}_{i;m}, \mathbf{Z}_{j;m} \right)}, \\
            & = O(mn^{-1}c_{m}).
        \end{aligned}
    \end{equation*}
    By Markov's inequality, we have \( \tilde{\phi}_{2}(m) =O_{p}(mn^{-1}c_{m})\). We also notice that the \( \mathbf{Z}_{i;m}\) and \( \mathbf{Z}_{j;m} \) in the first term \( \tilde{\phi}_{1}(m) \) are independent, hence we have
    \begin{equation}\label{eq:tilde_phi_n1m_square}
        \begin{aligned}
            \mathbb{E}\tilde{\phi}^{2}_{1}(m) & = \binom{n}{2}^{-2} \sum_{j=1+m}^{n}\sum_{i=1}^{j-m}\sum_{t=1+m}^{n}\sum_{s=1}^{t-m} \\
            & \qquad \mathbb{E}\left[ \tilde{\phi}_{m}\left( \mathbf{Z}_{i;m}, \mathbf{Z}_{j;m} \right)\tilde{\phi}_{m}\left( \mathbf{Z}_{s;m}, \mathbf{Z}_{t;m} \right) \right]\times \mathbbm{1}\left( i,j\in S_{ij} \right),
        \end{aligned}
    \end{equation}
    where \( S_{ij}=\left[ s-m+1, s+m-1\right] \cup \left[ t-m+1, t+m-1\right]\). If at least one of \( i,j\notin S_{ij}\),
    by equation~\eqref{eq:phi_tilde_integration}, \( \mathbb{E}[ \tilde{\phi}_{m}( \mathbf{Z}_{i;m}, \mathbf{Z}_{j;m} )\tilde{\phi}_{m}( \mathbf{Z}_{s;m}, \mathbf{Z}_{t;m} ) ] =0 \).  The number of pair \( (s, t) \), \( 	t - s \geq m	 \) is of order \( O(n^2) \), for each given \( (s, t) \), if \( \mathbf{Z}_{i;m} \) has at least one overlap variable with \( \mathbf{Z}_{s;m} \) or \( \mathbf{Z}_{t;m} \), and  \( \mathbf{Z}_{j;m} \) has at least one overlap variable with \( \mathbf{Z}_{s;m} \) or \( \mathbf{Z}_{t;m} \) as well, then the expectation is nonzero. The indices of \( i, j \) have at most \( O(m) \) and \( O(m) \) choices respectively given \( m<M \). So the number of four summation terms in equation~\eqref{eq:tilde_phi_n1m_square} is of order \( O(n^{2}m^{2}) = O(n^{2})O(m)O(m) \), hence \( \mathbb{E}\tilde{\phi}^{2}_{1}(m) = O(m^{2}n^{-2}c_{m}^{2}) \) by Cauchy-Schwarz inequality and \( \mathbb{E} \tilde{\phi}_{m}^{2}( \mathbf{Z}_{i;m}, \mathbf{Z}_{j;m} )=O(c_{m}^{2})\). It follows that \( \tilde{\phi}_{1}(m) = O_{p}(mn^{-1}c_{m}) \) by Chebyshev's inequality, and finally \( \tilde{\phi}(m) =O_{p}(mn^{-1}c_{m})\).

    Next, similarly, we discuss the order of the second term in equation~\eqref{eq:rewrite_phi_nm}. Note that \( \phi_{m1}( \mathbf{Z}_{i;m} ) -\phi_{m0} \) is an \( m \)-dependence process with mean 0 and
    \begin{equation*}
        \mathbb{E} \left \{ \left[ \phi_{m1}\left( \mathbf{Z}_{i;m} \right) -\phi_{m0} \right] \left[ \phi_{m1}\left( \mathbf{Z}_{j;m} \right) -\phi_{m0} \right]\right \}=0,
    \end{equation*}
    if \( j-i \geq m \). The number of nonzero terms in \( \mathbb{E} { \{ \sum_{i=1}^{n}[ \phi_{m1}( \mathbf{Z}_{i;m} ) -\phi_{m0}]  \}}^2 \)  is of order \( O(nm) \). By these facts, \( \mathbb{E}{[ \phi_{m1}( \mathbf{Z}_{i;m} ) -\phi_{m0}]}^{2}\leq C \) and Chebyshev's inequality, we have \( 2n^{-1}\sum_{i=1}^{n}[ \phi_{m1}( \mathbf{Z}_{i;m} ) -\phi_{m0}]=O_{p}(m^{1/2}n^{-1/2}) \). This completes the proof.
\end{proof}
\begin{lemma}\label{lemma:third_U_statistics}
    Let \( \mathbf{Z}_{i;m}={(X_{i},\ldots,X_{i+m-1})}^{\top} \), \( m<M \) and \( \{X_{t}\} \) is \textit{i.i.d.} with CDF \( G_{1}(\cdot) \). Consider a third-order U-statistics
    \begin{equation}\label{eq:estimator_third_U_statistics}
        \hat{\phi}(m)=\binom{n}{3}^{-1}\sum_{k=3}^{n}\sum_{j=2}^{k-1}\sum_{i=1}^{j-1}\phi_{m}\left( \mathbf{Z}_{i;m}, \mathbf{Z}_{j;m}, \mathbf{Z}_{k;m} \right),
    \end{equation}
    where \( \phi_{m}\left( \cdot, \cdot, \cdot \right) \) is a kernel function in its argument and \( \forall\ \mathbf{z}_{1}\in \mathbb{I}^{m} \)
    \begin{equation}\label{eq:double_int_third_U_statistics}
        \int_{\mathbb{I}^{m}}\int_{\mathbb{I}^{m}}\phi_{m}\left( \mathbf{z}_{1}, \mathbf{z}_{2}, \mathbf{z}_{3} \right) \,\mathrm{d}G_{m}(\mathbf{z}_{2})\,\mathrm{d}G_{m}(\mathbf{z}_{3})=0,
    \end{equation}
    holds, where \( G_{m}(\mathbf{z})=G_{1}(x_{1})G_{1}(x_{2})\cdots G_{1}(x_{m}) \) and \( \mathbf{z}={(x_{1},\ldots,x_{m})}^{\top} \). Let
    \begin{equation*}
        \phi_{m2}(\mathbf{z}_{1},\mathbf{z}_{2})=\int_{\mathbb{I}^{m}}\phi_{m}\left( \mathbf{z}_{1}, \mathbf{z}_{2}, \mathbf{z}_{3} \right) \,\mathrm{d}G_{m}(\mathbf{z}_{3}).
    \end{equation*}
    Suppose \( \mathbb{E} \phi_{m}^{2}( \mathbf{Z}_{i;m}, \mathbf{Z}_{j;m},\mathbf{Z}_{k;m} )=O(c_{m}^{2})\), then we have
    \begin{equation*}
        \hat{\phi}(m)= 3\binom{n}{2}^{-1}\sum_{j=2}^{n}\sum_{i=1}^{j-1} \phi_{m2}\left( \mathbf{Z}_{i;m}, \mathbf{Z}_{j;m} \right) +O_{p}(m^{3/2}n^{-3/2}c_{m}).
    \end{equation*}
\end{lemma}
\begin{proof}[Proof of \textcolor{darkblue}{Lemma}~\ref{lemma:third_U_statistics}]
    \textcolor{darkblue}{Lemma}~\ref{lemma:third_U_statistics} is the version of Lemma B.2 in~\citet{hongAsymptoticDistributionTheory2005}. For \( m \)-consecutive variables, we still use the same notations as~\citet{hongAsymptoticDistributionTheory2005}'s.  As in \textcolor{darkblue}{Lemma}~\ref{lemma:second_U_statistics}, we construct a new symmetric third-order \textit{U}-statistics
    \begin{equation*}
        \tilde{\phi}_{m}(\mathbf{z}_{1},\mathbf{z}_{2},\mathbf{z}_{3})=\phi_{m}(\mathbf{z}_{1},\mathbf{z}_{2},\mathbf{z}_{3}) - \phi_{m2}(\mathbf{z}_{1},\mathbf{z}_{2})-\phi_{m2}(\mathbf{z}_{2},\mathbf{z}_{3})-\phi_{m2}(\mathbf{z}_{3},\mathbf{z}_{1}),
    \end{equation*}
    and it is easy to verify
    \begin{equation}\label{eq:integration_new_third_U_statistics}
        \int_{\mathbb{I}^{m}}\tilde{\phi}_{m}\left( \mathbf{z}_{1}, \mathbf{z}_{2}, \mathbf{z}_{3} \right) \,\mathrm{d}G_{m}(\mathbf{z}_{3})=0,\quad \forall\ \mathbf{z}_{1},\mathbf{z}_{2}\in \mathbb{I}^{m},
    \end{equation}
    given equation~\eqref{eq:double_int_third_U_statistics}. We can rewrite equation~\eqref{eq:estimator_third_U_statistics} using \( \tilde{\phi}_{m}(\cdot,\cdot,\cdot) \) as
    \begin{equation*}
        \begin{aligned}
            \hat{\phi}(m) & = 3\binom{n}{2}^{-1}\sum_{j=2}^{n}\sum_{i=1}^{j-1} \phi_{m2}\left( \mathbf{Z}_{i;m},
            \mathbf{Z}_{j;m} \right) \\
            & \qquad + \binom{n}{3}^{-1}\sum_{k=3}^{n}\sum_{j=2}^{k-1}\sum_{i=1}^{j-1}\tilde{\phi}_{m}\left( \mathbf{Z}_{i;m}, \mathbf{Z}_{j;m}, \mathbf{Z}_{k;m} \right), \\
            & = 3\binom{n}{2}^{-1}\hat{\phi}_{2}(m) + \binom{n}{3}^{-1}\tilde{\phi}(m).
        \end{aligned}
    \end{equation*}
    Let
    \begin{equation}\label{eq:triple_independent}
        \tilde{\phi}_{1}(m)=\sum_{k=1+2m}^{n}\sum_{j=1+m}^{k-m}\sum_{i=1}^{j-m}\tilde{\phi}_{m}\left( \mathbf{Z}_{i;m}, \mathbf{Z}_{j;m}, \mathbf{Z}_{k;m} \right),
    \end{equation}
    then \( \mathbf{Z}_{i;m}, \mathbf{Z}_{j;m}, \mathbf{Z}_{k;m} \) in~\eqref{eq:triple_independent} are mutually independent and~\eqref{eq:triple_independent} includes \( \binom{n-2m+2}{3} \) terms. Then \( \tilde{\phi}_{2}(m)=\tilde{\phi}(m) -\tilde{\phi}_{1}(m) \) includes \( \binom{n}{3}- \binom{n-2m+2}{3}=O(mn^{2})\) terms by \( m<M \). Using Cauchy-Schwarz inequality, we can verify \( \mathbb{E} \tilde{\phi}_{m}^{2}( \mathbf{Z}_{i;m}, \mathbf{Z}_{j;m},\mathbf{Z}_{k;m} )=O(c_{m}^{2})\) as well. Hence, we have
    \begin{equation*}
        \begin{aligned}
            \mathbb{E}\left\vert \tilde{\phi}_{2}(m) \right\vert & = O(mn^{2}) \mathbb{E}\left\vert \tilde{\phi}_{m}\left( \mathbf{Z}_{i;m}, \mathbf{Z}_{j;m}, \mathbf{Z}_{k;m}\right) \right\vert, \\
            & \leq O(mn^{2}) \sqrt{\mathbb{E} \tilde{\phi}_{m}^{2}\left( \mathbf{Z}_{i;m}, \mathbf{Z}_{j;m}, \mathbf{Z}_{k;m} \right)}, \\
            & \qquad(\text{by Cauchy-Schwarz inequality }), \\
            & = O(mn^{2}c_{m}).
        \end{aligned}
    \end{equation*}
    For \( \tilde{\phi}_{1}(m) \), we have
    \begin{equation}\label{eq:tilde_phi_n1m_square_triple}
        \begin{aligned}
            \mathbb{E}\tilde{\phi}^{2}_{1}(m) & =  \sum_{k=1+2m}^{n}\sum_{j=1+m}^{k-m}\sum_{i=1}^{j-m}\sum_{r=1+2m}^{n}\sum_{t=1+m}^{r-m}\sum_{s=1}^{t-m} \\
            & \qquad\mathbb{E}\left[ \tilde{\phi}_{m}\left( \mathbf{Z}_{i;m}, \mathbf{Z}_{j;m}, \mathbf{Z}_{k;m} \right) \tilde{\phi}_{m}\left( \mathbf{Z}_{s;m}, \mathbf{Z}_{t;m}, \mathbf{Z}_{r;m} \right) \right] \\
            & \qquad \times \mathbbm{1}\left( i,j,k\in S_{ijk} \right),
        \end{aligned}
    \end{equation}
    where
    \begin{equation*}
        S_{ijk}=\left[ s-m+1, s+m-1\right]
        \cup \left[ t-m+1, t+m-1\right]\cup
        \left[ r-1\right].
    \end{equation*}
    If at least one  of  \( i,j,k\notin S_{ijk}\), then by
    equation~\eqref{eq:integration_new_third_U_statistics},
    \begin{equation*}
        \mathbb{E}\left[ \tilde{\phi}_{m}\left( \mathbf{Z}_{i;m}, \mathbf{Z}_{j;m}, \mathbf{Z}_{k;m} \right)\tilde{\phi}_{m}\left( \mathbf{Z}_{s;m}, \mathbf{Z}_{t;m} , \mathbf{Z}_{r;m}\right) \right] =0.
    \end{equation*}
    The number of triplet \( (s, t, r) \), \( 	t- s \geq m	 \) and \( 	r- t \geq m	 \) is of order \( O(n^3) \), for each given triplet \( (s, t, r)\), if each of \( \mathbf{Z}_{i;m} \), \( \mathbf{Z}_{j;m} \), \( \mathbf{Z}_{k;m} \) has at least one overlap variable with \( \mathbf{Z}_{s;m} \), \( \mathbf{Z}_{t;m} \) or \( \mathbf{Z}_{r;m} \),
    then the expectation is nonzero. The indices of \( i, j, k \) have at most \( O(m) \), \( O(m) \) and \( O(m) \) choices respectively given \( m<M \). So the number of six summation terms in equation~\eqref{eq:tilde_phi_n1m_square_triple} is of order \( O(n^{3}m^{3}) = O(n^{3})O(m)O(m)O(m) \), hence \( \mathbb{E}\tilde{\phi}^{2}_{1}(m) = O(m^{3}n^{3}c_{m}^{2}) \) by Cauchy-Schwarz inequality and \( \mathbb{E} \tilde{\phi}_{m}^{2}( \mathbf{Z}_{i;m}, \mathbf{Z}_{j;m}, \mathbf{Z}_{k;m})=O(c_{m}^{2})\). Finally, it follows that \( \binom{n}{3}^{-1}\tilde{\phi}(m) =O_{p}(m^{3/2}n^{-3/2}c_{m})\) by Chebyshev's inequality and Markov's inequality. This completes the proof.
\end{proof}
\begin{lemma}\label{lemma:order_an_bn}
    Given \textcolor{darkblue}{Assumptions}~\ref{assump:kernel_assumption} and~\ref{assump:density_assumption}, \( \forall\ \mathbf{z}_{1}, \mathbf{z}_{2}\in \mathbb{I}^{m} \), if \( \mathbf{z}_{1}, \mathbf{z}_{2} \) are independent, then \( \mathbb{E}A_{m}^{2}(\mathbf{z}_{1}, \mathbf{z}_{2}) \)  is of order \( O(h^{-m}) \). Furthermore, for \( 1\leq m \leq M\), \( \max_{m}\sup_{\mathbf{z}_{1}\in\mathbb{I}^{m}} B_{m}(\mathbf{z}_{1})=O(h^{2})\).
\end{lemma}
\begin{proof}[Proof of \textcolor{darkblue}{Lemma}~\ref{lemma:order_an_bn}]
    To prove this lemma, we investigate the univariate case, i.e.,
    \begin{equation*}
        a(x_{1},x_{2})=\left[ K_{h}^{J}(x_{1},x_{2})- \int_{0}^{1}K_{h}^{J}(x_{1},x)g_{1}(x)dx\right]/g_{1}(x_{1}).
    \end{equation*}
    For simplicity, we denote \( \bar{K}_{h}^{J}(x_{1}) =\int_{0}^{1}K_{h}^{J}(x_{1},x)g_{1}(x)dx\), then
    \begin{equation*}
        a(x_{1},x_{2})=\left[ K_{h}^{J}(x_{1},x_{2})-
        \bar{K}_{h}^{J}(x_{1})\right]/g_{1}(x_{1}),
    \end{equation*}
    \begin{equation*}
        \int_{0}^{1} a(x_{1},x_{2})g_{1}(x_{2})\,\mathrm{d}x_{2}=0,
    \end{equation*}
    and
    \begin{align*}
        \mathbb{E}a(x_{1},x_{2}) & =\int_{0}^{1} \int_{0}^{1} a(x_{1},x_{2})g_{1}(x_{1})g_{1}(x_{2})\,\mathrm{d}x_{1}\,\mathrm{d}x_{2}, \\
        & = \int_{0}^{1} \int_{0}^{1} \left[ K_{h}^{J}(x_{1},x_{2})-
        \bar{K}_{h}^{J}(x_{1})\right]g_{1}(x_{2})\,\mathrm{d}x_{1}\,\mathrm{d}x_{2}  =0.
    \end{align*}
    We have
    \begin{align*}
        a^{2}(x_{1},x_{2}) & = {\left[ \frac{K_{h}^{J}(x_{1},x_{2})-
        \bar{K}_{h}^{J}(x_{1})}{g_{1}(x_{1})} \right]}^{2}, \\
        & = \frac{{\left[  K_{h}^{J}(x_{1},x_{2})\right]}^{2}}{g_{1}^{2}(x_{1})} +\frac{{\left[ \bar{K}_{h}^{J}(x_{1}) \right]}^{2}}{g_{1}^{2}(x_{1})}-2\frac{K_{h}^{J}(x_{1},x_{2})\times \bar{K}_{h}^{J}(x_{1}) }{g_{1}^{2}(x_{1})}, \\
        & =\psi_{1}(x_{1},x_{2})+\psi_{2}(x_{1},x_{2})-2\psi_{3}(x_{1},x_{2}).
    \end{align*}
    Next, we will discuss each \( \psi_{i}(\cdot,\cdot), i=1,2,3\).
    \begin{align*}
        \mathbb{E}\psi_{1}(x_{1},x_{2})
        & = \int_{0}^{1} \int_{0}^{h}h^{-2}k^{2}_{(x_{1}/h)}\left( \frac{x_{1}-x_{2}}{h} \right) \frac{g_{1}(x_{2})}{g_{1}(x_{1})}\,\mathrm{d}x_{1}\,\mathrm{d}x_{2} \\
        & \quad +\int_{0}^{1} \int_{h}^{1-h}h^{-2}K^{2}_{h}\left( \frac{x_{1}-x_{2}}{h} \right) \frac{g_{1}(x_{2})}{g_{1}(x_{1})}\,\mathrm{d}x_{1}\,\mathrm{d}x_{2} \\
        & \quad +\int_{0}^{1} \int_{1-h}^{1}h^{-2}k^{2}_{([1-x_{1}]/h)}\left( \frac{x_{1}-x_{2}}{h} \right) \frac{g_{1}(x_{2})}{g_{1}(x_{1})}\,\mathrm{d}x_{1}\,\mathrm{d}x_{2}, \\
        & =\psi_{11}+\psi_{12}+\psi_{13}.
    \end{align*}
    By change of variable, we have
    \begin{align*}
        \psi_{11}
        & = \int_{0}^{1} \int_{0}^{h}h^{-2}k^{2}_{(x_{1}/h)}\left( \frac{x_{1}-x_{2}}{h} \right) \frac{g_{1}(x_{2})}{g_{1}(x_{1})}\,\mathrm{d}x_{1}\,\mathrm{d}x_{2}, \\
        & =\int_{0}^{h}\int_{-1}^{x_{1}/h} h^{-1}k^{2}_{(x_{1}/h)}\left( u \right) \frac{g_{1}(x_1-uh)}{g_{1}(x_{1})}\,\mathrm{d}u\,\mathrm{d}x_{1}, \\
        & =\int_{0}^{h}\int_{-1}^{x_{1}/h} h^{-1}k^{2}_{(x_{1}/h)}\left( u \right)\left[ 1-uh\frac{g_{1}^{\prime}(x_1)}{g_{1}(x_{1})} \right] \,\mathrm{d}u\,\mathrm{d}x_{1}, \\
        & =\int_{0}^{1}\int_{-1}^{\rho}k_{\rho}^{2}(u)\,\mathrm{d}u\,\mathrm{d}\rho-h\int_{0}^{1}\int_{-1}^{\rho}uk_{\rho}^{2}(u)\frac{g_{1}^{\prime}(\rho h)}{g_{1}(\rho h)}\,\mathrm{d}u\,\mathrm{d}\rho+O(h^{2}).
    \end{align*}
    Using the same method, we have
    \begin{align*}
        \psi_{12} & = \int_{0}^{1} \int_{h}^{1-h}h^{-2}K^{2}_{h}\left( \frac{x_{1}-x_{2}}{h} \right) \frac{g_{1}(x_{2})}{g_{1}(x_{1})}\,\mathrm{d}x_{1}\,\mathrm{d}x_{2}, \\
        & = \int_{h}^{1-h}\int_{-1}^{1}h^{-1}K^{2}_{h}\left( u \right)\left[ 1-uh \frac{g_{1}^{\prime}(x_{1})}{g_{1}(x_{1})}+\frac{1}{2}u^{2}h^{2}\frac{g_{1}^{\prime\prime}(x_{1})}{g_{1}(x_{1})}\right] \,\mathrm{d}u\,\mathrm{d}x_{1}, \\
        & =(h^{-1}-2)\int_{-1}^{1} K^{2}(u)\,\mathrm{d}u +O(h),
    \end{align*}
    and
    \begin{align*}
        \psi_{13}
        & =\int_{0}^{1}\int_{-\rho}^{1}k_{\rho}^{2}(u)\,\mathrm{d}u\,\mathrm{d}\rho-h\int_{0}^{1}\int_{-\rho}^{1}uk_{\rho}^{2}(u)\frac{g_{1}^{\prime}(\rho h)}{g_{1}(\rho h)}\,\mathrm{d}u\,\mathrm{d}\rho +O(h^{2}).
    \end{align*}
    Hence,
    \begin{equation*}
        \mathbb{E}\psi_{1}(x_{1},x_{2})=(h^{-1}-2)\int_{-1}^{1} K^{2}(u)+2\int_{0}^{1}\int_{-1}^{\rho}k_{\rho}^{2}(u)\,\mathrm{d}u\,\mathrm{d}\rho + O(h),
    \end{equation*}
    \begin{equation*}
        \mathbb{E}\psi_{2}(x_{1},x_{2})=\int_{0}^{1} {\left[ \bar{K}_{h}^{J}(x_{1}) \right]}^{2}/g_{1}(x_{1})\,\mathrm{d}x_{1},
    \end{equation*} and
    \begin{align*}
        \mathbb{E}\psi_{3}(x_{1},x_{2}) & =\int_{0}^{1} \int_{0}^{1} \frac{K_{h}^{J}(x_{1},x_{2})\times \bar{K}_{h}^{J}(x_{1}) }{g_{1}^{2}(x_{1})}g_{1}(x_{1})g_{1}(x_{2})\,\mathrm{d}x_{1}\,\mathrm{d}x_{2}, \\
        & =\int_{0}^{1} {\left[ \bar{K}_{h}^{J}(x_{1}) \right]}^{2}/g_{1}(x_{1})\,\mathrm{d}x_{1}.
    \end{align*}
    Now we need to expand \( \bar{K}_{h}^{J}(x_{1}) \), when \(
    0\leq x_{1}\leq h \),
    \begin{align*}
        \bar{K}_{h}^{J}(x_{1})
        & =g_{1}(x_{1})\int_{-1}^{x_{1}/h}k_{(x_{1}/h)}\left( u \right)du-hg_{1}^{\prime}(x_{1})\int_{-1}^{x_{1}/h}uk_{(x_{1}/h)}\left( u \right)du \\
        & \qquad+\frac{1}{2}h^{2}g_{1}^{\prime\prime}(x_{1})\int_{-1}^{x_{1}/h}u^{2}k_{(x_{1}/h)}\left( u \right)du, \\
        & =g_{1}(x_{1})+O(h^{2}),
    \end{align*}
    because
    \begin{align*}
        \int_{-1}^{x_{1}/h}k_{(x_{1}/h)}\left( u \right)du & = 1, \\
        \int_{-1}^{x_{1}/h}uk_{(x_{1}/h)}\left( u \right)du & =   0.
    \end{align*}
    Note that for Jackknife kernel \( k_{\rho}(u) \),
    \begin{equation*}
        \int_{-1}^{\rho} k_{\rho}(u)\,\mathrm{d}u=\int_{-1}^{\rho} (1+\beta)\frac{K(u)}{\omega_{0}(\rho)}\,\mathrm{d}u-\int_{-\alpha}^{\rho}\frac{\beta}{\alpha}\frac{K(u/\alpha)}{\omega_{0}(\rho/\alpha)}\,\mathrm{d}u,
    \end{equation*}
    for \( h\leq x_{1}\leq 1-h \) and \( 1-h<x_{1}\leq 1 \), we also have \( \bar{K}_{h}^{J}(x_{1})= g_{1}(x_{1})+O(h^{2})\), so \( {\left[ \bar{K}_{h}^{J}(x_{1}) \right]}^{2}/g_{1}(x_{1})=g_{1}(x_{1})+O(h^{2}) \), hence \( \mathbb{E}\psi_{2}(x_{1},x_{2})=\mathbb{E}\psi_{3}(x_{1},x_{2})=1+O(h^{2} ) \). Finally,
    \begin{equation*}
        \begin{aligned}
            \mathbb{E}a^{2}(x_{1},x_{2}) & = \mathbb{E}\psi_{1}(x_{1},x_{2})+\mathbb{E}\psi_{2}(x_{1},x_{2})-2\mathbb{E}\psi_{3}(x_{1},x_{2}) \\
            & = (h^{-1}-2)\int_{-1}^{1} K^{2}(u)du \\
            & \qquad+2\int_{0}^{1}\int_{-1}^{\rho}k_{\rho}^{2}(u)\,\mathrm{d}u\,\mathrm{d}\rho -1 + O(h), \\
            & =O(h^{-1}).
        \end{aligned}
    \end{equation*}
    Now, we have obtained the result of the univariate case. For multivariate case, let \( \mathbf{z}_{1}={(z_{10},\ldots,z_{1(m-2)},z_{1(m-1)})}^{\top}={( \mathbf{y}_{1}^{\top},z_{1(m-1)} )}^{\top} \) and
    \begin{equation*}
        \mathbf{z}_{2}={(z_{20},\ldots,z_{1(m-2)},z_{1(m-1)})}^{\top}={\left( \mathbf{y}_{2}^{\top},z_{2(m-1)} \right)}^{\top},
    \end{equation*}
    by the \textcolor{darkblue}{Assumptions}~\ref{assump:kernel_assumption} and~\ref{assump:density_assumption}, the definition of multivariate kernel and \( \mathbb{H}_{0} \), we can prove that
    \begin{equation}\label{eq:A_expand}
        \begin{aligned}
            \mathbb{E}A^{2}_{m}(\mathbf{z}_{1},\mathbf{z}_{2}) & = \mathbb{E}A^{2}_{(m-1)}(\mathbf{y}_{1},\mathbf{y}_{2})\mathbb{E}a^{2}(z_{1(m-1)}, z_{2(m-1)}) \\
            & \qquad +  \mathbb{E}A^{2}_{(m-1)}(\mathbf{y}_{1},\mathbf{y}_{2})\mathbb{E} {\left[ \frac{\bar{K}^{J}_{h}(z_{1(m-1)})}{g_{1}(z_{1(m-1)})} \right]}^{2} \\
            & \qquad+  \mathbb{E}a^{2}(z_{1(m-1)},z_{2(m-1)})\mathbb{E} {\left[ \frac{\bar{\mathcal{K}}^{(m-1)}_{h}(\mathbf{y}_{1})}{g_{1}(\mathbf{y}_{1})} \right]}^{2}, \\
            & = \mathbb{E}A^{2}_{(m-1)}(\mathbf{y}_{1},\mathbf{y}_{2})O(h^{-1}) \\
            & \qquad +  \mathbb{E}A^{2}_{(m-1)}(\mathbf{y}_{1},\mathbf{y}_{2})\mathbb{E} {\left[ b(z_{1(m-1)}) +1\right]}^{2} \\
            & \qquad+  \mathbb{E}a^{2}(z_{1(m-1)},z_{2(m-1)})\mathbb{E} {\left[ B_{(m-1)}(\mathbf{y}_{1})+1 \right]}^{2},
        \end{aligned}
    \end{equation}
    where \( \bar{K}^{J}_{h}(z_{1})=\int_{0}^{1} K_{h}^{J}(z_{1},z_{2})g_{1}(z_{2})\,\mathrm{d}z_{2} \) and
    \begin{equation*}
        \bar{\mathcal{K}}^{(m-1)}_{h}(\mathbf{y}_{1})=\int_{\mathbbm{I}^{m-1}} \mathcal{K}_{h}^{(m-1)}(\mathbf{y}_{1},\mathbf{y})g(\mathbf{y})\,\mathrm{d}\mathbf{y}.
    \end{equation*}
    Iteratively, we can obtain \( \mathbb{E}A_{m}^{2}(\mathbf{z}_{1}, \mathbf{z}_{2})=O(h^{-m}) \) by equations~\eqref{eq:A_expand} and~\eqref{eq:B_expand}.

    For the last part, we have \( b(x_{1})=\bar{K}_{h}^{J}(x_{1})/g_{1}(x_{1})-1=O(h^{2})\), given \( \mathbb{H}_{0} \), we can also verify that
    \begin{equation}\label{eq:B_expand}
        B_{m}(\mathbf{z}_{1})=B_{(m-1)}(\mathbf{y}_{1})b(z_{1(m-1)})+B_{(m-1)}(\mathbf{y}_{1})+b(z_{1(m-1)}).
    \end{equation}
    This immediately completes the proof.
\end{proof}
\begin{lemma}\label{lemma:H1_expectation}
    Given \( \mathbb{H}_{0} \) and \( 1\leq m< M \), we have
    \begin{equation*}
        \mathbb{E}H_{1m}(\mathbf{z}_{1},\mathbf{z}_{2})=
        \begin{cases}
            0        & \text{if }  \mathbf{z}_{1},\mathbf{z}_{2}\, \text{are independent}, \\
            O(h^{2}) & \text{otherwise}.
        \end{cases}
    \end{equation*}
\end{lemma}
\begin{proof}[Proof of \textcolor{darkblue}{Lemma}~\ref{lemma:H1_expectation}] When \( \mathbf{z}_{1} \) and \( \mathbf{z}_{2} \) have no overlap variable, i.e., \( \mathbf{z}_{1} \) and \( \mathbf{z}_{2} \) are independent, by the definition~\eqref{eq:A_nm},~\eqref{eq:H_1nm_z1_z2} and \( \mathbb{H}_{0} \), we have \( \mathbb{E}H_{1m}(\mathbf{z}_{1},\mathbf{z}_{2})=0 \). Next, we will prove the order of \( \mathbb{E}H_{1m}(\mathbf{z}_{1},\mathbf{z}_{2}) \) is \( O(h^{2}) \) if \( \mathbf{z}_{1} \) and \( \mathbf{z}_{2} \) have one or more overlap variables. Note that \( \hat{H}_{1}(m) \) is a U-statistics, \( \mathbf{z}_{1}\neq \mathbf{z}_{2} \), then \( \mathbf{z}_{1} \) and \( \mathbf{z}_{2} \) have at most \( m-1 \) overlap variables. By the fact \( \tilde{A}_{m}(\mathbf{z}_{1}, \mathbf{z}_{2}) - A_{m}(\mathbf{z}_{1}, \mathbf{z}_{2})=\gamma_{m}(\mathbf{z}_{1},\mathbf{z}_{2})\) and \textcolor{darkblue}{Lemma}~\ref{lemma:Gamma_m}, we only need to proof \( \mathbb{E}A_{m}(\mathbf{z}_{1}, \mathbf{z}_{2})=O(h^{2}) \). First, we prove the order is \( O(h^{2}) \) if \( \mathbf{z}_{1} \) and \( \mathbf{z}_{2} \) sharing \( m-1 \) overlap variables, then extend the result to the case whence \( \mathbf{z}_{1} \) and \( \mathbf{z}_{2} \) share only \( 1 \) overlap variable. Let \( \mathbf{z}_{1} =(z_{1},\ldots,z_{m})\) and \( \mathbf{z}_{2} =(z_{2},\ldots,z_{m+1})\), from \textcolor{darkblue}{Lemma}~\ref{lemma:order_an_bn}, we know \( \forall\, \mathbf{z}_{1}\in \mathbbm{I}^{m} \), \( \bar{\mathcal{K}}^{(m)}_{h} =g(\mathbf{z}_{1})+O(h^{2})\). Hence,
    \begin{equation}\label{eq:expectation_A_nm}
        \begin{aligned}
            & \mathbb{E}A_{m}(\mathbf{z}_{1},
            \mathbf{z}_{2}) \\
            & \quad = \int_{0}^{1} \cdots \int_{0}^{1} \frac{\mathcal{K}^{(m)}_{h}(\mathbf{z}_{1},
                \mathbf{z}_{2})-\bar{\mathcal{K}}^{(m)}_{h}(\mathbf{z}_{1})}{g(\mathbf{z}_{1})}g(\mathbf{z}_{1})g_{1(z_{m+1})}\,\mathrm{d}z_{1}\cdots \,\mathrm{d}z_{m+1}, \\
            & \quad =  \int_{0}^{1} \cdots \int_{0}^{1} \left[ \mathcal{K}^{(m)}_{h}(\mathbf{z}_{1},
                \mathbf{z}_{2})- g(\mathbf{z}_{1})+O(h^{2})\right] g_{1(z_{m+1})}\,\mathrm{d}z_{1}\cdots \,\mathrm{d}z_{m+1}, \\
            & \quad=\int_{0}^{1} \cdots \int_{0}^{1} \mathcal{K}^{(m)}_{h}(\mathbf{z}_{1},
            \mathbf{z}_{2})g_{1(z_{m+1})}\,\mathrm{d}z_{1}\cdots \,\mathrm{d}z_{m+1}-1+O(h^{2}).
        \end{aligned}
    \end{equation}
    We also notice that
    \begin{equation}\label{eq:integral_z_m+1}
        \begin{aligned}
            & \int_{0}^{1} \mathcal{K}^{(m)}_{h}(\mathbf{z}_{1},
            \mathbf{z}_{2})g_{1(z_{m+1})}\,\mathrm{d}z_{m1} \\
            & \quad =K^{J}_{h}(z_1,z_{2})\times\cdots \times K^{J}_{h}(z_{m-1},z_{m})\int_{0}^{1} K^{J}_{h}(z_m,z_{m+1})g_{1(z_{m+1})}\,\mathrm{d}z_{m1}, \\
            & \quad =  K^{J}_{h}(z_1,z_{2})\times\cdots \times K^{J}_{h}(z_{m-1},z_{m})\left[ g_{1}(z_{m}) +O(h^2)\right],
        \end{aligned}
    \end{equation}
    iteratively substituting~\eqref{eq:integral_z_m+1} into integration~\eqref{eq:expectation_A_nm}, we finally obtain
    \begin{equation*}
        \mathbb{E}A_{m}(\mathbf{z}_{1},
        \mathbf{z}_{2})=O(h^{2}).
    \end{equation*}
    Using the similar method, one can easily prove \( \mathbb{E}A_{m}(\mathbf{z}_{1}, \mathbf{z}_{2})=O(h^{2}) \) if \( \mathbf{z}_{1}, \mathbf{z}_{2} \) have only one overlap variable. This completes the proof.
\end{proof}
\begin{lemma}\label{lemma:H2_expectation}
    Given \( \mathbb{H}_{0} \) and \( 2\leq m< M \), we have
    \begin{equation*}
        \mathbb{E}H_{2m}(\mathbf{z}_{1},\mathbf{z}_{2})=
        \begin{cases}
            0               & \text{if }  \mathbf{z}_{1},\mathbf{z}_{2}\, \text{are independent}, \\
            \tau^{m}-1+O(h) & \text{otherwise, }
        \end{cases}
    \end{equation*}
    where \( \tau=\int_{-1}^{1} \int_{-1}^{1} K(u)K(u+v)\,\mathrm{d}u\,\mathrm{d}v \).
\end{lemma}
\begin{proof}[Proof of \textcolor{darkblue}{Lemma}~\ref{lemma:H2_expectation}]
    Let \( \mathbf{z}_{0}={(z_{01},\ldots,z_{0m})}^{\top} \), \( \mathbf{z}_{1}={(z_{1},\ldots,z_{m})}^{\top} \) and \\ \( \mathbf{z}_{2}={(z_{2},\ldots,z_{m+1})}^{\top} \), then we have
    \begin{align*}
        H_{2m}(\mathbf{z}_{1},\mathbf{z}_{2}) & = \int_{\mathbbm{I}^{m}}A_{m}(\mathbf{z}_{0},\mathbf{z}_{1})A_{m}(\mathbf{z}_{0},\mathbf{z}_{2})g(\mathbf{z}_{0}) \,\mathrm{d}\mathbf{z}_{0}, \\
        & = \int_{\mathbbm{I}^{m}}\frac{\mathcal{K}^{(m)}_{h}(\mathbf{z}_{0},\mathbf{z}_{1})\mathcal{K}^{(m)}_{h}(\mathbf{z}_{0},\mathbf{z}_{2})}{g(\mathbf{z}_{0})}\,\mathrm{d}\mathbf{z}_{0} \\
        & \qquad-\int_{\mathbbm{I}^{m}}\left[\mathcal{K}^{(m)}_{h}(\mathbf{z}_{0},\mathbf{z}_{1})+\mathcal{K}^{(m)}_{h}(\mathbf{z}_{0},\mathbf{z}_{2}) \right]\,\mathrm{d}\mathbf{z}_{0}+1+O(h^{2}),
    \end{align*}
    therefore
    \begin{align*}
        & \mathbb{E}H_{2m}(\mathbf{z}_{1},\mathbf{z}_{2}) \\
        & \quad = \int_{0}^{1} \int_{\mathbbm{I}^{m}}\int_{\mathbbm{I}^{m}}\frac{\mathcal{K}^{(m)}_{h}(\mathbf{z}_{0},\mathbf{z}_{1})\mathcal{K}^{(m)}_{h}(\mathbf{z}_{0},\mathbf{z}_{2})}{g(\mathbf{z}_{0})}g(\mathbf{z}_{1})g_{1}(z_{m+1}) \,\mathrm{d}\mathbf{z}_{0} \,\mathrm{d}\mathbf{z}_{1}\,\mathrm{d}z_{m+1}-1+O(h^{2}).
    \end{align*}
    By change of variable and the first-order Taylor expansion, the first term can be expressed as \( \tau^{m}+O(h) \). One can also obtain the same result  for \( \mathbf{z}_{1}={(z_{1},\ldots,z_{m})}^{\top} \) and \( \mathbf{z}_{2}={(z_{m},\ldots,z_{2m-1})}^{\top} \) using the same discussion. This completes the proof.
\end{proof}
\section{Simulation Results}\label{sec:simulation_results}
\subsection{Case 3}\label{subsec:Case 3} 
This case is designed to evaluate the performance of change-point detection in nonlinear time series models:
\begin{align*}
    \text{Model 1:} & \qquad x_{i} = 0.138 + (0.316 + 0.982x_{i-1})e^{-3.89x_{i-1}^{2}} + \varepsilon_{1i}, \\ 
    \text{Model 2:} & \qquad y_{i} = -0.437 - (0.659 + 1.260y_{i-1})e^{-3.89y_{i-1}^{2}} + \varepsilon_{2i},
\end{align*}
where \(\varepsilon_{1i}\) and \(\varepsilon_{i2}\) are Gaussian noise with mean zero and variance \(0.2^2\).
We let \( N=500 \), then generate \( P_{1}=160 \) and \( P_{2}=80 \) time series from Model 1 and Model 2 respectively. Denote \( P = P_{1}+P_{2} \),  the change-point is 161. To investigate the robustness of \(\mathfrak{E}\) concerning the selection of \(m\), we appropriately allow \( m \) to change from 1 to 8. The other settings are the same as Case 2.~\textcolor{darkblue}{Table}~\ref{tab:Case3_apen_rlen} summarizes the comparison between \(\mathfrak{E}\) and ApEn. We can obtain the same conclusion as Case 2.
\begin{table}[htbp]
    \centering
    \caption{The Comparison Between \(\mathfrak{E}\) and ApEn for Different \( m \)
        in Case 3}\label{tab:Case3_apen_rlen}
    \begin{tabular}{ccccccc}
        \hline
        \multirow{2}{*}{\( m \)} &\multicolumn{3}{c}{\(\mathfrak{E}\)}&\multicolumn{3}{c}{ApEn} \\
        \cmidrule(lr){2-4}\cmidrule(lr){5-7}
        & MAD & Failure & \( \tau=161 \) & MAD & Failure &
        \( \tau=161 \) \\
        \hline
        1 & 0.2333 & 0 & 123/150 & 0.3067 & 0 & 114/150 \\
        2 & 0.2200 & 0 & 123/150 & 85.80 & 110 & 0/150 \\
        3 & 0.3000 & 0 & 116/150 & 1.120 & 0 & 79/150 \\
        4 & 0.4333 & 0 & 106/150 & 2.4467 & 0 & 51/150 \\
        5 & 0.4467 & 0 & 104/150 & 12.1727 & 11 & 19/150 \\
        6 & 0.3533 & 0 & 111/150 & 38.4800 & 100 & 2/150 \\
        7 & 0.4133 & 0 & 107/150 & 88.5385 & 137 & 0/150 \\
        8 & 0.4800 & 0 & 104/150 & 105.667 & 144 & 0/150 \\ \hline
    \end{tabular}
\end{table}
\begin{figure}[ht]
    \centering
    \begin{subfigure}[bt]{0.48\textwidth}
        \centering
        \includegraphics[width=\textwidth]{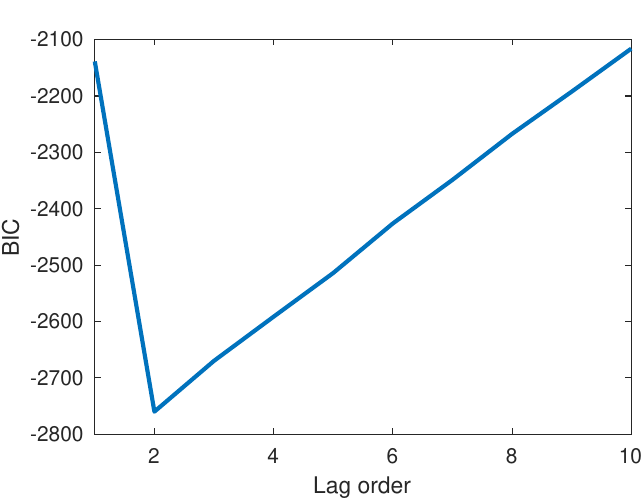}
        \caption{BIC}\label{fig:Case1BIC.pdf}
    \end{subfigure}
    \begin{subfigure}[bt]{0.48\textwidth}
        \centering
        \includegraphics[width=\textwidth]{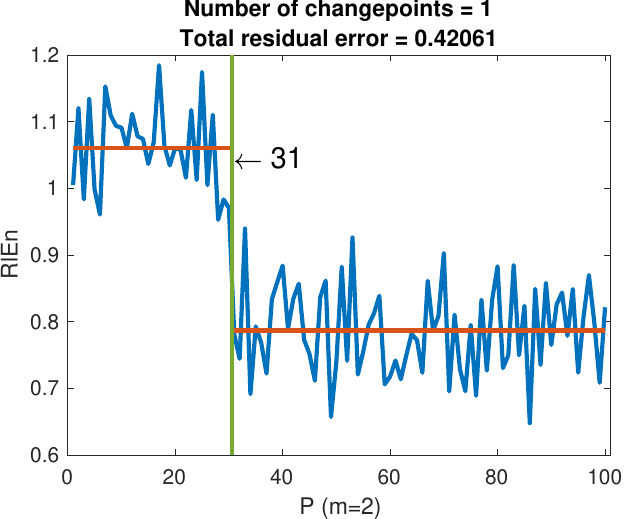}
        \caption{Change-point}\label{fig:Case1ChangPoint.pdf}
    \end{subfigure}
    \caption{Result of Case 1}\label{fig:Result_of_Case_1}
\end{figure}
\section{Real Data Analysis}\label{sec:real_data_analysis_more}
\subsection{Muscle Contraction Data from Single Subject}\label{subsec:muscle_contraction_data}
The real data contains 659977 observations, recorded at each millisecond.~\Cref{fig:allquad.png} shows the maximal contractions used to detect the extent of neuromuscular fatigue. There are also 10 smaller sharp rises afterwards - electrical stimuli to detect peripheral fatigue.
\begin{figure}[ht]
    \centering
    \begin{subfigure}[bt]{0.48\textwidth}
        \centering
        \includegraphics[width=\textwidth]{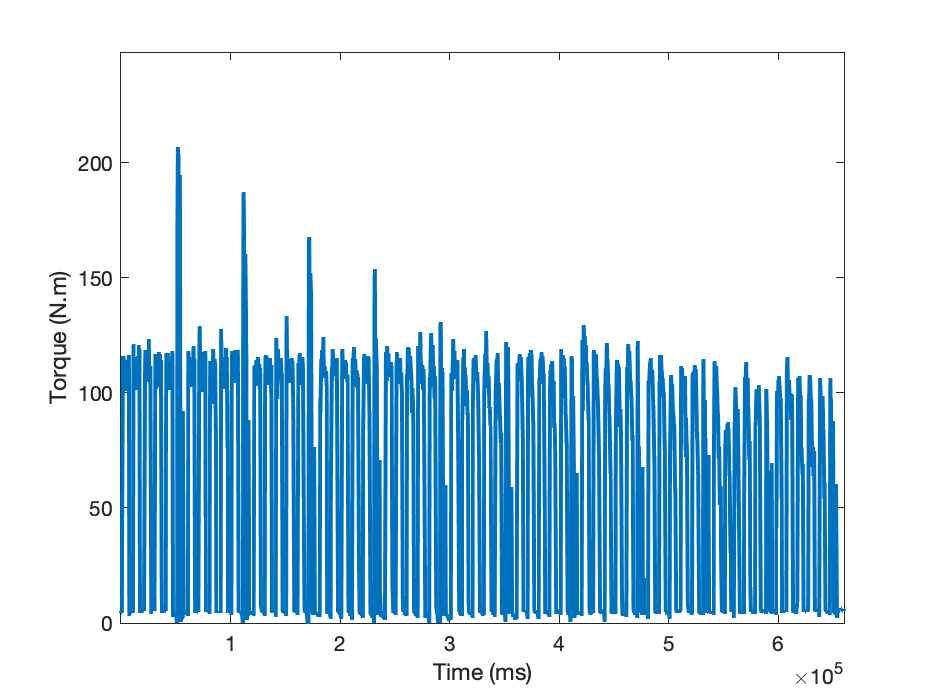}
        \caption{All sub-maximal contractions}\label{fig:allquad.png}
    \end{subfigure}
    \begin{subfigure}[bt]{0.48\textwidth}
        \centering
        \includegraphics[width=\textwidth]{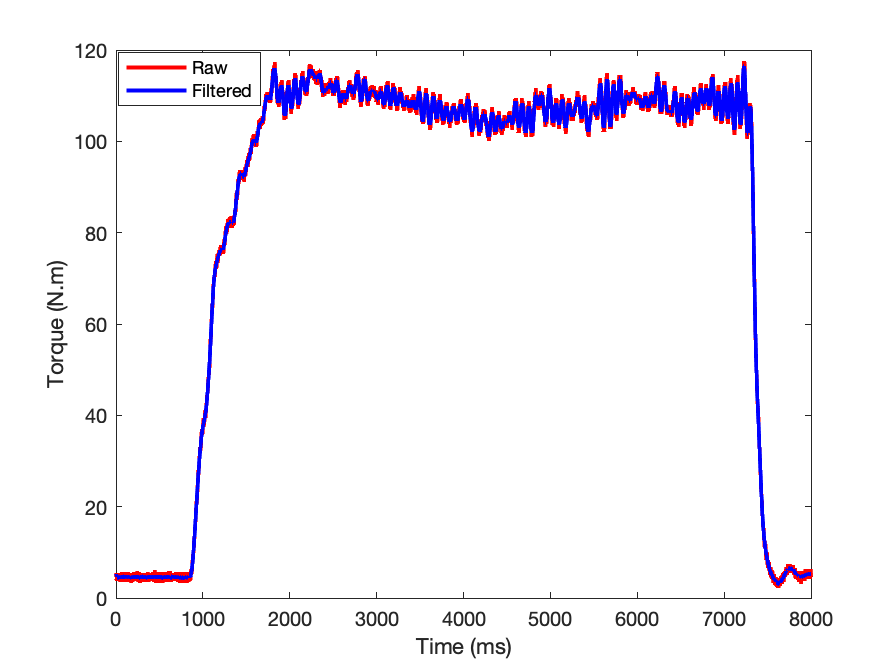}
        \caption{One sub-maximal contraction lasting ~5s}\label{fig:OneContraction.png}
    \end{subfigure}
    \caption{Muscle Contraction Data}\label{fig:Muscle_Contraction_Data}
\end{figure}
\textcolor{darkblue}{Figure}~\ref{fig:allquad.png} shows that the observations have lots of noise data. We need to extract valuable observations. We first cut the small periods into five pieces; each piece contains about 10000 observations (depending on the situation). For each piece, see Figure~\textcolor{darkblue}{Figure}~\ref{fig:OneContraction.png}, we extract 5000 consecutive observations in which the moving variance is minimum. For details, let \(a_i, i=1,\ldots,10000\) represent the 10000 observations, \(A_i=(a_{i},\ldots,a_{i+4999})\) denotes the consecutive 5000 observations, then the minimum moving variance of \(A_{s}\) can be found at \( s=\mathop{\arg\min}_{\substack{i}}(\text{Var}(A_i), i=1,\ldots,5001) \).
Furthermore, we also use Butterworth method~\citep{butterworthTheoryFilterAmplifiers} to filter the time series before extraction.
\textcolor{darkblue}{Figure}~\ref{fig:extractions.png} shows 52 extractions after using Butterworth Filter.~\textcolor{darkblue}{Figure}~\ref{fig:rlen_real_data.png} shows the result of change-point detection, the two change-points are 16 and 22.
\begin{figure}[ht]
    \centering
    \begin{subfigure}[bt]{0.48\textwidth}
        \centering
        \includegraphics[width=\textwidth]{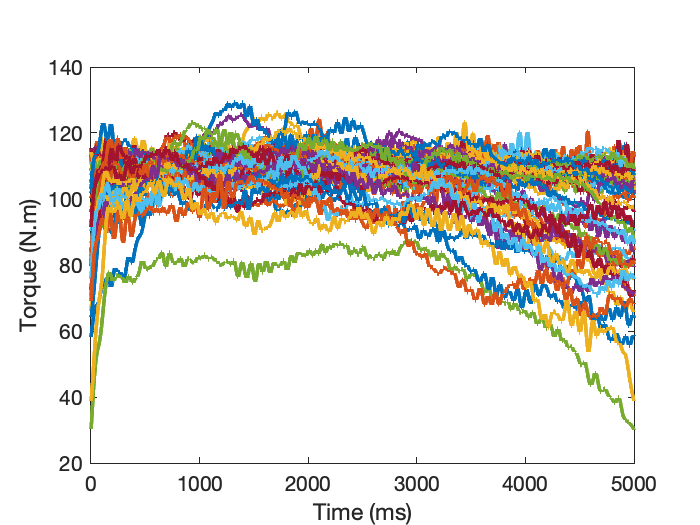}
        \caption{All Extractions}\label{fig:extractions.png}
    \end{subfigure}
    \begin{subfigure}[bt]{0.48\textwidth}
        \centering
        \includegraphics[width=\textwidth]{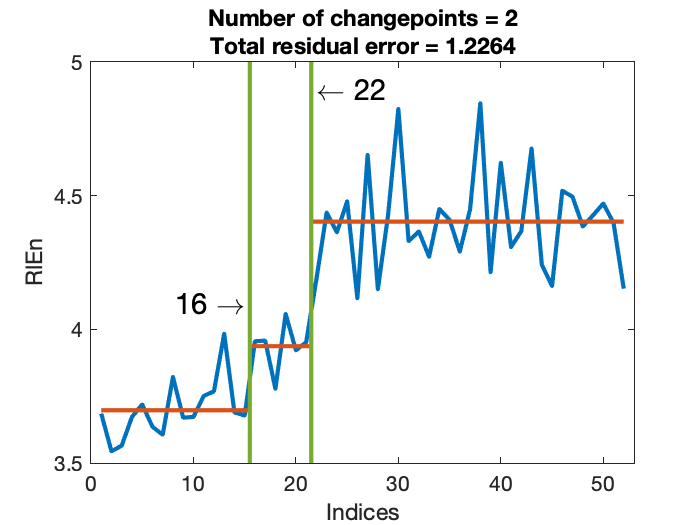}
        \caption{\(\mathfrak{E}\)}\label{fig:rlen_real_data.png}
    \end{subfigure}
    \caption{The Results of Extractions and Change-point Detection}\label{fig:Result}
\end{figure}

\begin{figure}[ht]
    \centering
    \begin{subfigure}[bt]{0.48\textwidth}
        \centering
        \includegraphics[width=\textwidth]{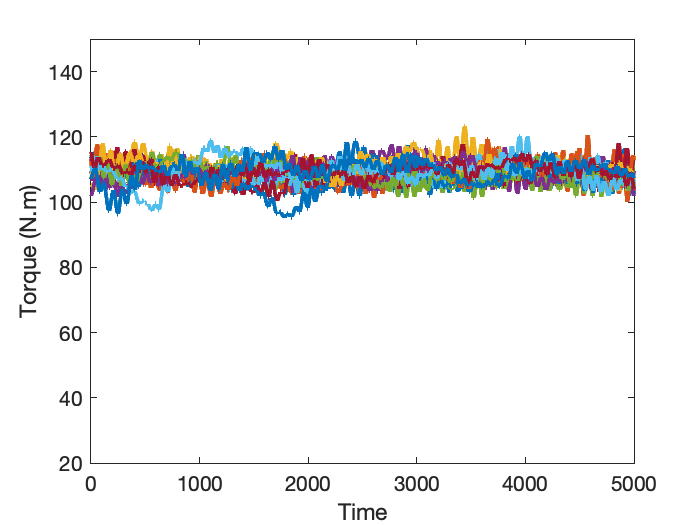}
        \caption{Group 1}\label{fig:rlen_three_groups1.png}
    \end{subfigure}
    \begin{subfigure}[bt]{0.48\textwidth}
        \centering
        \includegraphics[width=\textwidth]{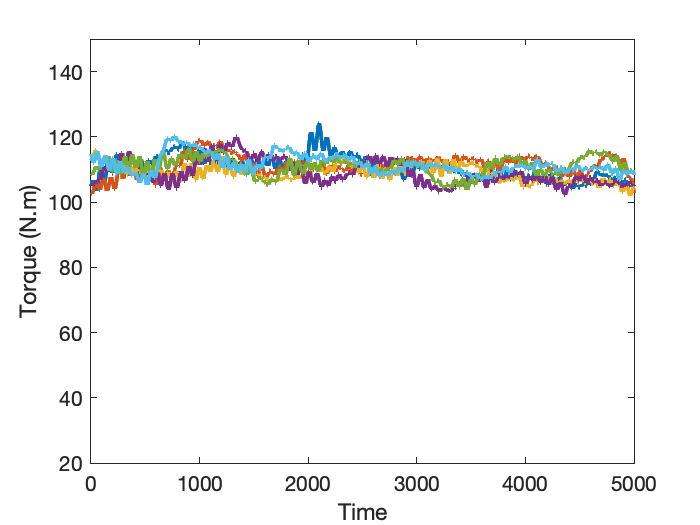}
        \caption{Group 2}\label{fig:rlen_three_groups2.png}
    \end{subfigure}
    \begin{subfigure}[bt]{0.48\textwidth}
        \centering
        \includegraphics[width=\textwidth]{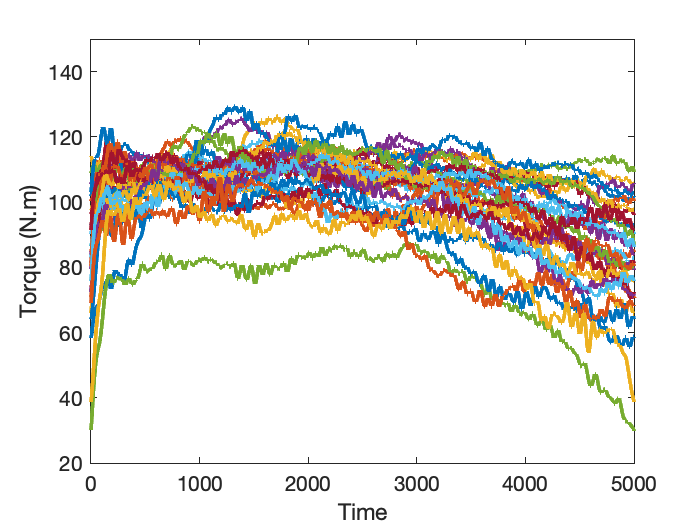}
        \caption{Group 3}\label{fig:rlen_three_groups3.png}
    \end{subfigure}
    \begin{subfigure}[bt]{0.48\textwidth}
        \centering
        \includegraphics[width=\textwidth]{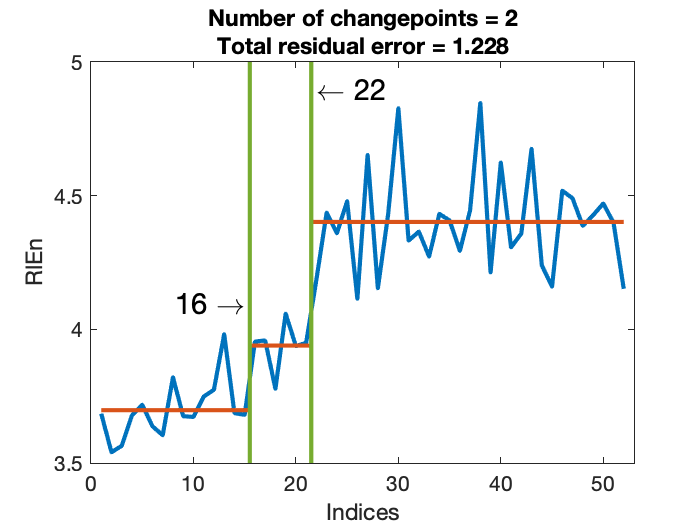}
        \caption{}\label{fig:rlen_change_pointsSim.png}
    \end{subfigure}

    \caption{The Divided Groups and Change-point Detection Result for the New Simulation Dataset}\label{fig:Group_Division_according_to_rlen}
\end{figure}

To verify the performance of \(\mathfrak{E}\), we further divide the time series into three groups based on \textcolor{darkblue}{Figure}~\ref{fig:rlen_real_data.png}, namely, Group 1, 2 and 3.  We obtain a seasonal ARIMA\( (p,d,q) \) process for each group. Again, 52 new time series are generated from the latest seasonal ARIMA processes. Then we regard them as observations and apply our \(\mathfrak{E}\) to these new observations to check whether our approach can detect the change-points correctly.

First, we need to estimate three seasonal ARIMA processes. For simplicity, let \( L \) be the lag operator notation, i.e., \( L^{i}x_{t}=x_{t-i} \). We found that this sports dataset is more complex than we expected; the degree of integration for the three groups are 2, 2 and 2, respectively, according to the Augmented Dickey-Fuller test. The real sport dataset contains seasonal effects and seasonal differences for the three groups as well, so it is a better choice to build the seasonal ARIMA processes\footnote{\url{https://uk.mathworks.com/help/econ/seasonal-arima-sarima-model.html}}:
\begin{equation*}
    \phi(L)\Phi(L){(1-L)}^{D}{(1-L^{s})}^{D_{s}}x_{t}=c+\theta(L)\varepsilon_{t},
\end{equation*}
where \( \phi(L)=1-\phi_{1}L-\cdots-\phi_{p}L^{p} \) and \( \theta(L)=1+\theta_{1}L+\cdots+\theta_{q}L^{q} \) represent the AR and MA operator polynomials.   \( \Phi(L)=1-\Phi_{p_{1}}L^{p_{1}}-\Phi_{p_{2}}L^{p_{2}}-\cdots-\Phi_{p_{s}}L^{p_{s}} \) is seasonal auto-regressive operator polynomials. \( {(1-L^{s})}^{D_{s}} \) is the so-called Seasonal Difference factor; for more details of seasonal ARIMA, see Section 9.9 in~\citet[][]{hyndmanForecastingPrinciplesPractice2013a}.   The spectrum analysis of time series determines the order of \( \Phi(L) \).  We use Bayesian Information Criterion (BIC) to choose the order \( p \) and \( q \) in \( \phi(L) \) and \( \theta(L) \).

Based on the average time series of each group, we have three processes:
\begin{align*}
    \mathbf{Process\  1:\ } & D =2,\quad D_{s}=1,\quad \hat{p} = 2, \\
    & \hat{q}=2, \quad s = 75, \quad c = 2.9993\times 10^{-6}, \\
    & \hat{\phi}(L)= 1-1.9414L+0.693L^{2}, \\
    & \hat{\theta}(L)= 1+1.82984L+0.9931L^{2}, \\
    & \hat{\Phi}(L) = 1-0.02037L^{75},\  \hat{\sigma}^{2} = 2\times 10^{-7}.
\end{align*}
\begin{align*}
    \mathbf{Process\ 2:\ }
    & D =2,\quad D_{s}=1,\quad \hat{p} = 2, \\
    & \hat{q}=2, \quad s = 67 , \quad c = 2.1477\times 10^{-6}, \\
    & \hat{\phi}(L)=1-1.9631L+0.9851L^{2}, \\
    & \hat{\theta}(L)= 1+1.9619L+0.9910L^{2}, \\
    & \hat{\Phi}(L) = 1+0.2818L^{67},\  \hat{\sigma}^{2} = 2\times10^{-7}.
\end{align*}
\begin{align*}
    \mathbf{Process\  3:\ }
    & D =2,\quad D_{s}=1,\quad \hat{p} = 2, \\
    & \hat{q}=1, \quad s = 81, \quad c = 3.9159\times 10^{-7}, \\
    & \hat{\phi}(L)=1-1.9768L+0.98801L^{2}, \\
    & \hat{\Phi}(L) = 1-0.1474L^{81}, \\
    & \hat{\theta}(L)= 1+0.3421L, \  \hat{\sigma}^{2} = 2\times10^{-7}.
\end{align*}
The number of time series generated from Processes 1, 2 and 3 are 15, 6 and 31 respectively.~\textcolor{darkblue}{Figure}~\ref{fig:rlen_change_pointsSim.png} shows our method can detect the change-points precisely at 16 and 22.
\begin{figure}[ht]
    \centering

\end{figure}
\end{document}